%% file: main.tex
\documentclass{article}[11pt]
\usepackage{titling}
\usepackage{ amssymb }
\usepackage{amsfonts}
\usepackage{amsthm}
\usepackage{amsmath}
\usepackage{bbm}
\usepackage{setspace}
\usepackage[normalem]{ulem}

\usepackage{float}

\usepackage[utf8]{inputenc}

\def\showauthornotes{0}

\def\showkeys{0}
\def\showdraftbox{0}
\def\showcolorlinks{1}
\def\usemicrotype{1}
\def\showfixme{0}

\usepackage{csquotes}

\usepackage[
backend=biber,
maxcitenames=50,
maxbibnames=99,
style=alphabetic,
]{biblatex}
\usepackage{titlesec}
\title{Agreement theorems for high dimensional expanders in the low acceptance regime: the role of covers}
\author{Yotam Dikstein\thanks{Weizmann Institute of Science, ISRAEL. email: yotam.dikstein@weizmann.ac.il.} \;and Irit Dinur\thanks{Weizmann Institute of Science, ISRAEL. email: irit.dinur@weizmann.ac.il. Both authors are supported by Irit Dinur's ERC grant 772839, and ISF grant 2073/21.}}
\date{\today}

\input{macros.tex}

\begin{document}\clearpage\thispagestyle{empty}

\maketitle
\begin{abstract}
Let $X$ be a family of $k$-element subsets of $[n]$ and let $\sett{f_s:s\to\Sigma}{s\in X}$ be an ensemble of local functions, each defined over a subset $s\subset [n]$. Is there a global function $G:[n]\to\Sigma$ such that $f_s = G|_s$ for all $s\in X$ ? An agreement test is a randomized property tester for this question. 
One such test is the V-test, that chooses a random pair of sets $s_1,s_2\in X$ with prescribed intersection size and accepts if $f_{s_1},f_{s_2}$ agree on the elements in $s_1\cap s_2$. 

The low acceptance (or $1\%$) regime is concerned with the situation that the test succeeds with low but non-negligible probability $\agr (\set{f_s}) \geq \eps>0$. A ``classical'' low acceptance agreement theorem says  
   \begin{equation*} \tag{$*$}
    \agr (\set{f_s}) > \eps  \quad \Longrightarrow \quad \exists G:[n]\to\Sigma,\quad \Pr_s[f_s\overset{0.99}{\approx} G|_s]\geq\poly(\eps).
\end{equation*}
Such statements are motivated by PCP questions. The case $X=\binom{[n]}k$ is well-studied and known as ``direct product testing'', which is related to the parallel repetition theorem. Finding sparser families $X$ that satisfy $(*)$ is known as derandomized direct product testing. Prior to this work, the sparsest family satisfying $(*)$ had $|X|\approx  n^{25}$, and we show $X$ with $|X|\approx n^2$. 

We study the general behavior of high dimensional expanders with respect to agreement tests in the low acceptance regime. High dimensional expanders, even  very sparse ones with $|X|=O(n)$, are known to satisfy the high acceptance variant (where $\epsilon =1-o(1)$). It has been an open challenge to analyze the low acceptance regime. 
Surprisingly, topological covers of $X$ play an important role. We show that:
\begin{enumerate}
    \item If $X$ has no connected covers, then $(*)$ holds, provided that $X$ satisfies an additional expansion property, called swap cosystolic expansion.  
    \item If $X$ {\em has} a connected cover, then $(*)$  fails. 
    \item If $X$ has a connected cover (and swap-cosystolic-expansion), 
    we replace $(*)$ by a statement that takes covers into account:
    \begin{align*}\tag{$**$}
        \agr(\set{f_s})> \eps\Longrightarrow \quad &\exists \hbox{ cover }\rho:Y\twoheadrightarrow X,\hbox{ and }G:Y(0)\to\Sigma, \hbox{ such that }\\ 
   &\Pr_{{\tilde s\twoheadrightarrow s}}[f_s \overset{0.99}{\approx} G|_{\tilde s}] \geq\poly(\eps),
    \end{align*} where ${\tilde s\twoheadrightarrow s}$ means that $\rho(\tilde s)=s$. 
    \end{enumerate} 
    The property of swap-cosystolic-expansion holds for quotients of the Bruhat Tits buildings. As a corollary we derive $(*)$ for $X$ being a spherical building, yielding a derandomized family with $|X| \approx n^2$.
    We also derive $(**)$ for LSV complexes $X$, for which $|X|=O(n)$. 
\end{abstract}

\newpage
\tableofcontents
    
\newpage\pagestyle{plain}
\setcounter{page}{1}

\section{Introduction}
Let $G:[n]\to\Sigma$ be a function. We often specify $G$ as a length $n$ string over an alphabet $\Sigma$. Alternatively, $G$ can be specified by providing its restrictions to a pre-determined family of subsets of $[n]$. Namely, given a family of subsets $X= \set{s\subset [n]}$, the function $G$ is represented by a table with $|X|$ entries, where the $s$-th entry specifies the local function $G|_s$. Formally, we have  an ensemble of ``local functions'' $\sett{f_s:s\to\Sigma}{s\in X}$ such that $f_s = G|_s$. Such a representation has built-in redundancy which  potentially can be used for amplifying distances while providing local testability. This is why such encodings are prevalent within PCP reductions, where a global PCP proof is broken up into many possibly-overlapping pieces and each piece is further encoded by an inner PCP gadget (see below for more details and references). The PCP verifier is tasked, among other things, with testing whether the ensemble of pieces corresponds to a coherent global proof. 

The agreement testing question is as follows. Given an ensemble $\sett{f_s}{s\in X}$, test whether there is some global function $G:[n]\to\Sigma$ such that $f_s = G|_s$ for most $s$.
The PCP setting imposes a stringent requirement that the number of queries be as small as possible, with two queries being the golden standard, as dictated by applications to tight hardness of approximation. 
The natural two-query test is to select a pair of overlapping subsets $s_1,s_2$ according to a pre-specified distribution\footnote{Throughout the introduction we suppress the precise distribution of the two-query test. For concreteness, think of the distribution given by selecting two $k$-sets that intersect on some $k'$ elements, and for example $k'=\sqrt k$. The distributions we use are precisely defined in \pref{sec:agreement-tests}.} and to check whether $f_{s_1},f_{s_2}$ agree. Namely, accept iff $f_{s_1}(v) = f_{s_2}(v)$ for all $v\in s_1\cap s_2$. Such a test is called an \emph{agreement test} and we denote the success probability of the test by $\agr(\set{f_s})$. 

Where agreement tests are concerned, there are two main regimes of interest. The ``$99\%$'' or high acceptance regime which is natural in the world of property testing, and the ``$1\%$'' or low acceptance regime which is the regime of interest in the world of PCPs.
\begin{itemize}
    \item \textit{The ``$99\%$'' \textbf{H}igh \textbf{A}cceptance regime}. We are given an ensemble $\set{f_s}_{s\in X}$ that passes the agreement test with probability $1-\epsilon$, for small $\epsilon>0$. We want to conclude that there is a function $G:[n]\to\Sigma$ such that $\set{G|_s}_s \approx \set{f_s}_s$. Formally, an agreement testing theorem describes a distribution $\calD$ over pairs of subsets from $X$ for which
    \begin{equation}\label{eq:99p}\tag{$HA$}
        \agr (\set{f_s})> 1-\epsilon \quad \Longrightarrow  \quad \exists G:[n]\to\Sigma,\quad \prob{f_s = G|_s} \geq 1-O(\varepsilon) .
    \end{equation}
    \item \textit{The ``$1\%$'' \textbf{L}ow \textbf{A}cceptance regime}. Here we assume that the test succeeds with much smaller probability $\epsilon>0$ for a small positive constant $\eps$, say $\eps=0.01$, or even $\eps = o_k(1)$. The hope is that even this modest success probability implies that the ensemble $\set{f_s}$ has some global structure.
    The $1\%$ regime, also called the small-soundness regime, is especially important in PCP settings, where small $\eps$ translates to a larger gap between completeness and soundness of the PCP.  
    We say that $X$ satisfies \eqref{eq:LA} if for any ensemble $\set{f_s}_{s\in X}$, \eqref{eq:LA} holds,
    namely, for every ensemble $\set{f_s}_{s\in X}$, 
    \begin{equation}\label{eq:LA}\tag{$LA$}
    \agr (\set{f_s}) > \eps  \quad \Longrightarrow \quad \exists G:[n]\to\Sigma,\quad \Pr_s[f_s\overset{0.99}{\approx} G|_s]\geq\poly(\eps).
\end{equation}
In words, \eqref{eq:LA} says that $\epsilon$ agreement implies global structure. 
{Of course, to be completely formal we must equip \eqref{eq:LA} with parameters. This is omitted in the introduction and is done in \pref{def:distribution-soundness}.} 

Note that the conclusion allows approximate equality between $f_s$ and $G|_s$ (meaning that they are equal on almost\footnote{The notation \(f_s \overset{1-\zeta}{\approx} G|_s\) throughout the paper means that \(f(v)=G(v)\) on a \(1-\zeta\)-fraction of \(v\in s\).} all of $s$). This is unavoidable, see discussion in \cite{DinurG2008}. 
\end{itemize}

Proving that \eqref{eq:LA} holds is quite non-trivial even when $X=\binom{[n]}{k}$, where we think of $k$ as a sufficiently large constant while $n\to\infty$. This is known as the (symmetrized version of the) ``direct product testing question''. In this case the function $G:[n]\to\Sigma$ is represented by its restriction to all possible $k$-element subsets. The ensemble $\set{G|_s}_s$ is called the (symmetrized) direct product encoding of $G$. This encoding has been studied both in the context of hardness amplification (see, e.g. \cite{IJKW08}), and as a generalization of PCP low degree tests, as initiated by \cite{GolSaf97}. Following a sequence of works, \cite{GolSaf97,DinurR06, DinurG2008,ImpagliazzoKW2012, DinurS14, DinurL2017}, it is known, by now, that $\binom{[n]}k$ satisfies \eqref{eq:LA}. Moreover, the tradeoffs between $k,\epsilon$ and $1-\zeta$, as well as the significant difference between two and three queries, have all been relatively well studied, see further discussion below. 

Goldreich and Safra \cite{GolSaf97} and later Impagliazzo et. al \cite{ImpagliazzoKW2012} studied derandomized direct products. Namely, finding smaller families $X\subset \binom{[n]}k$, such that $|X| \ll n^k$ yet $X$ still satisfies \eqref{eq:99p} or \eqref{eq:LA}. The motivation for this question is that the number of subsets in $X$ controls the efficiency of the encoding. Impagliazzo et. al \cite{ImpagliazzoKW2012} showed that for every $\epsilon>0$ there exists $k$ and a sequence of families $\set{X_n}$ that satisfy \eqref{eq:LA}, where $|X_n|\leq O(n^{c})$ for $c\geq 25$ an absolute constant independent of $k,\epsilon$. It is natural to ask about the limits of such derandomization. Can the same hold true when $|X_n| = O(n)$? 

The recent ``discovery'' of bounded-degree high dimensional expanders by the theoretical computer science community has lead to a hope that these might provide a path towards such results.
In \cite{DinurK2017} (and later improved in \cite{DiksteinD2019}) it was shown that if $X$ is a spectral high dimensional expander (see \pref{sec:preliminaries} for precise definitions), then it supports a $99\%$ agreement theorem as in \eqref{eq:99p}. This gave the first sequence $\set{X_n}$ of linear-size families satisfying \eqref{eq:99p}. 

In this paper we address the same question for \eqref{eq:LA}. Are there families $X_n$ that satisfy \eqref{eq:LA} and whose size is linear in $n$ ? This has been mentioned as an open question in \cite{DinurK2017}. We conjecture the following: 
\begin{conjecture}\label{conj:1}
    For every $\epsilon>0$ there exists an integer $k(\epsilon)$, such that for all $k>k(\epsilon)$, there is a sequence $\set{X_n}_{n\to\infty}$ of families of $k$-element subsets of $n$ that satisfy \eqref{eq:LA} and such that $|X_n| \leq O_{k,\epsilon}(n)$.
\end{conjecture}
It turns out that as far as \eqref{eq:LA} is concerned, the story is more complicated. 
Our first result is that spectral high dimensional expansion is not enough.  
We show that many complexes $X$ that satisfy \eqref{eq:99p}, fail to satisfy \eqref{eq:LA} in quite a strong way. In fact, as long as a complex has a connected topological cover, it fails to satisfy \eqref{eq:LA}. We remark that many known sparse high dimensional expanders fall under this category as the constructions of both \cite{LubotzkySV2005a} and \cite{KaufmanO181} give a sequence of high dimensional expanders $\set{X^{(i)}}_{i\to \infty}$ where $X^{(i)}$ is covered by $X^{(i+1)}$. 
\begin{lemma}\torestate{\label{lem:cex}
    Let 
    \(\zeta < \frac{1}{2}\), $k\in \mathbb{N}$, and let \(\lambda \leq \frac{0.99}{k}\). Let \(X\) be a $k$-dimensional \(\lambda\)-two-sided high dimensional expander, and assume $X$ has a connected \(\ell\)-cover for some $\ell>1$ (e.g. $\ell=2$). Then there exists an ensemble of functions \(\mathcal{F} = \sett{f_r:r \to \set{0,1}}{r \in X(k)}\) such that \(Agree(\mathcal{F}) \geq \frac{1}{\ell}\), and yet for every \(G:X(0) \to \set{0,1}\) it holds that 
    \[\Prob[r \in X(k)]{f_r \overset{1-\zeta}{\approx} G|_r} \leq\exp(-\Omega_\zeta(k)).\]}
\end{lemma}
(We are using standard notation where $X(0)$ are the vertices of $X$, and $X(k)$ are the $k$-faces of $X$). This lemma reveals connected $2$-covers (and more generally connected $\ell$-covers) to be an obstruction to $1\%$ agreement theorems.  The notion of a connected cover is a basic notion in topology, generalizing the more familiar lift of a graph to higher dimensions. In a nutshell, a complex $Y$ is a $\ell$-cover of a complex $X$ whenever there is a $\ell$-to-$1$ homomorphism $\rho: Y \twoheadrightarrow X$ such that the preimage of every face of $X$ is $\ell$ pairwise disjoint faces in $Y$,  (see \pref{sec:preliminaries} for details). 

\pref{lem:cex} shows that high dimensional expanders $X$ that have connected covers necessarily fail to satisfy \eqref{eq:LA}. In contrast, previous works \cite{DinurG2008, ImpagliazzoKW2012, DinurS14, DinurL2017} show that the complete complex $X = \binom{[n]}k$, which is known to be a superb high dimensional expander, does satisfy \eqref{eq:LA}. This seeming contradiction is resolved by observing that the complete complex \emph{has no connected covers}. This suggests that sparser high dimensional expanders that are cover-free might satisfy \eqref{eq:LA}, as we show further below in \pref{cor:spherical-building-list-agreement}.
\\

Our main result is concerned with the more general case. We introduce a new variant of (LA) that {\em is} satisfied by high dimensional expanders, as long as we extend the definition of a global function to allow covers.  We essentially show that the obstruction uncovered in \pref{lem:cex} is the only possible obstruction. Our theorem has an additional condition, however, which is that $X$ is also required to have swap-coboundary-expansion (see \pref{def:swapce}). In a companion paper we prove that many familiar complexes, including spherical buildings as well as the Ramanujan complexes of \cite{LubotzkySV2005b}, have this property (see \pref{thm:DD-faces} below). 

\begin{theorem}[Main Theorem, see \pref{thm:main} for a formal statement]\label{thm:abstract-general} Let $k\in \mathbb{N}$, and let $\eps >\Omega(1/\log k)$. Let $d>k$ be sufficiently large and let $X$ be a $d$-dimensional high dimensional expander with sufficiently good swap-cosystolic-expansion. 
    Let $\sett{f_s:s\to\Sigma}{s\in X(k)}$ be an ensemble of local functions on $X(k)$. 
\begin{equation}\label{eq:CLA}\tag{$CLA$}
    \agr (\set{f_s}) > \eps  \quad \Longrightarrow \quad \exists Y\xrightarrowdbl{\rho} X, \exists G:Y(0)\to\Sigma,\quad \Pr_s[f_s\hbox{ is explained by }G]\geq\poly(\eps).
\end{equation} 
where $\rho:Y\to X$ is a $\ell=\poly(1/\epsilon)$ covering map. 
\end{theorem}
Here the phrase ``$f_s$ is explained by $G$'' informally means that $G|_{s'}\approx f_s$ for some $s'\in Y$ that covers $s$. More formally, the covering map $\rho$ gives a bijection from $s'\in Y(k)$ to $s\in X(k)$, and by $G|_{s'}\approx f_s$ we mean that for (almost) every $v\in s'$, $G(v) = f_s(\rho(v))$. 

This theorem should be viewed as a new type of structure theorem, replacing \eqref{eq:LA} with its cover version, \eqref{eq:CLA}: Every $\epsilon$-agreeing ensemble on $X$ comes from a global function $G$, perhaps not on the vertices of $X$, but on the vertices of a cover $Y$ of $X$, which is not much larger than $X$.

One might wonder whether any progress has been made by replacing the quest for complexes that satisfy $1\%$ agreement theorems with complexes having swap-coboundary-expansion. In a companion paper \cite{DiksteinD2023b}
we show that there are important families of complexes that have swap-coboundary-expansion.

\begin{theorem}[Main theorem in \cite{DiksteinD2023b}]\label{thm:DD-faces}
    Let $X$ be a \(d\)-dimensional simplicial whose links are spherical buildings for a sufficiently large prime \(q\). Then for every \(d_1 \leq d\), \(X\) is a \(\exp(-\sqrt{d_1})\)-swap cosystolic expander. 
\end{theorem}

The bounded-degree high dimensional expanders of \cite{LubotzkySV2005a, LubotzkySV2005b} satisfy the requirement of \pref{thm:DD-faces} and therefore also of \pref{thm:abstract-general}. Thus, we get the following lift-decoding theorem.

\begin{corollary}\label{cor:general} Let $k\in \mathbb{N}$, and let $\eps >\Omega( 1/\log k)$. Let $d>k$ be sufficiently large and let $X$ be a $d$-dimensional $\lambda=\frac{1}{d^2}$ high dimensional expander, whose vertex links are spherical buildings.

For any ensemble $\set{f_s}_{s\in X(k)}$ that satisfies 
   $\agr(\set{f_s})> \eps$, there must exist a $\poly(1/\eps)$-cover $\rho:Y\twoheadrightarrow X$, and a global function $G:Y(0)\to\Sigma$, such that 
   \[\Pr_{s}[f_s\hbox{ is explained by }G]\geq\poly(\eps).\]
\end{corollary}
As a special case of this corollary, we get a new $1\%$ agreement theorem for spherical buildings. These are relatively sparse simplicial complexes that are known to have no connected covers. The only possible $\ell$-cover of such a complex $X$ consists of $\ell$ disjoint isomorphic copies of $X$. In this case, the global function on $Y(0) = X(0)\times [\ell]$ can be interpreted as a list of $\ell$ global functions on $X(0)$.
\begin{corollary}[Agreement for spherical buildings]\label{cor:spherical-building-list-agreement}
    Let $\epsilon>0$, and let $k>\exp(\poly(1/\epsilon))$ be an integer. Let $d>k$ be sufficiently large and let $X$ be a $d$-dimensional spherical building that is a $\frac{1}{d^2}$ high dimensional expander. Then $X$ satisfies \eqref{eq:LA}.
\end{corollary}
This puts spherical buildings as the current best known derandomized direct product encoding that satisfies \eqref{eq:LA}, improving the size of the family from $|X_n| \leq O(n^{25})$ in \cite{ImpagliazzoKW2012} to $|X_n| = O(n^2)$.

Finally, we revisit \pref{conj:1}. There is still some hope for a positive answer via high dimensional expanders. Even when $X$ has connected covers,  there is a quantitative aspect connecting the agreement parameter $\eps$ and the size $\ell$ of the cover, which may allow us to move from \eqref{eq:CLA} to \eqref{eq:LA}. 
I.e., if $X$ has the additional property that all of its $\ell$-covers for $\ell < \poly(1/\eps)$ are necessarily disconnected, then we immediately get conclusion a la \eqref{eq:LA}. 
This motivates the study and construction of complexes that are free of connected $\ell$-covers for \emph{small} $\ell$.  A construction of such a family of complexes would give a linear-size family of complexes supporting \eqref{eq:LA}. %
\subsection*{Covers and Agreement} 
Let us describe how topological covers give rise to the counter-example driving \pref{lem:cex}. 
\begin{enumerate}
    \item Let $X$ be a simplicial complex, and let $Y\xrightarrowdbl{\rho} X$ be a $\ell$-cover with $\ell=2$. Explicitly, $Y$ is a simplicial complex and we have a $2$-to-$1$ map $\rho:Y(0)\to X(0)$ such that the image of each $Y$-face is a $X$-face, and the preimage of each $X$-face is two disjoint $Y$-faces. 
    
    \item Fix a global function $H:Y(0)\to\bits$ with the property that every pair of vertices  $(v,0),(v,1)\in Y(0)$ that map to the same $v\in X(0)$, are given distinct values $H(v,0)\neq H(v,1)$. 

    \item Every face $s\in X(k)$ is covered by two faces $s_1, s_2 \in Y(k)$, i.e. \(\rho({s}_i)=s\).
    Thus, we naturally get two possible assignments to $f_s$ that biject down from either $H|_{s_1}$ or $H|_{s_2}$. By construction, these two assignments are negations of one another. 

    \begin{figure}
    \centering
    \includegraphics[scale=0.4]{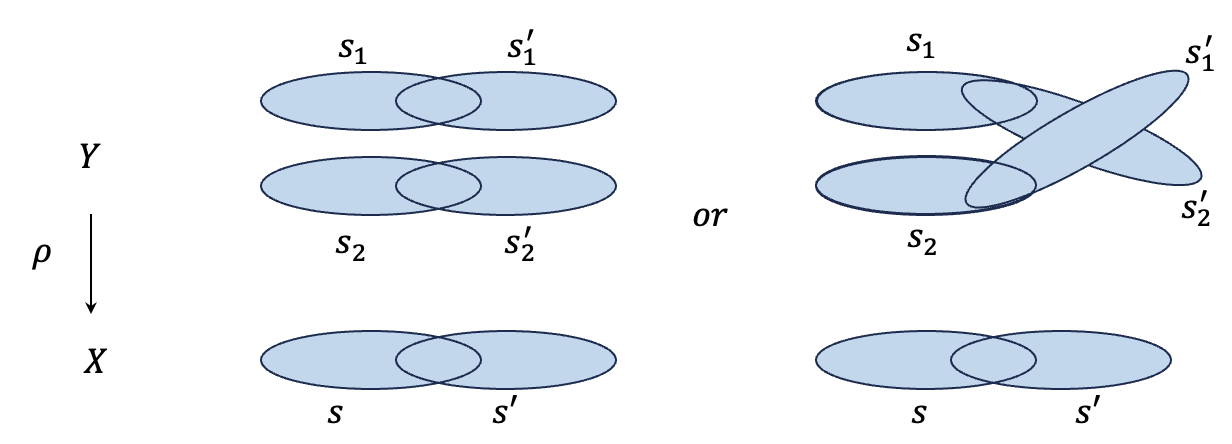}
    \caption{Possible ways two intersecting sets $s,s'\in X$ are covered by $Y$.}
    \label{fig:cover-intro}
\end{figure}
    Moreover, suppose, as in Figure \ref{fig:cover-intro},  that $s,s'\in X(k)$ intersect, and let $s_1,s_2\in Y(k)$ cover $s$ and $s'_1,s'_2\in Y(k)$ cover $s'$. Then $s_1$ intersects exactly one of $s'_1,s'_2$, and similarly for $s_2$. Namely, the intersection pattern gives a bijection between $\set{s_1,s_2}$ and $\set{s'_1,s'_2}$.
    \item Our final ensemble for $X$ is obtained by choosing, for each $s$ independently at random, one of the two possible assignments. For any fixed pair of intersecting sets $s,s'\in X(k)$, the probability that $f_{s}$ agrees with $f_{s'}$ is exactly $1/2$. Thus, in total, the average agreement probability of the entire ensemble $\set{f_s}$ is about $1/2$. Moreover, sampling arguments (relying on the high dimensional expansion of $X,Y$) will show that no global function $G:X(0)\to\bits$ can agree with more than $\exp(-\Omega(k))$ of the subsets.
\end{enumerate}
What this construction shows is that whenever a complex has a connected cover, it cannot satisfy \eqref{eq:LA}, but it could still possibly satisfy \eqref{eq:CLA}.
More details can be found in \pref{sec:proof-of-counterexample}.

\subsection*{Proof of main theorem - Overview}

We next turn to describing the technical ingredients that go into the proof of our main theorem,  \pref{thm:abstract-general}, which essentially reverses the construction described just above. 
Like many proofs in the world of high dimensional expansion, our proof has a local-to-global structure. We essentially show how to ``lift'' agreement theorems from the complete complex to high dimensional expanders. 

An agreement test on the complete complex is a distribution over $k$-element subsets of $[n]$. Let us focus for concreteness on the \(V\)-test distribution from \cite{ImpagliazzoKW2012}, in which two $k$-sets are sampled uniformly, conditioned on their intersection size being $\sqrt k$. 

Moving to a high dimensional expander $X$, the corresponding distribution $\mathcal{D}_X$ is obtained by selecting a face of size $2k-\sqrt k$, and inside it two faces $r_1,r_2$ whose sizes are $k$ and such that they intersect on a set of size $\sqrt k$. Our proof goes through the following steps.

\subsubsection*{Step 1 - Constructing a local list for every face}
Fix \(k \ll d_1 \ll d\). 
Here we treat every \(s \in X(d_1)\) as a ground set whose set of $k$-element subsets form a copy of the complete complex on $d_1+1$ vertices. We describe \(\mathcal{D}_X\) in an equivalent local way:
\begin{enumerate}
    \item Choose $s\in X(d_1)$ at random.
    \item Choose $r_1,r_2\subset s$ according to the original test distribution $\mathcal{D}$ on the complete complex on \(|s|\)-vertices.
\end{enumerate} 
Given an ensemble with $\agr(\set{f_s})>\eps$, a standard sampling argument will show that for a typical face $s\in X(d_1)$ the success of the test conditioned on the two faces being contained in $s$, is still about $\eps$ for nearly all sets $s\in X(d_1)$. So for each of these local ensembles, we can use the previous agreement theorem from \cite{ImpagliazzoKW2012}. We obtain a global function $g_s$ for each $s$ that agrees with about $\epsilon$ of the local functions in \(\set{f_r}_{r \subseteq s}\).

An important complication is that there could be several equally valid choices for $g_s$, each supported on a different $\epsilon$ fraction of the local functions $\set{f_r}_{r\subset s}$. Therefore, we must take into account an entire \textbf{list} of functions for $s$, which we denote by $\bar L(s)$.

Similarly to all $s\in X(d_1)$, we find such lists for faces $t \in X(2d_1+1)$ and $u \in X(3d_1+2)$.

\subsubsection*{Step 2 - Showing these lists match each other}
Next, we carefully modify the previous argument so that the obtained lists will be compatible with each other. More precisely, we show that one can take lists such that there is some integer \(\ell\) so that \(\abs{\overline{L}(s)}=\ell\) for nearly all \(s\). Moreover, denoting \(\overline{L}(s)=\set{L_s^1,L_s^2,\dots,L_s^\ell}\), for most pairs $s \subseteq t$, there is a permutation \(\pi_{s,t}:[\ell] \to [\ell]\) where \(\pi_{s,t}(i)=j\) if and only if \(L_s^i \approx L_t^j|_s\).

We remark that these two steps are similar to \cite{DinurHKLT2018}, where lists of assignments were derived for every face. However, while it is easy to see that a noticeable fraction of the lists match (which was what \cite{DinurHKLT2018} showed), this is not enough for us, and we must show that they match with probability $99\%$ for the next steps to work. This causes significant technical complications but turns out to be possible.

\subsubsection*{Step 3 - Unique games on the \emph{faces} complex}Before stating the next steps, we shift our point of view from \(X\) to its \emph{faces complex}.
\begin{definition}[The Faces Complex] \torestate{\label{def:face-complex-intro}
        Let \(X\) be a \(d\)-dimensional simplicial complex. Let \(d_1 \leq d\). We denote by \(\FX{d_1}X\) the simplicial complex whose vertices are \(\FX{d_1}X(0)=X(d_1)\) and whose faces are all \(\sett{\set{s_0,s_1,...,s_j}}{s_0\dunion s_1 \dunion \dots \dunion s_j \in X}\). }
\end{definition}  

The $1$-skeleton of $\FX{d_1}X$ is a graph whose vertices are the $d_1$-sets. Each vertex $s_1 \in \FX{d_1}X$ has a list of assignments, and (almost) every edge between $s_1$ and $s_2$ carries a bijection between the assignments of $s_1$ and the assignments of $s_2$. The bijection is just the composition \(\psi(s_1,s_2) = \pi_{t,s_2}^{-1} \circ \pi_{t,s_1}\) where \(t=s_1 \dunion s_2\). This can be viewed as a unique games instance over $\FX{d_1}X$. We show that for most triangles $s_1 s_2 s_3\in \FX{d_1}X$, there are no local contradictions. Namely, if we denote by $\psi_{s_i,s_j}$ the permutation from the list of $s_i$ to the list of $s_j$, then 
    \begin{equation} \label{eq:no-local-contradictions}
        \psi(s_2, s_3)\circ \psi(s_1, s_2)=\psi(s_1, s_3)
    \end{equation}
nearly always. In topological language this collection of bijections gives us a $1$-cochain in $C^1(\FX{d_1},Sym(\ell))$ where $Sym(\ell)$ is the groups of permutations over $\ell$ elements. The consistency implies that this $1$-cochain is nearly a cocycle.

\subsubsection*{Step 4 - Swap cosystolic expansion}
At this point, we wish to tweak the collection of bijections that are nearly consistent, and create a new collection that is consistent everywhere. This is possible if and only if the faces complex is a cosystolic expander (and this is precisely our definition of swap cosystolic expansion, see \pref{def:swapce}). 

It remains an open question to characterize complexes with sufficiently strong swap cosystolic expansion. We don't know if cosystolic expansion implies swap cosystolic expansion. Nevertheless, as stated above, we prove in a companion paper \cite{DiksteinD2023b} that every high dimensional expander whose links are spherical buildings, has \(\FX{d_1}X\) has swap-cosystolic expansion (\pref{thm:DD-faces}).  

\subsubsection*{Step 5 - Constructing the cover}
It is well known that any collection of permutations \(\set{\pi_{uv}}_{uv \in \FX{d_1}X(1)}\) with \emph{no} local contradictions\footnote{I.e. \eqref{eq:no-local-contradictions} holds for every triangle.} corresponds to a cover \(\rho_{\FX{d_1}X}:\widetilde{\FX{d_1}X} \to \FX{d_1}X\) (see e.g. \cite{Surowski1984}). Let us describe the underlying graph of the cover. Its vertices are pairs \((s,i)\) for every \(s \in \FX{d_1}X(0)\) and \(i \in [\ell]\). We connect \((s_1,i)\) to \((s_2,j)\) by an edge in \(\widetilde{\FX{d_1}X}\) if \(s_1 s_2 \in \FX{d_1}X(1)\) and \(\pi_{s_1,s_2}(i)=j\). The fact that \eqref{eq:no-local-contradictions} holds for every triangle is the reason that this underlying graph is indeed a graph skeleton of a cover \(\rho_{\FX{d_1}X}:\widetilde{\FX{d_1}X} \to \FX{d_1}X\) (see \pref{def:induced-cover} for the formal construction).
    
The cover we get from this step is a cover of the faces complex \(\FX{d_1}X\). The next step is to construct a cover $Y$ of $X$ from the cover \(\widetilde{\FX{d_1}X}\) of \(\FX{d_1}X\). Moreover, the complex $Y$ is such that its faces complex is isomorphic to \(\widetilde{\FX{d_1}X}\), \(\iota:\FX{d_1}Y \overset{\sim}{\to} \widetilde{\FX{d_1}X}\). This complex \(Y\) is the cover on which we define the global function. 

The use of cosystolic expansion for constructing covers from permutations with \emph{few local contradictions} was developed in \cite{DinurM2019}. Here we instantiate this idea, using the permutations defined by the first two steps.

\subsubsection*{Step 6 - From a cover to a global function via \(99\%\)-agreement}
We are nearly done, but we still need to construct the global function \(G:Y(0) \to \Sigma\). We define an ensemble of functions \(\set{h_{\tilde{s}}}_{\tilde{s} \in Y(d_1)}\), where \(h_{\tilde{s}}\) colors the vertices according to \(L_s^i \in \overline{L}(s)\), such that \(\iota(\tilde{s})=(s,i)\). Opening up the definition of the isomorphism \(\iota\) and the permutations \(\pi_{s_1,t}\) above, we show that there is a agreement distribution \(\mathcal{D}_Y\) where \[\agr_{\mathcal{D}_Y} (\set{h_{\tilde{s}}}) \geq 99\%.\] 
Then all that remains is to use a known \(99\%\)-agreement theorem a la \eqref{eq:99p} to obtain a global function \(G:Y(0) \to \Sigma\) that agrees with most \(h_{\tilde{s}}\). This is (finally!) the global function that explains the original distribution \(\set{f_r}\). 

Note that the quantitative cosystolic expansion needed in \pref{thm:abstract-general} is rather modest; it is allowed to decay to zero exponentially in \(d_1/k\). The local spectral expansion required is inverse polynomial in the dimension \(d\).

\subsection*{Related work}
\subsubsection*{PCPs and agreement tests} Agreement tests in the low acceptance ``$1\%$'' regime appear in several PCP constructions, such as in the plane-versus-plane \cite{RaSa} and line-versus-line low degree tests \cite{ArSu}. Analysis in the low acceptance regime translates to a PCP with perfect completeness and small soundness $\epsilon$. The fact that the test makes only two queries is essential for getting a two-query PCP such as label cover, which is an essential feature for downstream applications for hardness of approximation. 

The most popular large-gap PCP (also known as label cover) used in hardness of approximation reductions is the based on the parallel repetition theorem \cite{Raz-parrep}. Parallel repetition is essentially the direct product encoding, and the analysis is highly related to analysis of agreement tests on $X=\binom{[n]}k$, as can be seen both in \cite{ImpagliazzoKW2012} who give a black-box conversion from an agreement-test theorem to a parallel repetition theorem, see also \cite{DinurS14, DinurS14-parrep}.

A major motivation for studying derandomized agreement tests is towards construction of more length-efficient PCPs that still have a large gap between the completeness and soundness cases. The length parameter is important for various reasons ranging from practical implementations of PCPs to theoretical implications for hardness of approximation. An agreement test can often be transformed into a PCP construction, although such a transformation is not automatic. Indeed, Impagliazzo et. al  \cite{ImpagliazzoKW2012} constructed a derandomized family $X$ that satisfies \eqref{eq:LA}, but it took more work \cite{DinurM11} to transform it into an length-efficient PCP theorem with a large gap. Moshkovitz and Raz \cite{MR10} constructed a label cover PCP with nearly-linear length $n \cdot 2^{(\log n)^{0.99}}$, motivated by applications to hardness of approximation (efficient-length reductions give a much stronger connection between exponential hardness assumptions on exact and approximation problems). Getting the length further down (possibly to $n\cdot\poly\log n$) is an open question. 
Agreement tests for highly derandomized  families $X$ may lead to progress in this direction. 
\\

More recently agreement tests appear in the proof of the 2:2 theorem \cite{DKKMS,KMS}, where $X$ is the collection of all $\ell$-dimensional subspaces of a vector space (the Grassmannian), and the agreement test compares two subspaces that overlap on an $\ell-1$ dimensional subspace. This test constitutes the 2:2 inner verifier. It was shown \cite{BarakKS19} that in this setting agreement follows from small set expansion of the underlying complex, which was then proven in \cite{KMS}. Interestingly, the 2:2 agreement test cannot satisfy \eqref{eq:LA}, but rather a weakened version of it, that takes into account ``zoom-ins'' and ``zoom-outs''. 

Here as well as in other low degree tests the local functions $f_s$ are not arbitrary but rather restricted to having additional structure, namely being linear or having low degree. This gives the theorems a slightly different flavor (for example, it is much easier to get separation between list elements). 

\subsubsection*{Bogdanov's counter-example}
The gap amplification proof of the PCP theorem \cite{Dinur07} is based on repeated applications of an agreement test in the high acceptance regime. Once the soundness has reached some constant $1-\zeta$, it is natural to wonder if one (or a few) additional steps can lead to arbitrarily small soundness $\epsilon>0$. Bogdanov shows \cite{Bogdanov-comment05} an example of a PCP instance where the soundness cannot drop below $1/2$. We revisit his example in \pref{sec:bog} and show that it is a special case of our \pref{lem:cex}. This suggests, perhaps, high dimensional expanders as candidates for gap amplification with soundness below $1/2$.

\subsubsection*{Parameters of Agreement tests}
Direct product tests, or agreement tests for the complete complex $\binom{[n]}{k}$, have been relatively well studied. For two queries, $\eqref{eq:LA}$ holds for $\eps = \poly(1/k)$, and fails when $\epsilon\ll 1/k$ \cite{DinurG2008}. For three queries, $\eqref{eq:LA}$ holds for $\eps$ even exponentially small in $k$ \cite{ImpagliazzoKW2012,DinurS14, DinurL2017}, and beyond that it necessarily fails. The parameters we get for our main theorem are significantly weaker, with $\eps = \Omega(1/\log k)$; and further investigation of this tradeoff is left for future work. 

\subsubsection*{List-agreement and cosystolic expansion}
The idea behind the last step in our proof of \pref{thm:abstract-general} is close in spirit to the list decoding argument in \cite{GotlibK2022}. Their work builds upon \cite{DinurM2019} to list-decode ensembles of functions that have a \(99\%\)-regime list-agreement guarantee, using a closely related notion of coboundary expansion. They derive a complex similar (but not identical) to the faces complex, which they call the representation complex. Then they use their list-agreement test to construct permutations that have almost no local contradictions, and use coboundary expansion to transform this to a list of global functions on \(X\). Unfortunately, we could not have adapted their framework to our purposes; the notion of \(99\%\)-regime list-agreement they require is too strong to hold in our case. In addition, coboundary expansion does not capture complexes that are not simply connected, and we require a method to deal with such complexes.

\subsubsection*{Comparison with \cite{BafnaM2023}}
Independent work by Bafna and Minzer \cite{BafnaM2023} also gave a theorem characterizing low acceptance agreement in high dimensional expanders. 

In their work they give a necessary and sufficient condition for a complex satisfying \eqref{eq:LA}, that is similar to the swap-cosystolic expansion condition. Let us explain the condition by explaining the ``sufficient'' part of their proof. The first part of their proof resembles steps \(1-3\) of ours. Both works construct lists of functions on faces (using previous works on agreement tests) together with permutations on edges of the faces complex. 

The two proofs diverge in step four. Roughly, in their work, the condition they impose on \(X\) is that the faces complex is a swap-coboundary expander with respect to permutations \emph{that also come with lists of functions as in step \(3\)} (see \cite[Definitions 1.7,1.8]{BafnaM2023}). Let us call this condition $(\star)$. Using $(\star)$, they are able to construct local functions with \(99\%\) agreement on \(X\). Then they patch them up to a global function in a way similar to our step \(6\).

In our proof, step \(4\) uses swap-cosystolic expansion to do almost the same thing - but instead of constructing local functions on \(X\) we do so on a cover \(Y\) of \(X\). In the special case where \(X\) is a swap-{\em coboundary} expander (which is strictly stronger than being a swap-{\em cosystolic} expander), the only possible covers of \(X\) are disjoint copies of \(X\). That is, in the swap-coboundary expander case, our proof gives the same construction as theirs and shows \eqref{eq:LA}. We also mention, that in retrospect, our proof works {\em as is}, if we replace our assumptions by $(\star)$.

Our proof applies to a more general situation. It is common for a high dimensional expander to have covers (and not be a swap-coboundary expander) especially as known constructions of sparse high dimensional expanders are attained via families of complexes that cover one another. In this case, we manage to derive a structural characterization that is necessarily weaker than \eqref{eq:LA} yet still encapsulates the global structure hiding behind functions that pass the agreement test, i.e. \eqref{eq:CLA}.

In addition, the work of \cite{BafnaM2023} shows that $(\star)$ is a necessary condition for satisfying \eqref{eq:LA}. They show that if an ensemble fails $(\star)$, it must give rise to a counter example. %

Thus, their condition $(\star)$ is a {\em characterization} of complexes satisfying \eqref{eq:LA}. In contrast, our condition of swap-cosystolic expansion is not equivalent, and even swap-coboundary expansion seems to be stronger than \eqref{eq:LA}, so in any case not a characterization. Nevertheless, our condition applies to well-known constructions, allowing us to derive concrete corollaries (\pref{cor:general} and \pref{cor:spherical-building-list-agreement}) from our abstract theorem (for these derivations we rely on \pref{thm:DD-faces} from the companion paper). 

In terms of concrete examples of complexes that fail \eqref{eq:LA}, both works point to the same issue. Their work also observes that the \cite{LubotzkySV2005b} complexes sometimes fail to satisfy \eqref{eq:LA}. The counterexample they give is in fact the same counterexample we give in \pref{sec:proof-of-counterexample} - they use a cocycle that comes from a cover to produce the counterexample.

We remark that the work of \cite{BafnaM2023} raises an interesting open question: Does swap-coboundary expansion follow from their (a priori) weaker variant of $(\star)$?

\subsection*{Acknowledgements}
We thank Roy Meshulam and Amitay Kamber for insightful early discussions on the topic.

\input{sections/preliminaries}

\input{sections/theorem-proof}
\input{sections/list-ag-2}
\input{sections/cover-construction}
\input{sections/agreement-on-glob-func}
\input{sections/counterexample}
\printbibliography

\end{document}

%% file: macros.tex
\providecommand{\NN}{\mathbb{N}}

\newcommand{\calD}{{\mathcal{D}}}
\newcommand{\calF}{{\mathcal{F}}}

\newcommand{\calG}{{\mathcal{G}}}

\usepackage[l2tabu, orthodox]{nag}

\usepackage{xspace,enumerate}

\usepackage[dvipsnames]{xcolor}
\usepackage{tikz-cd}

\usepackage[T1]{fontenc}
\usepackage[full]{textcomp}

\usepackage[american]{babel}

\usepackage{mathtools}

\usepackage{amsthm}
\usepackage{amsmath}

\newtheorem{theorem}{Theorem}[section]
\newtheorem*{theorem*}{Theorem}

\newtheorem{proposition}[theorem]{Proposition}
\newtheorem*{proposition*}{Proposition}
\newtheorem{lemma}[theorem]{Lemma}
\newtheorem*{lemma*}{Lemma}
\newtheorem{corollary}[theorem]{Corollary}
\newtheorem*{corollary*}{Corollary}
\newtheorem*{conjecture*}{Conjecture}

\newtheorem*{fact*}{Fact}

\newtheorem*{hypothesis*}{Hypothesis}
\newtheorem{conjecture}[theorem]{Conjecture}

\theoremstyle{definition}
\newtheorem{definition}[theorem]{Definition}
\newtheorem*{definition*}{Definition}

\newtheorem{example}[theorem]{Example}

\theoremstyle{remark}
\newtheorem{claim}[theorem]{Claim}
\newtheorem*{claim*}{Claim}
\newtheorem{remark}[theorem]{Remark}
\newtheorem*{remark*}{Remark}

\newtheorem*{observation*}{Observation}

 \usepackage[letterpaper,
 top=1in,
 bottom=1in,
 left=1in,
 right=1in]{geometry}

\usepackage[varg]{pxfonts} 

\usepackage[sc,osf]{mathpazo}
\usepackage{lmodern}

\ifnum\showkeys=1
\usepackage[color]{showkeys}
\fi

\ifnum\showcolorlinks=1
\usepackage[
colorlinks=true,
urlcolor=blue,
linkcolor=blue,
citecolor=OliveGreen,
]{hyperref}
\fi

\ifnum\showcolorlinks=0
\usepackage[
colorlinks=false,
pdfborder={0 0 0}
]{hyperref}
\fi

\usepackage{prettyref}

\newcommand{\savehyperref}[2]{\texorpdfstring{\hyperref[#1]{#2}}{#2}}

\newrefformat{eq}{\savehyperref{#1}{\textup{(\ref*{#1})}}}
\newrefformat{lem}{\savehyperref{#1}{Lemma~\ref*{#1}}}
\newrefformat{def}{\savehyperref{#1}{Definition~\ref*{#1}}}
\newrefformat{thm}{\savehyperref{#1}{Theorem~\ref*{#1}}}
\newrefformat{cor}{\savehyperref{#1}{Corollary~\ref*{#1}}}
\newrefformat{cha}{\savehyperref{#1}{Chapter~\ref*{#1}}}
\newrefformat{sec}{\savehyperref{#1}{Section~\ref*{#1}}}
\newrefformat{app}{\savehyperref{#1}{Appendix~\ref*{#1}}}
\newrefformat{tab}{\savehyperref{#1}{Table~\ref*{#1}}}
\newrefformat{fig}{\savehyperref{#1}{Figure~\ref*{#1}}}
\newrefformat{hyp}{\savehyperref{#1}{Hypothesis~\ref*{#1}}}
\newrefformat{alg}{\savehyperref{#1}{Algorithm~\ref*{#1}}}
\newrefformat{rem}{\savehyperref{#1}{Remark~\ref*{#1}}}
\newrefformat{item}{\savehyperref{#1}{Item~\ref*{#1}}}
\newrefformat{step}{\savehyperref{#1}{step~\ref*{#1}}}
\newrefformat{conj}{\savehyperref{#1}{Conjecture~\ref*{#1}}}
\newrefformat{fact}{\savehyperref{#1}{Fact~\ref*{#1}}}
\newrefformat{prop}{\savehyperref{#1}{Proposition~\ref*{#1}}}
\newrefformat{prob}{\savehyperref{#1}{Problem~\ref*{#1}}}
\newrefformat{claim}{\savehyperref{#1}{Claim~\ref*{#1}}}
\newrefformat{relax}{\savehyperref{#1}{Relaxation~\ref*{#1}}}
\newrefformat{red}{\savehyperref{#1}{Reduction~\ref*{#1}}}
\newrefformat{part}{\savehyperref{#1}{Part~\ref*{#1}}}
\newrefformat{ex}{\savehyperref{#1}{Example~\ref*{#1}}}
\newrefformat{para}{\savehyperref{#1}{Paragraph~\ref*{#1}}}
\newrefformat{ass}{\savehyperref{#1}{Assumption~\ref*{#1}}}
\newrefformat{question}{\savehyperref{#1}{Question~\ref*{#1}}}

\newcommand{\Sref}[1]{\hyperref[#1]{\S\ref*{#1}}}

\usepackage{nicefrac}

\ifnum\usemicrotype=1
\usepackage{microtype}
\fi

\ifnum\showauthornotes=1
\newcommand{\Authornote}[2]{{\sffamily\small\color{red}{[#1: #2]}}}
\newcommand{\Authornotecolored}[3]{{\sffamily\small\color{#1}{[#2: #3]}}}
\newcommand{\Authorcomment}[2]{{\sffamily\small\color{gray}{[#1: #2]}}}
\newcommand{\Authorstartcomment}[1]{\sffamily\small\color{gray}[#1: }

\newcommand{\Authorfnote}[2]{\footnote{\color{red}{#1: #2}}}
\newcommand{\Authorfixme}[1]{\Authornote{#1}{\textbf{??}}}
\newcommand{\Authormarginmark}[1]{\marginpar{\textcolor{red}{\fbox{\Large #1:!}}}}
\else
\newcommand{\Authornote}[2]{}
\newcommand{\Authornotecolored}[3]{}
\newcommand{\Authorcomment}[2]{}
\newcommand{\Authorstartcomment}[1]{}

\newcommand{\Authorfnote}[2]{}
\newcommand{\Authorfixme}[1]{}
\newcommand{\Authormarginmark}[1]{}
\fi

\ifnum\showfixme=0

\fi

\usepackage{boxedminipage}

\newcommand{\brac}[1]{[#1]}
\newcommand{\Brac}[1]{\left[#1\right]}

\newcommand{\abs}[1]{\lvert#1\rvert}
\newcommand{\Abs}[1]{\left\lvert#1\right\rvert}

\newcommand{\card}[1]{\lvert#1\rvert}

\newcommand\sett[2]{\left\{ #1 \left| \; \vphantom{#1 #2} \right. #2  \right\}}
\newcommand{\set}[1]{\{#1\}}

\newcommand{\F}{\mathbb{F}}
\newcommand{\Esymb}{\mathbb{E}}
\newcommand{\Psymb}{\mathbb{P}}

\DeclareMathOperator*{\E}{\Esymb}

\DeclareMathOperator*{\ProbOp}{\Psymb}

\renewcommand{\Pr}{\ProbOp}
\def\one{{\mathbf{1}}}
\def\wt{{\mathrm{wt}}}

\newcommand{\prob}[1]{\Pr \left[ {#1} \right] }
\newcommand{\Prob}[2][]{\Pr_{{#1}}\left[#2\right]} 
\newcommand{\cProb}[3]{\Pr_{{#1}}\left[ #2 \left| \; \vphantom{#2 #3} \right. #3  \right]} 

\newcommand{\ex}[1]{\E\brac{#1}}
\newcommand{\Ex}[2][]{\E_{{#1}}\Brac{#2}}

\newcommand{\ve}{\;\hbox{and}\;}

 \usepackage{dsfont}
\usepackage{mathrsfs}

\newcommand{\textparen}[1]{\text{(#1)}}

\ifx\because\undefined
\newcommand{\because}[1]{\textparen{because #1}}
\else
\renewcommand{\because}[1]{\textparen{because #1}}
\fi

\newcommand{\bits}{\{0,1\}}

\newcommand\bdot\bullet

\DeclareMathOperator{\poly}{poly}

\DeclareMathOperator{\supp}{supp}
\DeclareMathOperator{\dist}{dist}

\renewcommand{\leq}{\leqslant}

\renewcommand{\geq}{\geqslant}

\ifnum\showdraftbox=1

\else

\fi

\let\epsilon=\varepsilon

\numberwithin{equation}{section}

\newcommand{\MYstore}[2]{%
  \global\expandafter \def \csname MYMEMORY #1 \endcsname{#2}%
}

\newcommand{\MYload}[1]{%
  \csname MYMEMORY #1 \endcsname%
}

\newcommand{\MYnewlabel}[1]{%
  \newcommand\MYcurrentlabel{#1}%
  \MYoldlabel{#1}%
}

\newcommand{\MYdummylabel}[1]{}

\newcommand{\torestate}[1]{%
  \let\MYoldlabel\label%
  \let\label\MYnewlabel%
  #1%
  \MYstore{\MYcurrentlabel}{#1}%
  \let\label\MYoldlabel%
}

\newcommand{\restatetheorem}[1]{%
  \let\MYoldlabel\label
  \let\label\MYdummylabel
  \begin{theorem*}[Restatement of \prettyref{#1}]
    \MYload{#1}
  \end{theorem*}
  \let\label\MYoldlabel
}

\newcommand{\restatelemma}[1]{%
  \let\MYoldlabel\label
  \let\label\MYdummylabel
  \begin{lemma*}[Restatement of \prettyref{#1}]
    \MYload{#1}
  \end{lemma*}
  \let\label\MYoldlabel
}

\newcommand{\restateprop}[1]{%
  \let\MYoldlabel\label
  \let\label\MYdummylabel
  \begin{proposition*}[Restatement of \prettyref{#1}]
    \MYload{#1}
  \end{proposition*}
  \let\label\MYoldlabel
}

\newcommand{\restateclaim}[1]{%
  \let\MYoldlabel\label
  \let\label\MYdummylabel
  \begin{claim*}[Restatement of \prettyref{#1}]
    \MYload{#1}
  \end{claim*}
  \let\label\MYoldlabel
}

\newcommand{\restatecorollary}[1]{%
  \let\MYoldlabel\label
  \let\label\MYdummylabel
  \begin{corollary*}[Restatement of \prettyref{#1}]
    \MYload{#1}
  \end{corollary*}
  \let\label\MYoldlabel
}

\newcommand{\restatefact}[1]{%
  \let\MYoldlabel\label
  \let\label\MYdummylabel
  \begin{fact*}[Restatement of \prettyref{#1}]
    \MYload{#1}
  \end{fact*}
  \let\label\MYoldlabel
}

\newcommand{\restatedefinition}[1]{
\let\MYoldlabel\label 
\let\label\MYdummylabel 
\begin{definition*}[Restatement of \prettyref{#1}] 
    \MYload{#1} 
\end{definition*} 
\let\label\MYoldlabel 
} 

\newcommand{\restate}[1]{%
  \let\MYoldlabel\label
  \let\label\MYdummylabel
  \MYload{#1}
  \let\label\MYoldlabel
}

\newcommand{\eps}{\epsilon}

\newcommand{\xrightarrowdbl}[2][]{%
  \xrightarrow[#1]{#2}\mathrel{\mkern-14mu}\rightarrow
}

\let\origparagraph\paragraph
\renewcommand{\paragraph}[1]{\origparagraph{#1.}}

\allowdisplaybreaks

\newcommand{\dunion}{\mathbin{\mathaccent\cdot\cup}}

\let\pref=\prettyref

\newcommand{\dir}[1]{\overset{\to}{#1}}
\newcommand{\bproof}[1]{\begin{proof}[Proof of \pref{#1}]}
\newcommand{\eproof}{\end{proof}}
\newcommand{\coboundary}{{\delta}}

\newcommand{\agr}{{\textrm {Agree}}}

\newcommand{\FX}[1]{\mathrm{F}^{#1}\!}

%% file: sections/preliminaries.tex
\section{Preliminaries} \label{sec:preliminaries}
\subsection{High dimensional expanders}
Most of the definitions in this subsection are standard, with the exception of the definition of the non-lazy up down walk \pref{def:non-lazy-up-down}.

A pure \(d\)-dimensional simplicial complex \(X\) is a hypergraph that consists of an arbitrary collection of sets of size \((d+1)\) together with all their subsets. The sets of size \(i+1\) in \(X\) are denoted by \(X(i)\). The vertices of \(X\) are denoted by \(X(0)\) (we identify between a vertex \(v\) and its singleton \(\set{v}\)). We will sometimes omit set brackets and write for example \(uvw\in X(2)\) instead of \(\set{u,v,w}\in X(2)\). As a convention \(X(-1) = \set{\emptyset}\).
Let \(X\) be a pure \(d\)-dimensional simplicial complex. Let \(k \leq d\). We denote the set of oriented \(k\)-faces in \(X\) by \(\dir{X}(k) = \sett{(v_0,v_1,...,v_k)}{\set{v_0,v_1,...,v_k} \in X(k)}\).

For \(k \leq d\) we denote by \(X^{\leq k} = \bigcup_{j=-1}^k X(j)\) the \(k\)-skeleton of \(X\). When \(k=1\) we call this complex \emph{the underlying graph of \(X\)}, since it consists of the vertices and edges in \(X\) (as well as the empty face).

A \emph{clique complex} is a simplicial complex such that if \(s \subseteq X(0)\) has that if \(s\) is a clique, that is, for every two vertices \(v,u \in s\) the edge \(vu \in X(1)\), then \(s \in X\).

A \((d+1)\)-partite \(d\)-dimensional simplicial complex is a generalizeation of a bipartite graph. It is a complex \(X\) such that one can decompose \(X(0) = A_0 \dunion A_1 \dunion \dots \dunion A_d\) such that for every \(s \in X(d)\) and \(i \in \set{0,\ldots,d}\) it holds that \(\abs{s \cap A_i} = 1\).

\subsubsection*{Probability over simplicial complexes}
Let \(X\) be a simplicial complex and let \(\Pr_d:X(d)\to (0,1]\) be a density function on \(X(d)\) (that is, \(\sum_{s \in X(d)}\Pr_d(s)=1\)). This density function induces densities on lower level faces \(\Pr_k:X(k)\to (0,1]\) by \(\Pr_k(t) = \frac{1}{\binom{d+1}{k+1}}\sum_{s \in X(d),s \supset t} \Pr_d(s)\). We can also define a probability over directed faces, where we choose an ordering uniformly at random. Namely, for \(s\in \dir{X}(k)\), \(\Pr_k(s) = \frac{1}{(k+1)!}\Pr_k(set(s))\) (where \(set(s)\) is the set of vertices participating in \(s\)). When clear from the context, we omit the level of the faces, and just write \(\Pr[T]\) or \(\Prob[t \in X(k)]{T}\) for a set \(T \subseteq X(k)\).

\subsubsection*{Links in a high dimensional expander} Let \(X\) be a \(d\)-dimensional simplicial complex and let \(s \in X\) be a face. The link of \(s\) is the \(d'=d-|s|\)-dimensional complex
\[X_s = \sett{t \setminus s}{t \in X, t \supseteq s}.\]
For a simplicial complex \(X\) with a measure \(\Pr_d:X(d) \to (0,1]\), the induced measure on \(\Pr_{d',X_s}:X_s(d-|s|)\to (0,1]\) is
\[\mathbb{P}_{d',X_s}(t \setminus s) = \frac{\Pr_d(t)}{\sum_{t' \supseteq s} \Pr_d(t')}.\]

We denote by \(\lambda(X_s)\) the (normalized) second largest eigenvalue of the adjacency operator of \(X_s^{\leq 1}\). We denote by \(\abs{\lambda}(X_s)\) to be the (normalized) second largest (in absolute value) eigenvalue of the adjacency operator of \(X_s^{\leq 1}\).

\begin{definition}[High dimensional expander]
    Let \(X\) be a \(d\)-dimensional simplicial complex and let \(\lambda \in (0,1)\). We say that \(X\) is a \emph{\(\lambda\)-one sided high dimensional expander} if for every \(s \in X^{\leq d-2}\) it holds that \(\lambda(X_s) \leq \lambda\). We say that \(X\) is a \emph{\(\lambda\)-two sided high dimensional expander} if for every \(s \in X^{\leq d-2}\) it holds that \(\abs{\lambda}(X_s) \leq \lambda\).
\end{definition}
We stress that this definition includes \(s= \emptyset\), which also implies that \(X^{\leq 1}\) itself should have a small second largest eigenvalue.

\subsubsection*{Walks on high dimensional expanders}
Let \(X\) be a \(d\)-dimensional simplicial complex. Let \(\ell \leq k \leq d\). The \((k,\ell)\)-containment graph \(G_{k,\ell}=G_{k,\ell}(X)\) is the bipartite graph whose vertices are \(L=X(k), R=X(\ell)\) and whose edges are all \((t,s)\) such that \(t \supseteq s\). The probability of choosing such an edge is by choosing \(t \sim X(k)\) and a uniform \(s \subseteq t\) of size \(\ell+1\).

\begin{theorem}[\cite{KaufmanO2020}] \label{thm:eignevalues-of-walk}
    Let \(X\) be a \(d\)-dimensional \(\lambda\)-one sided high dimensional expander. Let \(\ell \leq k \leq d\). Then the second largest eigenvalue of \(G_{k,\ell}(X)\) is upper bounded by \(\lambda(G_{k,\ell}(X)) \leq \frac{\ell+1}{k+1} + O(k\lambda)\).
\end{theorem}

A corollary proven in \cite{DinurK2017} is that this graph is also a good sampler.
\begin{corollary}[\cite{DinurK2017}] \label{cor:good-sampler-from-expansion}
    Let \(A \subseteq X(k)\), \(\ell \leq k\) and \(\zeta > 0\). Let
    \[B(A) = \sett{s \in X(\ell)}{\Abs{\Prob[t\supseteq s]{A} - \prob{A}} > \zeta}.\]
    Then \(\prob{B(A)} = O(\frac{\ell+1}{(k+1) \zeta^2}\prob{A})\).
\end{corollary}

A related walk is the swap walk. Let \(k,\ell,d\) be integers such that \(\ell+k\leq d-1\). The \(k,\ell\)-swap walk \(S_{k,\ell}=S_{k,\ell}(X)\) is the bipartite graph whose vertices are \(L=X(k), R=X(\ell)\) and whose edges are all \((t,s)\) such that \(t \dunion s \in X\) (and in particular \(t,s\) are disjoint). The probability of choosing such an edge is the probability of choosing \(u \in X(k+\ell+1)\) and then uniformly at random partitioning it to \(u=t\dunion s\). This walk has been defined and studied independently by \cite{DiksteinD2019} and by \cite{AlevFT2019}, who bounded its spectral expansion.

\begin{theorem}[\cite{DiksteinD2019,AlevFT2019}]
    Let \(X\) be a \(\lambda\)-two sided high dimensional expander. Then the second largest eigenvalue of \(S_{k,\ell}(X)\) is upper bounded by \(\lambda(S_{k,\ell}(X)) \leq (k+1)(\ell+1)\lambda\).
\end{theorem}

We also define another random walk call the non-lazy up down walk.
\begin{definition}[Non lazy up down walk] \label{def:non-lazy-up-down}
    Let \(X\) be a \(d\)-dimensional simplicial complex. Let \(d_1\leq d_2 \leq d\) such that \(2d_1 \geq d_2-1\). The \((d_1,d_2)\)-non lazy up down walk is the distribution \((s_1,s_2) \sim DU_{n}\) of \(X(d_1)\) where \(s_1,s_2\) are chosen by first choosing \(t \in X(d_2)\) and then uniformly sampling \(s_1,s_2 \in X(d_1)\) such that \(s_1 \cup s_2 = t\).
\end{definition}
The reason this is called non-lazy, is because we do not allow \(s_1 \cup s_2 \subsetneq t\) as in the usual up-down walk defined in \cite{DiksteinDFH2018} for example. We note that \(\Abs{s_1 \cap s_2} = 2d_1-d_2+1\), i.e. the intersection size doesn't depend on the pair chosen. We can decompose this walk also as sampling \(r= s_1 \cap s_2 \in X(2d_1-d_2)\), then sampling \(p_1,p_2 \in X_r(d_2-d_1)\) according to the swap walk, and then outputting \(s_1 = r \cup p_1, s_2 = r \cup p_2\).

\subsubsection*{The spherical building}
Let \(d \in \NN\) and \(q\) be a prime power. 
\begin{definition}\label{def:sph-bldg}
    The spherical building (sometimes called the \(SL_d(\mathbb{F}_q)\)-spherical building), is the complex \(X\) whose vertices are all non-trivial linear subspaces of \(\mathbb{F}_q^d\). Its higher dimensional faces are all flags \(\sett{W_0 \subseteq W_1 \subseteq \dots \subseteq W_m}{W_0,W_1,\dots,W_m \subseteq \mathbb{F}_q^d}\). 
\end{definition}
This complex is \((d-2)\)-dimensional.

\begin{claim}[\cite{EvraK2016}] \label{claim:spherical-building-hdxness}
    Let \(X\) be a \(SL_d(\mathbb{F}_q)\)-spherical building. Then \(X\) is a \(O(\frac{1}{\sqrt{q}})\)-one sided high dimensional expander. Moreover, \(X^{\leq k}\) is a \(\max \set{O(\frac{1}{\sqrt{q}}), \frac{1}{d-k}}\)-two sided high dimensional expander.
\end{claim}

\subsection{Agreement tests} \label{sec:agreement-tests}

Let \(k<d\) and let \(X\) be a \(d\)-dimensional simplicial complex. Let \(\Sigma\) be some fixed alphabet and suppose we have an ensemble of functions \(\mathcal{F} = \sett{f_r:r\to \Sigma}{r \in X(k)}\).%

An two-query agreement test is a distribution \(\mathcal{D}\) over pairs \(r_1,r_2 \in X(k)\). The agreement of an ensemble is
\begin{equation}
    \agr_{\mathcal{D}}(\mathcal{F}) = \Prob[r_1,r_2\sim \mathcal{D}]{f_{r_1} = f_{r_2}}.
\end{equation}
When we write \(f_{r_1} = f_{r_2}\) we mean that \(f_{r_1}(v)=f_{r_2}(v)\) for every \(v \in r_1 \cap r_2\).

More generally, a \(q\)-ary agreement test is a distribution of \(r_1,r_2,...,r_q \in X(k)\), where the agreement of the ensemble
\begin{equation}
    \agr_{\mathcal{D}}(\mathcal{F}) = \Prob[r_1,r_2,\dots r_q \sim \mathcal{D}]{\forall i,j \; f_{r_i} = f_{r_j}}.
\end{equation}

Let \(\Delta_k(d)\) be the \(k\)-dimensional complete complex over \(d\) vertices. Let \(\mathcal{D}\) be a \(q\)-ary agreement test on \(\Delta_k(d)\) and assume that it is symmetric\footnote{i.e. for every permutation \(\pi:[d] \to [d]\) it holds that \(\prob{s_1,s_2,...,s_q} = \prob{\pi(s_1),\pi(s_2),...,\pi(s_q)}.\)}. Let \(X\) be another \(d\)-dimensional simplicial complex. We define the \emph{extension} \(\mathcal{D}_X\) of \(\mathcal{D}\) to an agreement test on \(X\), as follows:
\begin{enumerate}
    \item Sample \(t \in X(d)\).
    \item Query \(r_1,r_2,...,r_q \subseteq t\) according to \(\Delta_k(d)\).
\end{enumerate}
We note that by the symmetry of \(\mathcal{D}\) the second step doesn't depend on the way we identify the vertices of \(t\) with the vertices of \(\Delta_k(d)\). Let us give two examples for such tests, that were considered in previous works. %
\begin{definition}[Two-query \(V\)-test]
    Let \(d = 2k-\sqrt{k+1}+1\).
    \begin{enumerate}
        \item Sample some \(t \in X(d)\).
        \item Sample uniformly \(s_1,s_2 \in X(k)\) such that \(s_1, s_2 \subseteq t\), conditioned on \(\abs{s_1 \cap s_2} = \sqrt{k+1}\).
    \end{enumerate}
\end{definition}

\begin{definition}[Three-query \(Z\)-test]
    Let \(d=3k-2\sqrt{k+1}+2)\).
    \begin{enumerate}
        \item Sample some \(t \in X(d)\).
        \item Sample three \(s_1,s_2,s_3 \in X(k)\) such that \(s_1, s_2, s_3 \subseteq t\), conditioned on \(\abs{s_1 \cap s_2}, \abs{s_2 \cap s_3} = \sqrt{k+1}\) and \(s_1 \cap s_3 = \emptyset\).
    \end{enumerate}
\end{definition}

A sound distribution is a distribution that supports an agreement theorem.
\begin{definition} \label{def:distribution-soundness}
    Let \(X\) be a simplicial complex and let \(\mathcal{D}\) be an agreement distribution on \(X\). Let \(\eta,\varepsilon_0,e > 0\) be constants. We say that \(\mathcal{D}\) is \emph{\((\eta,\varepsilon_0,e)\)-sound} if for every ensemble of functions \(\mathcal{F}\) such that \(\agr_{\mathcal{D}}(\mathcal{F}) = \varepsilon \geq \varepsilon_0\), there exists a function \(G:X(0)\to \Sigma\), such that
    \begin{equation} \label{eq:dist-soundness}
        \Prob[r_1,r_2,\dots,r_q \sim \mathcal{D}]{\forall j \; G|_{r_j} \overset{1-\eta}{\approx} f_{r_j} \ve \forall i,j \; f_{r_i} = f_{r_j}} \geq \frac{1}{2}\varepsilon^e.    
    \end{equation}
    
\end{definition}
Here \(f \overset{1-\eta}{\approx} g\) means that \(f,g\) differ on at most a \(\eta\)-fraction of their coordinates, or stated differently \(\dist(f,g) \leq \eta\).

We extend this definition to soundness with respect to covers in \pref{def:distribution-cover-soundness}, after we introduce simplicial covers formally.

Here are some examples of such distributions.
\begin{example} \label{ex:sound-distributions}
    \begin{enumerate}
    \item Dinur and Goldenberg showed that the \(V\)-test extended to \(X=\Delta_k(n)\) is \((\sqrt{k},k^{-c},1)\)-sound for \(d \geq k^3\)  and \(c > 0\) \cite{DinurG2008}\footnote{In fact, \cite{DinurG2008} considered a \(V\)-test with a smaller intersection size, but the same result hold for \(\sqrt{k}\) too. See \cite{ImpagliazzoKW2012} for a proof for the intersection size of \(\sqrt{k}\).}.
    \item The \(1\%\)-agreement theorem by Impagliazzo, Kabanets and Wigderson, showed that the \(Z\)-test extended to \(X=\Delta_k(n)\) together with the techniques used in \cite[Theorem 5.1]{DinurG2008} show that the \(Z\)-test is \((k^{-0.2},exp(-\Omega(k^{1/2}),O(1))\)-sound.
    \item Dinur and Livni-Navon, together with the techniques used in \cite[Theorem 5.1]{DinurG2008} show that the \(Z\)-test is \((\lambda,exp(-\Omega(k)),O(1))\)-sound for every constant \(\lambda > 0\) \cite{DinurL2017}.
\end{enumerate}
\end{example}

Here \(n \gg k\) in all cases. We remark that the exact distributions in the works in the example were not described the way we described them (the conditioning in the second step required the intersection to be at least a certain size, not exactly a certain size). However, when the complex is large enough the TV-distance between the distributions in all previous works, and the distributions we described above, is negligible in the regime of parameters we consider. We therefore ignore this small point.

We also stress that while previous works give very strong agreement results, that apply for \(\varepsilon_0 = \exp(-\Omega(k))\), our technique currently yields results on for \((\eta,\varepsilon_0,O(1))\) such that \(\eta \exp(\poly(1/\varepsilon)) \ll 1\), so in particular, we will always think of \(\varepsilon = 1/\log(k)\) in the main theorem.

\subsection{Covering maps}
In this subsection we give a short introduction to covers and their connection to \(1\)-cohomology. We stress that everything we state in this subsection is well known. For a more in depth discussion, see \cite{Surowski1984}.

\begin{definition}[Covering map]\label{def:cover}
    Let \(Y,X\) be simplicial complexes. We say that a map \(\rho:Y(0) \to X(0)\) is a covering map if the following holds.
    \begin{enumerate}
        \item \(\rho\) is a surjective homomorphism.
        \item For every \(v \in X(0)\), and \((v,i) \in \rho^{-1}(\set{v})\) it holds that \(\rho|_{Y_{(v,i)}}:Y_{(v,i)}(0) \to X_v(0)\) is an isomorphism.
    \end{enumerate}
    We often denote \(\rho:Y\to X\). We say that \(\rho\) is an \(\ell\)-cover if for every \(v \in X(0)\) it holds that \(\Abs{\rho^{-1}(\set{v})}=\ell\). If there exists such a covering map \(\rho:Y\to X\) we say that \(Y\) covers \(X\).
\end{definition}

Let us see two examples for covers.
\begin{example} \label{ex:trivial-cover}
    Let \(X\) be any simplicial complex. Let \(Y\) be \(\ell\)-disconnected copies of \(X\). Then the projection of every to \(X\), \(\rho:Y \to X\), is a covering map. This cover a called the trivial cover.
\end{example}

\begin{figure}
    \centering
    \includegraphics[scale=0.5]{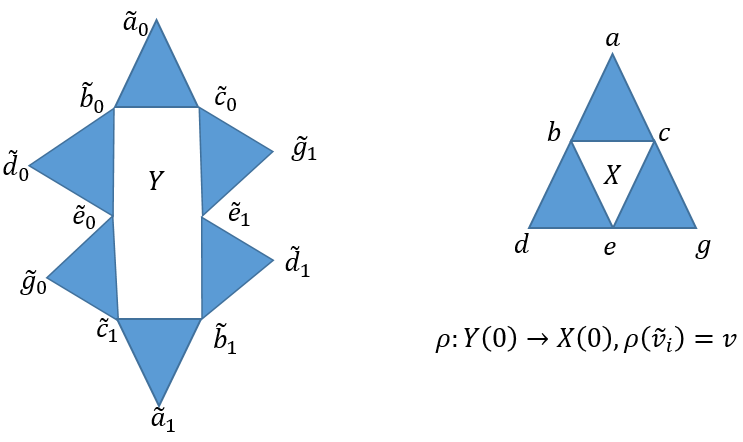}
    \caption{Non trivial connected cover}
    \label{fig:cover-example}
\end{figure}
\begin{example}
    \pref{fig:cover-example} contains an example of a \(\rho:Y \to X\) such that \(Y\) is connected.
\end{example}
Next we develop some basic properties of covers that are both necessary for our result, and give some intuition of the structure of covers.

The first property we show is that a cover \(\rho:Y \to X\) induces a permutations \(\pi_{uv}\) for every directed edge \(uv \in \tilde{X}(1)\) that ``encode'' the covers information as in the claim below.
\begin{claim} \label{claim:permutations-from-covers}
    Let \(X\) be a simplicial complex and let \(\rho:Y \to X\) be a cover. Let \(vu \in \dir{X}(1)\). Let \(\rho^{-1}(v) = \set{(v,1),(v,2),\dots,(v,\ell)}\) and \(\rho^{-1}(u) = \set{(u,1),(u,2),\dots,(u,\ell')}\). Then \(\ell = \ell'\) and there exists \(\pi_{uv}:[\ell]\to [\ell]\) such that \(\pi_{uv}(i)=j\) if and only if \(\set{(v,i), (u,j)} \in Y(1)\).
\end{claim}

Using these permutations we can show that all covers of a connected simplicial complex are \(\ell\)-covers.
\begin{corollary} \label{cor:every-cover-is-regular}
    If \(X\) is connected then any cover of \(X\) is an \(\ell\)-cover for some \(\ell \in \NN \cup \set{\infty}\).
\end{corollary}

\begin{proof}[Proof of \pref{claim:permutations-from-covers}]
    First we note that for every \((v,i)\) there is a unique \((u,j)\) such that \(\set{(v,i), (u,j)} \in Y(1)\). This is because \(\rho|_{\tilde{X}_{(v,i)}}:X_{(v,i)}(0)\to X_v(0)\) is an isomorphism so there is a single preimage of \(u\) in \(X_{(v,i)}(0)\). Hence there are functions \(\pi_{uv}:[\ell] \to [\ell']\) and \(\pi_{uv}:[\ell'] \to [\ell]\) such that \(\pi_{uv}(i)=j\) if and only if \(\set{(v,i), (u,j)} \in Y(1)\). These functions invert each other from their definition, i.e. \(\pi_{uv}(i)=j\) if and only if \(\set{(v,i), (u,j)} \in Y(1)\) if and only if \(\pi_{vu}(j)=i\) hence \(\ell=\ell'\) and this is the required permutation.
\end{proof}

\begin{proof}[Proof of \pref{cor:every-cover-is-regular}]
    Let \(X\) be a connected complex and assume that there exists a cover \(\rho:Y\to X\) such that \(Y\) is not an \(\ell\)-cover for any \(\ell\). Let \(\ell\) be the number of images of some arbitrary vertex \(v \in X(0)\). Let \(B = \sett{v' \in X(0)}{\Abs{\rho^{-1}(v')} = \ell} \ne \emptyset\). If \(X(0) \setminus B\) is not empty then the cut between \(B\) and \(X(0) \setminus B\) has an edge crossing it. But the number of preimages for both sides of the edge is equal which lead to a contradiction.
\end{proof}

Finally, let us show that in an \(\ell\)-cover, every \(s \in X\) has \(\ell\)-inverse images.
\begin{claim} \label{claim:inverse-image-of-l-cover}
    Let \(\rho:Y \to X\) be an \(\ell\)-cover. Then for every non-empty \(s \in X\),
    \[\abs{\rho^{-1}(s)} = \ell.\]
\end{claim}

\begin{proof}[Proof of \pref{claim:inverse-image-of-l-cover}]
    Let \(s \in X\). Let \(v \in s\) be any vertex inside \(s\) and we write \(s = \set{v} \dunion r\). By definition, there are \(\ell\)-vertices \((v,1),(v,2),\dots,(v,\ell)\) such that \(\rho((v,i)) = v\). By the local isomorphism property between \(X_v\) and every \(Y_{(v,i)}\), there are faces \(\tilde{r}_i \in Y_{(v,i)}\) such that \(\rho(\tilde{r}_i) = r\). The faces \(\tilde{s}_i = \tilde{r}_i \dunion \set{(v,i)}\) are \(\ell\)-inverse images of \(s\). This shows that \(\abs{\rho^{-1}(s)} \geq \ell\). Let us see that every preimage of \(s\) it one of these \(\tilde{s}_i\). Indeed, let \(\hat{s} \in \rho^{-1}(s)\). Let \((v,i) \in \hat{s}\) be the preimage of \(v\) in \(\hat{s}\). Then \(\hat{r} = \hat{s} \setminus \set{(v,i)} \in Y_{(v,i)}\) maps to \(r\). By the fact that \(\rho|_{Y_{(v,i)}}:Y_{(v,i)}(0) \to X_v(0)\), this implies that \(\hat{r}=\tilde{r}_i\). Thus \(\hat{s}= \tilde{s}_i\).
\end{proof}

\subsubsection{The induced function}
Let \(\rho:Y \to X\) be an \(\ell\)-cover. Without loss of generality we identify 
\[Y(0) = X(0) \times [\ell] = \sett{(v,i)}{v \in X(0), \; i=1,2,\dots,\ell}\]
where \(\rho((v,i)) = v\). We define \emph{the induced function} 
\begin{equation} \label{eq:induced-function}
    \psi_{\rho}:\dir{X}(1) \to Sym(\ell), \; \psi_{\rho}(vu)=\pi_{uv}
\end{equation}
where \(\pi_{uv}(i)=j\) are such that \(\set{(v,i), (u,j)} \in Y(1)\) (as in \pref{claim:permutations-from-covers}). The first thing we notice is that this function is asymmetric, i.e., \(\psi(uv)=\psi(vu)^{-1}\) for every edge \(uv \in \dir{X}(1)\). Let 
\[C^1(X,Sym(\ell)) = \sett{\psi:\dir{X}(1)\to Sym(\ell)}{f \text{ is asymmetric}}\]
be the space of asymmetric functions. These functions are sometimes referred to as non-abelian cochains in the literature. One may suspect if there is a bijection between covers \(\rho\) and \(\psi \in C^1(X,Sym(\ell))\), but there are many \(\psi \in C^1(X,Sym(\ell))\) whose permutations don't give rise to a cover. It turns out that and asymmetric function \(\psi \in C^1(X,Sym(\ell))\) corresponds to a cover if and only if for every triangle \(uvw \in \dir{X}(2)\) it holds that \begin{equation} \label{eq:cover-condition}
    \psi_{\rho}(vw) \circ \psi_{\rho}(uv)=\psi_{\rho}(uw).
\end{equation}
We denote by \(Z^1(X,Sym(\ell)) \subseteq C^1(X,Sym(\ell))\) 
\[ Z^1(X,Sym(\ell)) = \sett{\psi \in  C^1(X,Sym(\ell))}{\forall uvw \in \dir{X}(2), \; \psi(vw) \circ \psi(uv)=\psi(uw)}.\]

\begin{claim} \label{claim:necessary-conditionfor-covers}
    Let \(\rho:Y \to X\) be an \(\ell\)-cover. Then \(\psi_{\rho} \in Z^1(X,Sym(\ell))\).
\end{claim}

\begin{proof}[Proof of \pref{claim:necessary-conditionfor-covers}]
    Let \(\rho:Y \to X\) and fix \(uvw \in X(2)\), let us show that \(\psi_{\rho}(vw) \circ \psi_{\rho}(uv)=\psi_{\rho}(uw)\). Indeed, by definition \(\psi_{\rho}(xy) = \pi_{yx}\) where \(\pi_{yx}(i)=j\) if and only if \(\set{(x,i), (y,j)} \in Y(1)\). Indeed, let \(\pi_{uv}(i)=j\) and \(\pi_{wv}(i)=k\) and we need to show that \(\pi_{wu}(j)=k\), which is equivalent to showing that if \(\set{(v,i),(u,j)}, \set{(v,i),(w,k)} \in Y(1)\) then \(\set{(u,j),(w,k)} \in Y(1)\). \(\rho:Y_{(v,i)}(0) \to X_v(0)\) is an isomorphism, it holds that \(\set{(u,j), (w,k)} \in Y_{(v,i)}(1)\) and in particular the edge is in \(Y(1)\).
\end{proof}

If \(\psi \in Z^1(X,Sym(\ell))\), then we can construct a cover \(\rho=\rho_\psi:Y \to X\) such that \(\psi_{\rho} = \psi\) as follows.
\begin{definition} \label{def:induced-cover}
    Let \(\psi \in C^1(X,Sym(\ell))\) and denote by \(\psi(vu)=\pi_{uv}\). The induced cover is the complex \(Y=Y_\psi\) and mapping \(\rho_\psi:Y \to X\) defined by
    \begin{enumerate}
        \item \(Y(0) = X(0) \times [\ell]\) and \(\rho_\psi((v,i))=v\).
        \item \(Y(1) = \sett{\set{(v,i),(u,\pi_{uv}(i))}}{v \in X(0), vu \in \dir{X}(1), i \in [\ell]}\).
        \item \(\tilde{s}=\set{(v_0,i_0),(v_1,i_1),\dots, (v_m,i_m)} \in Y(m)\) if and only if \(s=\rho(\tilde{s})=\set{v_0,v_1,\dots,v_m} \in X(m)\) and for every \(v_{p},v_{q} \in s\) it holds that \(\pi_{v_q,v_p}(i_p)=i_q\).
    \end{enumerate}
\end{definition}

\begin{claim} \label{claim:reconstruction of a cover}
    Let \(\psi \in Z^1(X,Sym(\ell))\). Then \(\rho_\psi:Y \to X\) is an \(\ell\)-cover of \(X\) and \(\psi_{\rho}=\psi\).
\end{claim}

\begin{proof}[Proof of \pref{claim:reconstruction of a cover}]
It is obvious that \(\rho\) is surjective and that it is a homomorphism from the definition. Let \(v \in X(0)\) and \((v,i) \in Y(0)\). We show that \(\rho|_{Y_{(v,i)}(0)}:Y_{(v,i)}(0) \to X_v(0)\) is an isomorphism. First let us make sure that this is a bijection. By definition, for every edge \(\set{v,u} \in X(1)\) there is an edge \(\set{(v,i),(u,\pi_{uv}(i))} \in Y(1)\) so \(\rho|_{Y_{(v,i)}(0)}\) is surjective. A priori, there could have been more edges \(\set{(v,i),(u,k)}\) if \(i=\pi_{vu}(k)\), but because \(\psi\) is asymmetric, \(i=\pi_{vu}(k)\) implies that \(\pi_{uv}(i)=k\) so the restriction is also injective.

We continue by showing that this is an isomorphism. First note that this bijection is a homomorphism by definition (there are no faces \(\tilde{s} \in Y(m)\) unless their projection \(\rho(\tilde{s}) \in X(m)\), so the same holds for the link of \((v,i)\)). Thus it remains to show that for every \(s \in X_v(m)\) there is some \(\tilde{s} \in X_{(v,i)}(m)\) such that \(\rho(\tilde{s})=s\). By definition of the link, \(s \dunion \set{v}=t \in X(m+1)\). We will show that
\(\tilde{t}=\set{(v,i)} \dunion \sett{(u,j_u)}{u \in s, j_u=\pi_{uv}(i)}\) is in \(Y(m+1)\), which implies that \(\sett{(u,j_u)}{u \in s, j_u=\pi_{uv}(i)}=\tilde{s} \in Y_{(v,i)}(m)\) maps to \(s\).

As \(\rho(\tilde{t})=t\) this amounts to showing that all possible edges \(xy \in Y(1)\) for \(x,y \in \tilde{t}\). One case is when, say, \(x=(v,i)\) and \(y=(u,j_u)\). In this case, \(\set{(v,i),(u,\pi_{uv}(i))} \in Y(1)\) and \(\pi_{uv}(i)=j_u\) by definition, hence the \(xy \in Y(1)\). The other case is when \(x=(u,{j_u})\) and \(y=(w,{j_w})\) for some \(u,w \in s\). in this case we note that \(\pi_{uv}(i)=j_u, \pi_{wv}(i)=j_w\). By \eqref{eq:cover-condition} applied for \(vuw \in X(2)\), it holds that \(j_w=\pi_{wv}(i)=\pi_{wu}(\pi_{uv}(i))=\pi_{wu}(j_u)\). Hence \(\set{(u,{j_u}),(w,{j_w})} \in Y(1)\) and the statement is proven.

The fact that \(\psi_{\rho}=\psi\) follows directly from the definition of the edges in \(Y\).
\end{proof}

\subsubsection{Connectivity of covers}
A simply connected complex is a connected complex such that every cover looks like \pref{ex:trivial-cover}, i.e., a bunch of disconnected copies of the original complex.
\begin{definition}[simply connected complex]
    Let \(X\) be a connected simplicial complex. We say that \(X\) is \emph{simply connected} if it is connected and if for every \(\ell\)-cover \(\rho:Y \to X\) there exists a partition of \(Y\) to \(\ell\) disconnected components \(Y= Y_1 \dunion Y_2 \dunion ... \dunion Y_\ell\) such that \(\rho|_{Y_i}:Y_i \to X\) is an isomorphism.
\end{definition}

\subsubsection{Further properties of covers}
\begin{claim} \label{claim:restriction-of-cover}
Let \(\rho:Y \to X\) be an \(\ell\)-cover. Let \(A \subseteq X(0)\) and let \(X'\) be the induced subcomplex over the vertices of \(A\). Let \(Y'=\rho^{-1}(A)\). Then \(\rho|_{Y'}:Y' \to X'\) is an \(\ell\)-cover.
\end{claim}

\begin{proof}[Proof of \pref{claim:restriction-of-cover}]
    Let \(\psi \in Z^1(X,Sym(\ell))\) such that \(Y=X^\psi\) is the induced cover as in \pref{def:induced-cover}. Then \(\psi|_{X'(1)} \in Z^1(X',Sym(\ell))\) and \(\rho|_{Y'}\) is the induced cover of this restriction.
\end{proof}

\begin{claim}\label{claim:cover-of-clique-complex}
    Let \(X\) be a clique complex and let \(\rho:Y\to X\) be a cover of \(X\). Then \(Y\) is a clique complex.
\end{claim}

\begin{proof}[Proof of \pref{claim:cover-of-clique-complex}]
    Let \(\tilde{s} = \set{(v_0,i_0),(v_1,i_1),...,(v_k,i_k)} \subseteq Y(0)\) be clique. Thus \(\rho(\tilde{s}) = s \subseteq X(0)\) is also a clique, and thus \(s \in X(k)\), or equivalently \(s \setminus \set{v_0} \in X_{v_0}(k-1)\). The link of \(v_0\) in \(X\) is isomorphic to the link of \((v_0,i_0)\) via \(\rho\). In particular this implies that \(\tilde{s} \setminus \set{(v_0,i_0)} \in Y_{(v_0,i_0)}(k-1)\) which is equivalent to \(\tilde{s} \in Y(k)\).
\end{proof}

\begin{claim} \label{claim:expansion-of-cover}
    Let \(X\) be a \(\lambda\)-one or two sided high dimensional expander. Then any connected cover is a \(\frac{\lambda}{1-\lambda}\)-one or two sided spectral expander respectively.
\end{claim}

The proof of this claim relies on the trickling down theorem by Oppenheim \cite{Oppenheim2018}.

\begin{theorem}[\cite{Oppenheim2018}] \label{thm:trickle-down}
    Let \(X\) be a \emph{connected} simplicial complex and assume that for any vertex \(v \in X(0)\) it holds that \(X_v\) is a \(\lambda\)-one or two sided high dimensional expander. Then the underlying graph of \(X\) is a \(\frac{\lambda}{1-\lambda}\)-one or two sided spectral expander (respectively), which implies that \(X\) is a \(\frac{\lambda}{1-\lambda}\)-one or two sided high dimensional expander (respectively).
\end{theorem}

\begin{proof}[Proof of \pref{claim:expansion-of-cover}]
    Let \(\rho:Y \to X\) be a cover such that \(Y\) is connected. For every vertex \(\tilde{v} \in Y\) such that \(\rho(\tilde{v})=v\), it holds that \(Y_{\tilde{v}} \cong X_v\). Thus in particular, if \(X\) is a \(\lambda\)-one or two sided high dimensional expander, this implies that \(Y_{\tilde{v}}\) is a \(\lambda\)-high dimensional expander. By \pref{thm:trickle-down}, \(Y\) is also a \(\frac{\lambda}{1-\lambda}\)-one or two sided high dimensional expander.
\end{proof}

\subsection[Cover soundness à la (CLA)]{Cover-soundness à la \pref{eq:CLA}}
Let us extend the definition of sound agreement tests to accommodate for our new notion of covers.
\begin{definition}[Cover sound] \label{def:distribution-cover-soundness}
    Let \(X\) be a simplicial complex and let \(\mathcal{D}\) be an agreement distribution on \(X\). Let \(\eta,\varepsilon_0,e > 0\) be constants. We say that \(\mathcal{D}\) is \emph{\((\eta,\varepsilon_0,e)\)-cover-sound} if for every ensemble of functions \(\mathcal{F} = \sett{f_r:r \to \Sigma}{r \in X(k)}\) such that \(\agr_{\mathcal{D}}(\mathcal{F}) = \varepsilon \geq \varepsilon_0\),
    there exists a simplicial $\frac{1}{\varepsilon^e}$-cover \(\rho:Y \to X\) and a global function \(G:Y(0)\to \Sigma\) such that 
    \begin{equation} \label{eq:dist-cover-soundness}
        \Prob[r \in Y(k)]{f_{\rho(r)}\circ \rho \overset{1-\eta}{\approx} G|_r} \geq \frac{1}{2}\varepsilon^e.
    \end{equation}
\end{definition}
The probability in \eqref{eq:dist-cover-soundness} is equivalent to choosing \(r \in X(k)\) and then a random preimage \(\tilde{r} \in Y(k)\). Therefore this inequality implies that for at least \(\frac{1}{2}\varepsilon^e\)-fraction of the \(f_r\)'s, the function agrees with \(G|_{\tilde{r}}\) for one of the preimages \(\tilde{r} \in \rho^{-1}(r)\) (i.e. using \(\rho\) to send vertices from \(\tilde{r}\) to \(r\)).

We note that here we did not require the cover to explain most of the agreement, but this definition actually implies such a statement. See discussion in \cite{DinurG2008} (there it is done for the complete complex, but their techniques are general).

\subsection{Cosystolic expansion and cover property testing}
Recall that \(Z^1(X,Sym(\ell)) \subseteq C^1(X,Sym(\ell))\) are all asymmetric functions such that \(\psi(uw)=\psi(vw)\circ \psi(uv)\) for every triangle \(uvw \in \dir{X}(2)\). For our result we will need simplicial complexes where this relation is locally testable. For this we define for every two function \(\psi,\phi: \dir{X}(k) \to Sym(\ell)\) their distance
\begin{equation} \label{eq:def-of-dist}
    \dist(\psi,\phi) = \Prob[s \in \dir{X}(k)]{\psi(s) \ne \phi(s)}.
\end{equation}
We also denote the weight of the function \(\wt(\psi) = \dist(\psi,Id)\) (where \(Id:\dir{X}(k) \to Sym(\ell)\) assigns every face \(s \in \dir{X}(k)\) the identity permutation).

For \(\psi \in C^1(X,Sym(\ell))\) we define \(\coboundary \psi: \dir{X}(2) \to Sym(\ell)\) by
\begin{equation} \label{eq:def-of-coboundary}
    \coboundary(\psi) = \psi(wu)\circ \psi(vw) \circ \psi(uv).
\end{equation}

We are ready to define cosystolic expansion.
\begin{definition} \label{def:def-of-cosyst-exp}
    Let \(X\) be a \(d\)-dimensional simplicial complex for \(d \geq 2\). Let \(\beta >0\). We say that \(X\) is a \(\beta\)-cosystolic expander if for every \(\ell \in \NN\), and every \(\psi \in C^1(X,Sym(\ell))\) there exists some \(\phi \in Z^1(X,Sym(\ell))\) such that
    \begin{equation} \label{eq:def-of-cosyst-exp}
        \beta \dist(\psi,\phi) \leq \wt(\coboundary \psi).
    \end{equation}
    Furthermore, we say that \(X\) is a \(\beta\)-coboundary expander if \(X\) is a \(\beta\)-cosystolic expander, and it is simply connected.
\end{definition}
An equivalent way to define coboundary expansion is that for every, $\phi\in Z^1(X,Sym(\ell))$ there is some $h:X(0)\to Sym(\ell)$, such that $\phi(uv) = h(v)\circ h(u)^{-1}$ for all $uv\in \vec X(1)$. We call \(\phi\)'s that have such an \(h\) coboundaries, hence the name ``coboundary expansion''.

An explanation is in order. We think of the equations \(EQ= \sett{\psi(uw)=\psi(vw)\circ \psi(uv)}{uvw \in \dir{X}(2)}\) as a set of tests and the weight \(wt(\coboundary \psi)\) measures the probability that \(\psi(uw)\ne \psi(vw)\circ \psi(uv)\), or in other words, the probability that \(\psi\) fails the test. When \(X\) is a \(\beta\)-cosystolic expander, this implies that if \(\wt(\psi) = \varepsilon\) then there is a function \(\phi \in Z^1(X,Sym(\ell))\) that is \(\varepsilon/\beta\) close to \(\psi\).

We remark that in other works, many other coefficient groups were used instead of \(Sym(\ell)\). For our result this definition is sufficient. Note that some prior works such as \cite{EvraK2016} also require that \(\psi \in Z^1 \setminus B^1\) have large support. This separate requirement is not necessary in our work.

Dinur and Meshulam already observed that cosystolic expansion (and coboundary expansion) is equivalent testability of covers, which they call \emph{cover stability} \cite{DinurM2019}. 

\subsubsection{Near-cosystols from flag complexes}

The following technical claim will be convenient later. The setup is as follows. Let \(X\) be a simplicial complex. Often is the case where for every edge \(uv \in X(1)\) we have permutations \(\pi_{u,uv}, \pi_{v,uv}\) and we are interested in constructing a co-chain \(\psi(vu)=\pi_{uv,u}^{-1} \circ \pi_{uv,v} = \pi_{u,uv} \circ \pi_{uv,v}\)\footnote{As a notational convention we use \(\pi_{x,y}=\pi_{y,x}^{-1}\).}. There is a natural consistency property that implies that \(\psi\) is (close to) a cocycle: suppose that for every triangle \(t \in X(2)\) and every \(v \in t\) or sub-edge of \(e \subset t\) there are permutations \(\pi_{v,t}, \pi_{e,t}\). Then for a given \(t \in X(2)\), if for every \(v \in e \subset t\) it holds that
\begin{equation}
    \pi_{t,v} = \pi_{t,e} \circ \pi_{e,v},
\end{equation}
then \(\coboundary \psi(uvw) = Id\). 
\begin{figure}
    \centering
    \includegraphics[scale=0.3]{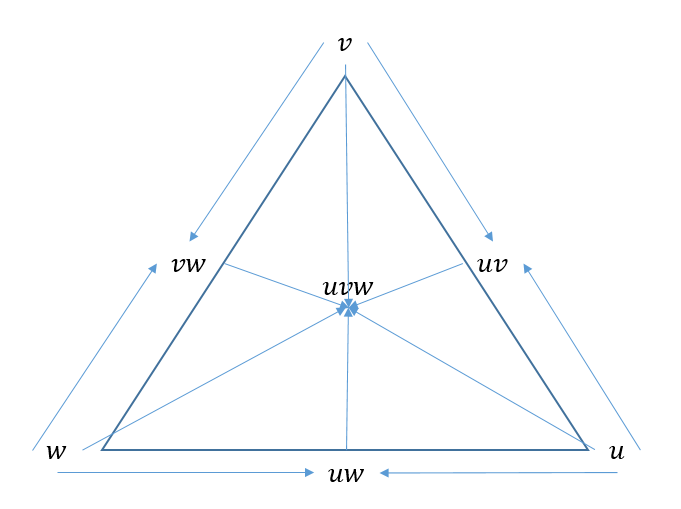}
    \caption{The following diagram should commute for \(\coboundary \psi(uvw)=Id\)}
    \label{fig:commuting}
\end{figure}
See \pref{fig:commuting} for an illustration.

To state this more easily, let us introduce the flag complex.

\begin{definition}[Flag complex] \label{def:flag-complex}
Let \(X\) be a simplicial complex. The flag complex is the complex \(GX\) whose vertices are the faces of \(X\), and \(\set{s_0,s_1,...,s_k} \in GX(k)\) if \(s_0 \subset s_1 \subset ... \subset s_k\) (for some ordering of the faces).
\end{definition}
\begin{claim} \label{claim:cochains-in-flag-complex-correspondence}
    Let \(X\) be a two dimensional simplicial complex and let \(\phi \in C^1(GX,Sym(\ell))\). Let \(\psi=\psi_\phi:X(1)\to Sym(\ell)\) be given by \(\psi(uv)=\phi(uv,v)\phi(u,uv)\). Then
    for every triangle \(t \in X(2)\), if for all six choices of $v\in e\in t$, \(\coboundary \phi(\set{v \in e \subset t}) = Id\), then \(\coboundary \psi(t) = Id\).
\end{claim}

\begin{corollary}\torestate{\label{cor:cocycles-in-flag-complex}
    Let \(X\) be a two dimensional simplicial complex and let \(\phi \in Z^1(GX,Sym(\ell))\). Let \(\psi=\psi_\phi:X(1)\to Sym(\ell)\) be given by \(\psi(uv)=\phi(uv,v)\phi(u,uv)\). Then  \(\psi \in Z^1(X,Sym(\ell))\).}
\end{corollary}

\begin{proof}[Proof of \pref{claim:cochains-in-flag-complex-correspondence}]
    The proof is just a calculation. Let \(uvw \in X(2)\) be a triangle. We need to show that \(\psi(wv)\circ \psi(uw)=\psi(uv)\). Note that
    \[\psi(uv)=\phi(uv,v)\phi(u,uv) = \phi(uv,v)\left [\phi(uvw,uv) \phi(uv,uvw) \right ] \phi(u,uv) = \phi(uvw,v)\phi(u,uvw).\]
    Where the last equality follows from \(\coboundary \phi(u,uv,uvw)=Id\) so \(\phi(uv,uvw) \phi(u,uv)=\phi(u,uvw)\) and \(\phi(uv,v)\phi(uvw,uv)=\phi(uvw,v)\).
    Similarly, we have that 
    \begin{equation}
        \psi(uw)=\phi(uvw,w)\phi(u,uvw); \; \psi(wv)=\phi(uvw,v)\phi(w,uvw).
    \end{equation}
    Thus,
    \begin{align*}
        \psi(uv) &= \phi(uvw,v)\phi(u,uvw)\\
        &=\phi(uvw,v)\left [ \phi(w,uvw) \phi(uvw,w) \right ]\phi(u,uvw)\\
        &=\left [\phi(uvw,v) \phi(w,uvw)\right ]\cdot  \left [\phi(uvw,w) \phi(u,uvw)\right ] \\
        &= \psi(wv)\psi(uw).
    \end{align*}
\end{proof}

\subsection{Swap cosystolic expansion and the Faces complex}
For any simplicial complex $X$, we define a related complex called the {\em faces complex} of $X$, denoted $\FX{d_1}X$. This is a complex whose vertices are $d_1$ faces of $X$ and whose edges connect two vertices whose disjoint union is also a face of $X$. Thus, the underlying graph of $\FX{d_1}X$ is the graph of the swap walk on $X$, see \cite{AlevFT2019,DiksteinD2019}. This definition already appeared as \pref{def:face-complex-intro} but we restate it again here.

\begin{definition}[Faces Complex] \label{def:face-complex}
    Let \(X\) be a \(d\)-dimensional simplicial complex. Let \(d_1 \leq d\). We denote by \(\FX{d_1}X=F(X,d_1)\) the simplicial complex whose vertices are \(\FX{d_1}X(0)=X(d_1)\) and whose faces are all \(\sett{\set{s_0,s_1,...,s_j}}{s_0\dunion s_1 \dunion \dots \dunion s_j \in X}\). 
\end{definition}
It is easy to verify that this complex is \(\left ( \lfloor \frac{d+1}{d_1+1} \rfloor - 1 \right )\)-dimensional and that if \(X\) is a clique complex then so is \(\FX{d_1}X\). We omit \(d_1\) from the notation when it is clear. 

For a face $s\in X$ we write $\FX{d_1}X_s = \FX{d_1}(X_s)$. In case \(s \in X(d_1)\) then \(\FX{d_1}X_s\) is actually the link of \(s\) in $\FX{d_1}X$, i.e. \(\FX{d_1}(X_s) = (\FX{d_1}X)_s\), so this notation is not overloaded. We will be interested in $\FX{d_1}X_r$ for faces $r$ not necessarily in \(X(d_1)\). These are sub-complexes of $\FX{d_1}X$ that are not links per se.

\begin{definition}[Well connected] \label{def:well-connected-complex}
    Let \(X\) be a \(d\)-dimensional simplicial complex. Let \(d_1 \leq d+2\). We say that \(X\) is \emph{\(d_1\)-well-connected} if for every \(r \in X^{\leq d_1}\) it holds that \(\FX{d_1}(X_r)=\FX{d_1}X_r\) is connected. Moreover, if \(r \in X(0)\) then we require that \(\FX{d_1}X_r\) is \emph{simply connected}. When \(d_1\) is clear from context we omit it and say that \(X\) is \emph{well connected}.
\end{definition}
We say that a complex $X$ has swap cosystolic expansion if $\FX{d_1}X$ is a cosystolic expander. Namely, recalling \pref{def:def-of-cosyst-exp},
\begin{definition}[Swap cosystolic expansion]\label{def:swapce}
    Let $d_1\in \mathbb{N}$ and let $\beta>0$. A simplicial complex $X$ of dimension $d>d_1$ has $(\beta,d_1)$ swap-cosystolic -expansion if $\FX{d_1}X$ is a $\beta$ cosystolic expander. We further say that $X$ has $(\beta,d_1)$ swap-coboundary-expansion if it has $(\beta,d_1)$ swap-cosystolic-expansion and furthermore, $\FX{d_1}X$ is simply connected. 
\end{definition}
\subsection{Sampling in HDXs}
\begin{theorem}[\cite{DiksteinH2023}] \label{thm:sampling-in-HDXs}
    Let \(k\leq d_1 \leq d\). Let \(\gamma < 1\) and let \(X\) be \(\frac{\gamma}{d}\)-one sided high dimensional expander. Let \(\zeta > 0\) and let \(f:X(k) \to [0,1]\). Let 
    \[B(f) = \sett{t \in X(d_1)}{\Abs{\Ex[r\in X(k), r\subseteq t]{f(r)}-\Ex[r \in X(k)]{f}} > \zeta}.\] 
    Then
    \[\Prob[t \in X(d_1)]{B(f)} \leq \exp(-\poly(\zeta) \frac{d_1}{k}).\]

    In particular, if \(A \subseteq X(k)\) and 
    \[B(A) = \sett{t \in X(d_1)}{\Abs{\cProb{r\in X(k)}{A}{r\subseteq t}-\Prob[r \in X(k)]{A}} > \zeta},\] 
    then
    \[\Prob[t \in X(d_1)]{B(A)} \leq \exp(-\poly(\zeta) \frac{d_1}{k}).\]
\end{theorem}
The ``in particular'' part follows from the case for \(f=\one_A\) the indicator of \(A\).

\pref{thm:sampling-in-HDXs} allows us to get a similar sampling guarantee on sampling of of tuples according to \(\mathcal{D}\) distributions extended to high dimensional expanders \(X\).

\begin{claim} \label{claim:edge-sampling}
    Let \(k \leq d_1 \leq d\) and \(q \in \NN\). Let \(d_1 \geq qk\) and let \(\zeta > 0\) be some constant. Let \(X\) be a \(d\)-dimensional simplicial complex as in \pref{thm:sampling-in-HDXs}. Let \(\mathcal{D}\) be an agreement distribution over \(\Delta_{qk+q}(k)\). and let \(E \subseteq \supp \mathcal{D}_X\). Then
    \[\Prob[s \in X(d_1)]{\Abs{\cProb{\set{r_i} \sim \mathcal{D}}{E}{\set{r_i} \subseteq s} - \Prob[\set{r_i} \sim D]{E}} > \zeta} \leq \exp(-\poly(\zeta) \frac{d_1}{k}).\]
\end{claim}

\begin{proof}[Proof of \pref{claim:edge-sampling}]
    Let \(f:X(qk+q-1) \to [0,1]\) be \(f(a)=\Prob[\set{r_i} \subseteq a]{E}\). Then \(\ex{f}=\prob{E}\) and \(\Ex[a \subseteq s]{f(a)} = \cProb{\set{r_i} \sim \mathcal{D}}{E}{\set{r_i} \subseteq s}\). \pref{thm:sampling-in-HDXs} gives us the claim.
\end{proof}

\subsection{Suitable complexes}
\begin{definition}[Suitable complex] \label{def:suitable-complexes}
    Let \(d>k>0\) be integers, and let \(\alpha > 0\). Let \(X\) be a \(d\)-dimensional simplicial complex. We say that \(X\) is \((d,k,\alpha)\)-\emph{suitable} if it has the following properties: 
\begin{enumerate}
    \item There exists some integer \(d_1<d\) with the following properties:
    \begin{enumerate}
        \item \(k^3 \leq d_1 \leq d \exp(-\alpha\frac{d_1}{k})\).
        \item \label{item:cob-exp} $X$ has \((\exp(-\alpha \frac{d_1}{k}),d_1)\)-swap-cosystolic expansion, as in \pref{def:swapce}. 
        \item \label{item:well-connected} \(X\) is \(d_1\)-well connected as in \pref{def:well-connected-complex}.
    \end{enumerate}
    \item \(X\) is a \(\frac{1}{d^2}\)-two-sided local spectral expander.
    \item \(X\) is a clique complex. 
    \end{enumerate}
We also say that $X$ is suitable with respect to $d_1$.
\end{definition}

These properties are abstract and a priori it is not clear whether there is a complex that satisfies them, especially the property of the face complex being a coboundary expander. However, the following two families of complexes satisfy all of these properties.
\begin{example}~
    \begin{enumerate}
    \item The \(SL_d(\mathbb{F}_q)\)-spherical building (see \pref{def:sph-bldg}) for sufficiently large prime power \(q\) and dimension \(d\). 
    \item \cite{LubotzkySV2005b} complexes, that are connected simplicial complexes whose links look like the spherical buildings in the first item.
\end{enumerate}
\end{example}

\begin{theorem}[Based on \cite{DiksteinD2023b}] \label{thm:spherical-building-and-LSV-complexes-are-suitable}
Let \(d,k,\alpha\) be such that \(\alpha \in (0,1)\) and \(d \geq \poly(k,\alpha)\). Then there exists a constant \(q_0=q_0(k,\alpha)\) such that for every prime power \(q>q_0\):
    \begin{enumerate}
        \item The \(SL_{d}(\mathbb{F}_q)\)-spherical building is a \((d,k,\alpha)\)-\emph{suitable}.
        \item \(d\)-dimensional \cite{LubotzkySV2005b} complexes, such that the link of every vertex is isomorphic to \(SL_{d-1}(\mathbb{F}_q)\)-spherical building are \((d,k,\alpha)\)-suitable. 
    \end{enumerate}
\end{theorem}
This paper is focused on the \(1\%\)-agreement theorem; the proof of swap cosystolic expansion gives rise to other significant challenges so we defer the proof of the swap cosystolic expansion to a companion paper \cite{DiksteinD2023b}.

\begin{proof}
In both cases we fix \(k\) and take \(d_1 = k^3\).
\begin{enumerate}
\item \begin{enumerate}
        \item\textbf{Swap-coboundary and cosystolic expansion.} This is the main theorem in \cite{DiksteinD2023b}, which shows that \(F{d_1}X\) is a \(\exp(-\Omega(\sqrt{d_1})\)-cosystolic expander for these complexes (a coboundary expander in the spherical building case).
        \item\textbf{Well connectedness.} Let \(s\) be a face in a complex \(X\) (which is one of the complexes in the theorem). Then \((\FX{d_1}X)_s = \FX{d_1}(X_s)\). In both cases, this is (isomorphic to) a faces complex of a link in the spherical building. In particular, \cite{DiksteinD2023b} shows that these are coboundary expanders - a strictly stronger property than simple connectivity.
\end{enumerate}
    \item\textbf{high dimensional expansion.} High dimensional expansion of the spherical buildings decays with \(q\). This bound is independent of the ambient dimension, as proven in \cite{AbdolazimiO2022}. Therefore there is a sufficiently large enough \(q\) such that the spherical building is a \(\frac{1}{d^2}\)-high dimensional expander. High dimensional expansion of the \cite{LubotzkySV2005b} complexes follows from the fact that their links are spectral expanders and \pref{thm:trickle-down} by \cite{Oppenheim2018}.
    \item\textbf{Clique complexes.} It is easy to show that the spherical building is a clique complex by definition. If \(\set{u_0,u_1,\dots,u_k}\) are such that all edges appear in the spherical building, it means that every two subspaces are contained in one another. Rearranging them by dimension shows that this in indeed a flag, i.e. a face in the spherical building. The fact that quotients of the spherical buildings are clique complexes appear e.g. in \cite{LubotzkySV2005a}. 

\end{enumerate}
\end{proof}

%% file: sections/theorem-proof.tex
\section{Main Agreement Theorem}
In this section we describe and prove our main agreement theorem. Recall that we defined \emph{suitable complexes} in \pref{def:suitable-complexes}.

\begin{theorem}[Main] \label{thm:main}
    For every $\epsilon_0>0$, \(p > 1\), and $k \geq \exp(\poly(1/\epsilon_0))$ there exists \(\alpha = \poly(\varepsilon)\) and \(d_0\in\mathbb{N}\) such that the following holds for any \(d \geq d_0\). Let \(\eta \leq \exp(-\poly(1/\varepsilon_0))\) be sufficiently small. Let \(\mathcal{D}\) be an agreement distribution on \(\Delta_k(pk)\) such that its extension to \(\Delta_k(m)\) is \((\eta,\poly(\frac{1}{k}),O(1))\)-sound for every \(m \geq k^3\). Let \(X\) be a \((d,k,\alpha)\)-suitable complex and let \(\mathcal{D}_X\) be the extension of \(\mathcal{D}\) to $X$. Then \(\mathcal{D}_X\) is \((\varepsilon_0,\gamma,O(1))\)-cover sound where \(\gamma = \exp(\poly(1/\varepsilon_0))\eta\).
\end{theorem}
We derive \pref{cor:spherical-building-list-agreement} and \pref{cor:general} as corollaries of this theorem in the end of this section.

\begin{corollary}[Formal version of \pref{cor:spherical-building-list-agreement}] \label{cor:spherical-building-agreement-expansion}
    For every \(\varepsilon_0 > 0\) and \(k \geq \exp(\poly(1/\epsilon_0))\) there exists \(d_0 \geq k\) and integer \(q_0\) such that the following holds. Let \(\eta \leq \exp(-\poly(1/\varepsilon_0))\). Let \(X\) be a \(d\)-dimensional spherical building with a sufficiently large field size \(q \geq q_0\) and dimension \(d\geq d_0\). Let \(\mathcal{D}\) be either the \(V\)-test or the \(Z\)-test for sets of dimension \(k\). Then:
    \begin{enumerate}
        \item \(\mathcal{D}_X\) is \((\eta, \varepsilon_0, O(1))\)-sound where \(\gamma = \exp(\poly(1/\varepsilon))\poly(1/k)\).
        \item There is an \(\ell = \poly(1/\varepsilon)\) and a list of functions \(G_1,G_2,\dots,G_\ell:Y(0) \to \Sigma\) such that
    \begin{equation} \label{eq:list-explains-most-of-the-agreement}
        \Prob[\set{r_i} \sim \mathcal{D}]{\text{\(\set{f_{r_i}}\) passes the agreement test but for every \(G_m\) there is at least one \(f_{r_i}\) s.t. \(f_{r_i} \overset{1-\gamma}{\napprox} G_m|_{r_i}\)} } \leq \varepsilon^2.
    \end{equation}
    \end{enumerate} 
\end{corollary}

\begin{corollary} \label{cor:existance-of-sparse-agreement-tests}
    For every \(\varepsilon_0> 0\), \(\eta > 0\) and \(k \geq \exp(\poly(1/\varepsilon_0)) \in \NN\) there exists \(C>0\) and an infinite family of (non-isomorphic) complexes \(X_n\) such that the following holds.
    \begin{enumerate}
        \item \(|X_n(k)| \leq C|X_n(0)|^2\).
        \item The extended \(V\) or \(Z\)-test on \(X_n\) is an \((\eta,\varepsilon_0,O(1))\)-sound.
    \end{enumerate}
\end{corollary}
 
\begin{corollary}[Formal version of \pref{cor:general}] \label{cor:LSV-building-agreement-expansion}
   For every \(\varepsilon_0 > 0\) and \(k \geq \exp(\poly(1/\epsilon_0))\), there exists \(d \geq k\) and integer \(q_0\) such that the following holds. Let \(X\) be a simplicial complex whose links are (isomorphic to)  the \(SL_d(\mathbb{F}_q)\)-spherical building with a sufficiently large field size \(q \geq q_0\). Let \(\mathcal{D}\) be either the \(V\)-test or the \(Z\)-test for sets of dimension \(k\). Then \(\mathcal{D}_X\) is \((\gamma,\varepsilon_0,O(1))\)-cover sound where \(\gamma = \exp(\poly(1/\varepsilon))\poly(1/k)\).
\end{corollary}
Note that this corollary applies for the \cite{LubotzkySV2005b} complexes in particular (but also slightly more generally).

\subsection*{Proof of the agreement theorem}
In this sub-section we prove \pref{thm:main}. We will give a high-level proof, relying on a few claims and lemmas that will be stated as needed. Then we will prove all outstanding claims in the following sections, formally completing the proof. For the rest of the section, we fix \(X,\calD,\calF \) to be as in the theorem statement. 

\subsection[From 1\% agreement to 99.9\% list-agreement]{From $1\%$ agreement to $99.9\%$ list-agreement} 
Our starting point is the following list-decoding lemma which we prove in \pref{sec:proof-of-lem-from-1-percent}.
\begin{lemma}\torestate{ \label{lem:from-1-percent-to-list-agreement}
    Let \(X, \mathcal{D},\mathcal{F}\) and $\epsilon$ be as in \pref{thm:main}. Let $d_1$ be an integer for which $X$ is suitable. Let $d_2=2d_1+1$ and let $d_3=3d_1+2$. There exists $\ell = \poly (1/\eps)$ such that for every $s\in X(d_1)\cup X(2d_1+1)\cup X(3d_1+2)$ there exists a tuple of functions \[ \overline{L}_s = (L_s^1,\ldots,L_s^\ell), \qquad L_s^i:s\to\Sigma\]
    such that the following holds. Let \(\gamma = \eta\exp(\poly(1/\varepsilon))\).
    \begin{enumerate}
        \item All but $\exp(-\poly(\varepsilon) \frac{d_1}{k})$ fraction of faces in $X(d)$ are good, where a face is good if for all $i=1,\ldots,\ell$,
\[\Prob[ { r \in X(k), r\subset s}]{L^i_s |_r \overset{1-\gamma}{\approx} f_r} \geq \poly(\varepsilon).\]
        \item For every pair \(s \subset t\), $s,t\in X(d_1)\cup X(2d_1+1)\cup X(3d_1+2)$ there is a permutation \(\pi_{t,s}:[\ell]\to [\ell]\) such that 
        \[\Prob[s,t]{\forall i=1,2,...,\ell, j=\pi_{t,s}(i) \; L_s^i \overset{1-\gamma}{\approx} L_t^{j}|_s} \geq 1-exp(-\poly(\varepsilon)\frac{d_1}{k}).\]
        \item For a random triple \(s,t,u \in X(d_1)\cup X(2d_1+1)\cup X(3d_1+2)\) such that \(s \subset t \subset u\) it holds that
        \[\Prob[s \subset t \subset u]{\pi_{s,u} = \pi_{s,t} \circ \pi_{t,u}} \geq 1-\exp(-\poly(\varepsilon)\frac{d_1}{k}).
        \] 
    \end{enumerate}}
\end{lemma}

\subsection[A cover of FX]{A cover $\widetilde{\FX{d_1}X}\twoheadrightarrow \FX{d_1}X$}
Using \pref{lem:from-1-percent-to-list-agreement} we construct an almost-cycle \(\psi\in C^1(\FX{d_1}X,Sym(\ell))\) as follows.
For each edge \((s_1,s_2)\) of \(\FX{d_1}X\), denoting \(t=s_1\dunion s_2\), let \[\psi((s_1,s_2))= \pi_{t,s_2}^{-1}\circ  \pi_{t,s_1}.\] %
    \begin{lemma}\label{lem:propreties-of-permutations-on-FX}~
    \begin{equation}\label{eq:almost-cocycle}
        \wt(\zeta(\psi))  = \Prob[(s_1,s_2,s_3)\in \FX{d_1}X(2)]{\psi(s_1,s_2)\circ \psi(s_2,s_3) \neq \psi(s_1,s_3)}= exp(-\poly(\varepsilon)\frac{d_1}{k}).
    \end{equation}
    \end{lemma}

    \begin{proof}[Proof of \pref{lem:propreties-of-permutations-on-FX}]
        Recall \pref{def:face-complex} of the faces complex $\FX{d_1}X$ and recall \pref{def:flag-complex} of the flags complex $GX$. We wish to use \pref{claim:cochains-in-flag-complex-correspondence} for the flag complex of the faces complex, namely for $G(\FX{d_1}X) = G\FX{d_1}X$.  We define \(\phi:G\FX{d_1}X(1) \to Sym(\ell)\) as follows. Let \(\set{\set{s_1},\set{s_1,s_2},\set{s_1,s_2,s_3}}\) be a triangle in the flag complex. Let \(t = s_1 \dunion s_2\) and let \(u = s_1 \dunion s_2 \dunion s_3\). Then
        \[\phi(\set{s_1}, \set{s_1,s_2})=\pi_{s_1,t}, \quad \phi(\set{s_1}, \set{s_1,s_2,s_3})=\pi_{s_1,u}, \quad \phi(\set{s_1,s_2}, \set{s_1,s_2,s_3})=\pi_{t,u}.\]
        Where the permutations are the ones given in \pref{lem:from-1-percent-to-list-agreement}. By the third item of \pref{lem:from-1-percent-to-list-agreement}, it holds that with probability \(\wt(\coboundary \phi) = \exp(-\poly(\varepsilon)\frac{d_1}{k})\) and by \pref{claim:cochains-in-flag-complex-correspondence} the same holds for \(\coboundary(\psi)\).
    \end{proof}
Using \pref{lem:propreties-of-permutations-on-FX} and \(\beta=\exp(-\alpha\frac{d_1}{k})\)-cosystolic expansion of \(\FX{d_1}X\) we have that
\begin{corollary} \label{cor:good-cocycle}
    There exists a cocycle \(\psi' \in Z^1(\FX{d_1}X,Sym(\ell))\) such that $\dist(\psi',\psi) = \exp(-\poly(\varepsilon)\frac{d_1}{k} + \alpha \frac{d_1}{k})$, namely,
    \[\Prob[s_1,s_2]{\psi'(s_1,s_2) = \pi_{s_2,t}^{-1}\circ  \pi_{t,s_1}} \geq 1-\exp(-\poly(\varepsilon)d_1/k). \qed\]
\end{corollary}
Here is where we require that \(\beta\) is large enough relative \(\exp(-\poly(\varepsilon)\frac{d_1}{k})\). By \pref{def:induced-cover}, \(\psi'\) gives rise to an \(\ell\)-cover \(\nu_{\psi'}:\widetilde{\FX{d_1}X} \to \FX{d_1}X\).
\subsection[Construction of a cover Y to X]{Construction of a cover \(Y \twoheadrightarrow X\)} Recall that our goal is to construct a cover \(Y \twoheadrightarrow X\) (together with a global function on the vertices of \(Y\) such that \eqref{eq:dist-cover-soundness} holds). So far we managed to construct a cover \(\nu = \nu_{\psi'}:\widetilde{\FX{d_1}X} \to \FX{d_1}X\). It seems natural to expect that every such cover \(\widetilde{\FX{d_1}X}\) is (isomorphic to) a face complex \(FY\) of some complex \(Y\). When \(X\) is a well-connected clique complex, this is indeed the case.

\begin{lemma} \label{lem:Y-is-an-honest-to-god-cover-of-X}
    Let \(\nu:\widetilde{\FX{d_1}X} \to \FX{d_1}X\) be an \(\ell\)-covering map. Then there exists an \(\ell\)-covering map \(\rho:Y \to X\) such that \(FY\) and \(\widetilde{\FX{d_1}X}\) are isomorphic.
\end{lemma}
\begin{center}
\begin{tikzcd}[arrows=rightarrow]
Y\ar[d,"\rho"]\ar[r,"faces"]  & \quad FY\stackrel{\iota}\cong \widetilde{\FX{d_1}X}\ar[d,"\nu"]  \\
		 X\ar[r,"faces"] & \quad \FX{d_1}X 
\end{tikzcd}
\end{center}
We denote by \(\iota: FY \overset{\sim}{\to} \widetilde{\FX{d_1}X}\) the isomorphism. We prove this lemma in \pref{sec:construction-of-cover}.
\subsection[The global function G]{The global function \(G\)}\label{sec:globalfunc}
Given \pref{lem:Y-is-an-honest-to-god-cover-of-X}, we can define an ensemble of functions for \(FY\) to be \(h_{\tilde{s}}:\tilde{s} \to \Sigma\), \(h_{\tilde{s}}=L_s^j\circ \rho\) such that \(\iota(\tilde{s}) =(s,j)\) is the identification promised in \pref{lem:Y-is-an-honest-to-god-cover-of-X}. We define a majority (plurality) function \(G:Y(0) \to \Sigma\) by \(G(v) = plurality_{s \ni v} h_s(v)\). We prove that
\begin{claim} \torestate{\label{claim:agreement-with-majority}
    There exists a universal constant \(c>0\) such that
    \(\Prob[\tilde{s}\in Y(d_1)]{h_{\tilde{s}} \overset{1-O(\gamma)}{\approx} G|_{\tilde{s}}} \geq 1-\exp(-\poly(\varepsilon)\left ( \frac{d_1}{k} \right )^c ) - \frac{1}{d_1^c}.\)}
\end{claim}

This claim is proven in \pref{sec:global-func}. Combining all the pieces we can finally prove the theorem. Consider the probability \(\Prob[\tilde{r} \in Y(k)]{f_{r}\circ \rho \overset{1-O(\eta)}{\napprox} G|_{\tilde{r}}}\). Sampling \(\tilde{r} \in Y(k)\) is the marginal of sampling a pair \((\tilde{r} \subseteq \tilde{s})\) in \(Y\) such that \(\tilde{s} \in Y(d_1)\). Denote by \(\iota(\tilde{s})=(s,j)\). If \(f_r \overset{1-O(\gamma)}{\approx} L_s^j|_{r}\), and \(h_{\tilde{s}}|_{\tilde{r}} \overset{1-O(\gamma)}{\approx} G|_{\tilde{r}}\) then
\[f_r \circ \rho \overset{1-O(\gamma)}{\approx} L_s^j|_{r} \circ \rho = h_{\tilde{s}}|_{\tilde{r}} \overset{1-O(\gamma)}{\approx} G|_{\tilde{r}}.\]

Thus we define the ``bad events''
\begin{enumerate}
    \item \(E_1\) the event that \(f_r \overset{1-\gamma}{\napprox} L_s^j|_{r}\).
    \item \(E_2\) the event that \(h_{\tilde{s}}|_{\tilde{r}} \overset{1-O(\gamma)}{\napprox} G|_{\tilde{r}}\).
\end{enumerate}
and we have that 
\[\Prob[\tilde{r} \in Y(k)]{f_{r}\circ \rho \overset{1-O(\gamma)}{\napprox} G|_{\tilde{r}}} \leq \prob{E_1} + \prob{E_2}.\]

By \pref{lem:from-1-percent-to-list-agreement} \(\prob{E_1} \leq 1-\poly(\varepsilon)\). Thus to show the theorem it is enough to bound \(\prob{E_2} \leq \exp(-\poly(\gamma)k) + \exp(-\poly(\varepsilon)\left ( \frac{d_1}{k} \right )^c )+\frac{1}{d_1^c}\) for the \(c>0\) in \pref{claim:agreement-with-majority} (\(k\) and \(d_1\) are chosen large enough such that this is much smaller than the upper bound on \(\prob{E_1}\)). Indeed, 
\[\prob{E_2} \leq \Prob[\tilde{s},\tilde{r}]{h_{\tilde{s}} \overset{1-O(\gamma)}{\napprox} G|_{\tilde{s}} } + \Prob[\tilde{s},\tilde{r}]{h_{\tilde{s}} \overset{1-O(\gamma)}{\approx} G|_{\tilde{s}} \ve h_{\tilde{s}}|_{\tilde{r}} \overset{1-O(\gamma)}{\approx} G|_{\tilde{r}}}.\]
The first item in this sum is bounded by \(\exp(-\poly(\varepsilon)\left ( \frac{d_1}{k} \right )^c )\) by \pref{claim:agreement-with-majority}. The second is bounded by \(\exp(-\poly(\gamma)k)\) by a standard Chernoff argument.

Thus 
\[\Prob[\tilde{r} \in Y(k)]{f_{r}\circ \rho \overset{1-O(\eta)}{\approx} G|_{\tilde{r}}} \geq \poly(\varepsilon) - \exp(-\poly(\varepsilon)\left ( \frac{d_1}{k} \right )^c )-\frac{1}{d_1^c} - \exp(-\poly(\gamma)k) = \poly(\varepsilon)\]
and \pref{thm:main} follows.
    
\subsection{Deriving the corollaries}

\begin{proof}[Proof of \pref{cor:spherical-building-agreement-expansion}]
We begin with the first item, i.e. showing that there exists \emph{one} global function \(G:X(0) \to \Sigma\) for \(\varepsilon^{O(1)}\) of the \(f_r\)'s, \(G|_r \approx f_r\).

Fix \(\varepsilon_0\) and let \(d,k,\alpha\) be as in \pref{thm:main} (for agreement distributions for either \(2\) or \(3\) sets). By \pref{thm:spherical-building-and-LSV-complexes-are-suitable} there is a \(q_0\) such that for \(q \geq q_0\), \(X\) is suitable. As seen in \pref{ex:sound-distributions}, both the \(V\)-test and the \(Z\)-test are \((\varepsilon_1,\eta)\)-sound for \(\varepsilon_1, \eta = \poly(1/k)\). We take \(k\) to be large enough so that \(\varepsilon_0 \geq \varepsilon_1\) and so that \(\gamma = \exp(1/\poly(\varepsilon_0))\eta < 1\). By \pref{thm:main}, there exists a cover \(\rho:Y \to X\) and a global function such that \(G':Y(0) \to \Sigma\) agrees with the ensemble as in \pref{def:distribution-cover-soundness}.

It is well known that the spherical building is simply connected (see e.g. \cite{DiksteinD2023} for a proof of this fact). By simple connectivity, every connected component in the cover is isomorphic to \(X\), so a cover is in fact composed of many disconnected components, each isomorphic to \(X\). Hence we can think of the global function \(G':Y(0) \to \Sigma\) as a list of functions \(G_1,G_2,\dots,G_\ell:X(0) \to \Sigma\), each is a restriction of \(G'\) to one of the components isomorphic to \(X\). From this we get that there is at least function \(G=G_i\) such that \eqref{eq:dist-cover-soundness} holds conditioned on sampling \(\tilde{r}\) in \(G_i\)'s component.

\medskip

Turning to \eqref{eq:list-explains-most-of-the-agreement}, by \cite[Theorem 5.1]{DinurG2008}, if
\[\agr_{\mathcal{D}}(\mathcal{F}) = \varepsilon \qquad \Rightarrow \qquad \exists G:X(0) \to \Sigma \; \Prob[r \in X(k)]{f_r \overset{1-\gamma}{\approx} G|_r} \geq \poly(\varepsilon)\]
holds, then there is an (inefficient) algorithm that outputs a list of functions \(G_1,G_2,\dots,G_\ell\) for \(\ell = \poly(1/\varepsilon)\) such that \eqref{eq:list-explains-most-of-the-agreement} holds. 
\end{proof}

\begin{proof}[Proof sketch of \pref{cor:existance-of-sparse-agreement-tests}]
    By \pref{cor:spherical-building-agreement-expansion} for every \(\varepsilon_0,\eta > 0\) and \(k\) sufficiently large, there is an integer \(d_0\) and a prime power \(q_0\) such that the following holds. Let \(X_q\) be the \(d_0\)-dimensional spherical building over any \(\F_q\) for \(q\geq q_0\). Then the \(V\) or \(Z\) tests on \(X\) are \((\eta,\varepsilon_0,O(1))\)-sound. 
    
    Let us show that \(|X_q(k)| = O_{d_0,k}|X_q(0)|\). The number of flags in a spherical building is $|X(d_0)|\approx q^{d^2/2}$ and for $k<d$ clearly $|X(k)|\leq 2^d\cdot|X(d)|$. However, the number of elements in \(X(0)\), which are subspaces in $|X(0)| \geq \binom{d_0}{d_0/2}_q \geq q^{d^2/4}$. The corollary follows.
\end{proof}

\begin{proof}[Proof of \pref{cor:LSV-building-agreement-expansion}]
Fix \(\varepsilon_0\) and let \(d,k,\alpha\) be as in \pref{thm:main} (for agreement distributions for either \(2\) or \(3\) sets). By \pref{thm:spherical-building-and-LSV-complexes-are-suitable} there is a \(q_0\) such that for \(q \geq q_0\), \(X\) is suitable. As seen in \pref{ex:sound-distributions}, both the \(V\)-test and the \(Z\)-test are \((\varepsilon_1,\eta,O(1))\)-sound for \(\varepsilon_1, \eta = \poly(1/k)\). We take \(k\) to be large enough so that \(\varepsilon_0 \geq \varepsilon_1\) and so that \(\gamma = \exp(1/\poly(\varepsilon_0))\eta < 1\). The corollary follows directly from \pref{thm:main}.
\end{proof}

%% file: sections/list-ag-2.tex
\section{From one percent agreement to list agreement} \label{sec:proof-of-lem-from-1-percent}
In this section we prove \pref{lem:from-1-percent-to-list-agreement}. 
\restatelemma{lem:from-1-percent-to-list-agreement}
Fix \(X\) to be a suitable simplicial complex and fix \(\mathcal{F} = \sett{f_r:r\to \Sigma}{r \in X(k)}\) to be an ensemble of functions on \(X(k)\). The distribution \(\mathcal{D}\) is \((\varepsilon_0,\eta)\)-sound. We also assume that \(\agr_{\mathcal{D}}(\mathcal{F})^{M}=\varepsilon^{M} \geq \varepsilon_0\) for some constant integer \(M > 0\) (which we do not explicitly calculate).

For a function \(g:X(0) \to \Sigma\) We denote by \(\supp_{\zeta}(g) = \sett{r \in X(k)}{f_r \overset{1-\zeta}{\approx} g}\) and call this set the \(\zeta\)-support of \(g\). We also denote by \(A_\zeta(g) = \sett{\set{r_1,r_2}}{r_1,r_2 \in \supp_\zeta (g) \ve f_{r_1}=f_{r_2}}\) the ``agreeing edge set of \(g\)'' and by \(\alpha_{\zeta}(g) = \Prob[\set{r_1,r_2} \sim D]{\set{r_1,r_2} \in A_{\zeta}(g)}\). Here \(D\) is the agreement distribution. When the ensemble of functions is not clear from context we denote \(\supp_{\zeta}^{\mathcal{F}}(g), A_{\zeta}^{\mathcal{F}}(g)\) etc.

For a face \(t \in X\), we denote by \(\mathcal{F}_t = \sett{f_r \in \mathcal{F}}{r \subseteq t}\). The set of functions whose agreeing set of edges is large is denoted by \(\overline{L}_{\tau,\zeta}(t) = \sett{g:t \to \Sigma}{\alpha_{\zeta}^{\mathcal{F}_t} (g) \geq \tau}\). In a lot of the following section we will consider partial functions \(g:t \to \Sigma\) (instead of all \(X(0)\)). In this case when we write \(\supp_\zeta(g), \alpha_\zeta(g)\) or \(A_\zeta (g)\) we will mean the support or agreement of \(g\) on its domain with respect to \(\mathcal{F}_t\).

\subsection{Overview - constructing the permutations in \pref{lem:from-1-percent-to-list-agreement}}
The idea is to use the the agreement theorem at hand to create lists of functions \(\overline{L}(s)=\set{L_s^1,L_s^2,\dots,L_s^\ell} \subseteq \overline{L}_{\tau,\zeta}(s)\) for every face \(s\). 
Then one needs to show that there is a matching between most lists \(\overline{L}(s), \overline{L}(t)\) such that \(s \subseteq t\). The matching we construct is so that \(\pi_{s,t}(i)=j\) if and only if \(L_t^i|_s \overset{1-\gamma}{\approx} L_t^j\) for some small \(\gamma>0\).

There are three degrees of freedom here, namely the choice of \(\tau,\zeta\) and \(\gamma\), and even after choosing them, it is not clear a priori which functions should appear in the list. We need the size of the lists to be \(\poly(\frac{1}{\varepsilon})\), so we cannot take all functions in \(\overline{L}_{\tau,\zeta}^0 (s)\). This gives rise to many problematic issues. The first two things we need to consider are:
\begin{enumerate}
    \item If \(f_r \overset{1-\gamma_1}{\approx} g|_r\) and \(g|_r \overset{1-\gamma_2}{\approx} g'|_r\) then \(f_r \overset{1-\gamma_1-\gamma_2}{\approx} g'|_r\). Hence, if two functions \(g,g':t \to \Sigma\) are close (in Hamming distance), then their support will have a lot of intersection. In particular, if \(L_t^i \in \overline{L}(t)\) we should exclude from the list a small Hamming ball around \(L_t^i\). Otherwise, if \(L_t^i|_s\) is \(\gamma\)-close to both \(L_s^j,L_s^{j'}\) we will have a difficulty of choosing whether \(\pi_{s,t}(i)=j\) or \(\pi_{s,t}(i)=j'\) in a way that will ensure that \(\pi_{s,u} = \pi_{s,t} \circ \pi_{t,u}\) for most \(s \subseteq t \subseteq u\).
    \item We want the lists to be exhaustive - meaning that every function \(g\) that has non-trivial agreement with the \(f_r\)'s has a close by function in the list. That is, if \(g:s \to \Sigma\) has that \(f_r \approx g|_r\) for \(\geq \tau=\poly(\varepsilon)\) fraction of the \(r \in X(k)\) inside \(s\), then either \(g \in \overline{L}(s)\) or \(g' \in \bar{L}(s)\) for some \(g' \approx g\). Otherwise, this could also lead to a difference in \(\overline{L}(s),\overline{L}(t)\)'s sizes, since for instance, it may be that \(g \in \overline{L}(t)\) but its restriction \(g|_s\) is not close to any element in the list of \(\overline{L}(s)\).
\end{enumerate}
Towards dealing with these issues we define \emph{density} and \emph{separation} of a set.
\begin{definition}
    Let \(\overline{L}_1,\overline{L}_2 \subseteq \set{g:s \to \Sigma}\). Let \(\gamma > 0\).
    \begin{enumerate}
        \item We say that \(\overline{L}_1\) is \(\gamma\)-separated if for every two \(g_1,g_2 \in \overline{L}_1\) it holds that \(\dist(g_1,g_2) > \gamma\).
        \item We say that \(\overline{L}_1\) is \(\gamma\)-dense in \(\overline{L}_2\) if for every \(h \in \overline{L}_2\) there exists \(g \in \overline{L}_1\) such that \(\dist(g,h) \leq \gamma\) (or as we denoted \(g \overset{1-\gamma}{\approx} h\)).
    \end{enumerate}
\end{definition}
Thus, for every \(s\) we will need to find a list \(\overline{L}(s)\) that is both separated (to overcome the first issue) and dense inside \(\overline{L}_{\tau,\zeta}(s)\) (to overcome the second issue). In fact, a method by \cite{DinurHKLT2018} allows us to find lists that are \(\gamma\)-dense but \(10\gamma\)-separated (for a smartly chosen \(\gamma\)). This is the content of \pref{claim:good-gamma-for-every-t-and-delta}. We will need this property in order to match between lists. If \(\overline{L}(s)\) is \(\gamma\)-dense in \(\overline{L}_{\tau,\zeta}(s)\) and \(\overline{L}(t)|_s = \sett{L_t^i|_s}{L_t^i \in \overline{L}(t)}\) is contained in \(\overline{L}_{\tau,\zeta}(s)\) then for every \(L_t^i \in \overline{L}(t)\) there exists some \(L_s^j \in \overline{L}(s)\) such that \(L_s^j \overset{1-\gamma}{\approx} L_t^i|_s\).\footnote{There is a subtlety here that even if \(\alpha_{\zeta}(L_t^i|_s) \geq \tau\), it could be that \(\alpha_{\zeta}(L_t^i) < \tau\). We will explain in the actual proof how to overcome this using sampling arguments, but for this overview let us just ignore this issue and assume that \(\alpha_{\zeta}(L_t^i|_s) = \alpha_{\zeta}(L_t^i)\).} %

The \(10\gamma\)-separation will promise that there is only a single such \(L_s^j\) for every \(L_t^i\), because if \(L_s^{j},L_s^{j'}\) are both \(\gamma\)-close to \(L_t^i\) then they are \(2\gamma\)-close to one another, which is a contradiction to \(10\gamma\)-separation. If we have such lists this implies that \(\pi_{s,t}(i)=j \Leftrightarrow L_t^i|_s \overset{1-\gamma}{\approx} L_s^j\) is a well defined function.

By making sure that \(\overline{L}(t)\) is also \(10\gamma\)-separated, we get that this function \(\pi_{s,t}\) is injective for most \(s \subseteq t\). If \(L_t^i,L_t^{i'}\) are \(10\gamma\)-separated, then for most \(s \subseteq t\), it will hold that \(L_t^i|_s, L_t^{i'}|_s\) will still be more than \(2\gamma\)-far. Hence, both \(L_t^i|_s, L_t^{i'}|_s\) cannot be \(\gamma\)-close to the same \(L_s^j\) by the triangle inequality. Finally, we show in the proof that the separation of all lists will ensure that their size stays \(\poly(\frac{1}{\varepsilon})\). For more details on the necessary condition for constructing these injections, see \pref{claim:injective-function} and its proof.

\medskip
Up until now we have explained in high level how to construct injective functions from \(\overline{L}(t)\) to \(\overline{L}(s)\) for \(t \supseteq s\). Surjectivity requires more care because \(\pi_{s,t}\) may still not be surjective even if \(\overline{L}(s),\overline{L}(t)\) are dense and separated. It could be that there is a ``new'' function \(L_s^j \in \overline{L}_{\tau,\zeta}(s)\) that is not close to \emph{any} function in the \(L_t^i|_s\). Another way of saying this is that \(\overline{L}(t)|_s\) may not be dense in \(\overline{L}_{\tau,\zeta}(s)\). %

To overcome this we will first show a structural property. We show that for some fixed \(t\), if \(\pi_{s,t}\) is not surjective for a non-negligible fraction of the \(s \subseteq t\), this implies that there is some \(\zeta' = O(\zeta)\), and a function \(g:t \to \Sigma\) such that \(\alpha_{\zeta'}(g) \geq \tau - \poly(\varepsilon)\), that is \(\zeta'\)-far from all functions in \(\overline{L}_{\tau,\zeta}(t)\). Showing this is where the previous agreement theorems come into play; if it is true that \(\overline{L}_{\tau,\zeta}(t)\) is not dense for many of the \(s \subseteq t\), this means that even after rerandomizing \(f_r\) for \(r \in \bigcup_{g \in \overline{L}_{\tau,\zeta}(t)} \supp_{\zeta}(g)\), the ensemble will pass the agreement test with non-negligible probability. By the agreement soundness guarantee, there is a new function \(g\) such that \(A_\zeta (g) \geq \tau^3\). We then show that by scaling \(\zeta\) by a constant factor, it also holds that \(A_{\zeta'} (g) \geq \tau-\poly(\varepsilon)\) (this happens because \(g\) needs to agree with the functions \(f_r\) for \(r \notin \bigcup_{g \in \overline{L}_{\tau,\zeta}(t)} \supp_{\zeta}(g)\), since these are the functions that were not rerandomized, although this needs to be meticulously argued).

On the other hand (keeping \(t\) fixed), we will show that there exists \(\zeta, \tau\) such that there are no new functions like this: We take any sequence \(\set{(\tau_i,\zeta_i)}_{i=1}^m\) where \(\tau_1=\varepsilon^2, \tau_{i+1}= \tau_i - \varepsilon^{50}\) and \(\zeta_{i+1} = 20\zeta_i\). Suppose that for every such pair there is a ``new'' \(g\) such that \(\alpha_{\zeta_{i+1}}(g) \geq \tau_{i+1} \geq \varepsilon^{10}\) that is far from all functions in \(\overline{L}_{\tau_i,\zeta_i}(t)\). This will imply that \(\prob{A_{\zeta_i}(g_i) \setminus \bigcup_{j=1}^{i-1} A_{\zeta_j}(g)} \geq \varepsilon^{10}\). This is because if \(g_i,g_{j}\) are far away from one another then most \(h_r\) that are close to \(g_i\) will be far from \(g_{j}\). Thus 
\[m\varepsilon^{10}\leq \sum_{i=1}^m \Prob[r_1,r_2]{A_{\zeta_i}(g_i) \setminus \bigcup_{j=1}^{i-1} A_{\zeta_j}(g_j)} = \Prob[r_1,r_2]{\bigcup_{i=1}^m A_{\zeta_i}(g_i)} \leq 1.\] 
We conclude that \(m \leq \varepsilon^{-10}\), i.e. that in every taking a long enough such sequence will result in some \(\set{(\tau_i,\zeta_i)}\) where this doesn't happen, which implies that \(\overline{L}_{\tau_i,\zeta_i}(t)\) is dense inside most \(\overline{L}_{\tau_i,\zeta_i}(s)\). This argument is depicted in \pref{claim:good-tau-and-delta-for-every-t}.

We use the pigeonhole principle to find a pair \((\tau_i,\zeta_i)\) and \(M_{(\tau_i,\zeta_i)} \subseteq X(d)\) of relative size \(\poly(\varepsilon)\), such that for every \(t \in M_{(\tau_i,\zeta_i)}\) and for all but a negligible fraction of \(s \subseteq t\) it holds that \(\overline{L}_{\tau_i,\zeta_i}(t)|_s\) is dense inside \(\overline{L}_{\tau_{i+1},\zeta_{i+1}}(s)\). Finally, using this set, we show that if \(t \in M_{(\tau_i,\zeta_i)}\) then most \(s \subseteq t\) also have the same property (i.e. that \(\overline{L}_{\tau_i,\zeta_i}(s)|_u\) is dense inside most \(\overline{L}_{\tau_{i+1},\zeta_{i+1}}(u)\) for most \(u \subseteq s\)). Using sampling, we conclude that this property propagates, i.e. that for a small enough dimension \(d' \ll d\) this density property holds for most \(t' \in X(d')\). This argument is in \pref{prop:good-parameters}.

Thus to conclude, we find our lists \(\overline{L}(s)\) by first finding the aforementioned \((\tau_i,\zeta_i)\). Then we find lists \(\overline{L}(s)\) that are \(\gamma\)-dense and \(10\gamma\)-separated inside \(\overline{L}_{\tau_{i+1},\zeta_{i+1}}(s)\) and define the permutation \(\pi_{s,t}(i)=j\) if and only if \(L_t^i|_s \overset{1-\gamma}{\approx} L_s^j\). These are well defined for most \(s \subseteq t\), and we use these for showing \pref{lem:from-1-percent-to-list-agreement}.

\subsection{Proof of \pref{lem:from-1-percent-to-list-agreement}}
This definition will be useful for the rest of this section.
\begin{definition}
    Let \(s \subseteq t\), \(\zeta > 0\) and \(L\) be a list of functions on \(t\). %
    We say that \(s\) \emph{samples \(L\) well} if %
    for every \(g_1,g_2 \in L\) it holds that \(\Abs{\dist(g_1|_s,g_2|_s) - \dist(g_1,g_2)} \leq \varepsilon^{100}\). We say that \(s\) \emph{samples \(L\)'s agreement} well if \(s\) samples \(L\) well and in addition, for every \(g \in L\) it holds that \(\Abs{\prob{A_{\zeta}(g)}- \prob{A_{\zeta}(g|_s)}} \leq \varepsilon^{100}\).
\end{definition}

We also define the requirements we need from the lists \(\overline{L}(s)\).
\begin{definition} \label{def:mi}
    Let \(d_1 \leq d_2 \leq d_3 \leq d\) be dimensions. Let \(\tau,\zeta,\gamma > 0\) be constants and let \(\ell\) be integers. Let \(M_i = M(\tau,\zeta,\gamma,\ell) \subseteq X(d_i)\) be all the \(t \in X(d_i)\) such that there exists some \(\overline{L}(t) \subseteq \overline{L}_{\tau,\zeta}(t)\) such that:
    \begin{enumerate}
        \item \(\overline{L}(t)\) is \(9\gamma\)-separated.
        \item \(\overline{L}(t)\) is \(\gamma\)-dense in \(\overline{L}_{\tau - \varepsilon^{100},\zeta}(t)\).
        \item \(|\overline{L}(t)|=\ell\).
        \item For \(i=2,3\), and \(i'< i\). Then the following set \(P_t(i')\) has probability \(\Prob[s \subseteq t, s \in X(i')]{P_t(i')} \geq 1-\exp(-\Omega(\poly(\varepsilon)\frac{d_1}{k}))\). \(P_t(i')\) is all the \(s \subseteq t, s \in X(i')\) such that \(s\) samples \(\overline{L}(t)\)'s agreement well.
    \end{enumerate}    
\end{definition}

The proof of \pref{lem:from-1-percent-to-list-agreement} relies on the following proposition.
\begin{proposition} \label{prop:good-parameters}
    There exists parameters \(\tau = \poly(\varepsilon),\zeta,\gamma \leq \eta\exp(\poly(1/\varepsilon))\) and \(\ell=\poly(1/\varepsilon)\) such that the following holds. Let \(d_1\leq d_2 \leq d_3 \leq d\) be such that \(\frac{d_3}{d} = \exp(-\poly(\varepsilon)\frac{d_1}{k})\). Then \(\prob{M_i} \geq 1- \exp(-\poly(\varepsilon)\frac{d_1}{k}) \).
\end{proposition}

We will also use this claim,
\begin{claim} \label{claim:injective-function}
    Let \(L_1 = \set{g_1,g_2,\dots,g_{\ell_1}}, L_2=\set{h_1,h_2,\dots,h_{\ell_2}}\) be such that \(L_1\) is \(\gamma\)-dense in \(L_2\) and such both \(L_1\) and \(L_2\) are \(2\gamma\)-separated. Then there exists an injective function \(\pi:[\ell_1]\to [\ell_2]\) such that \(\pi(j)=k\) if and only if \(g_j \overset{1-\gamma}{\approx} h_k\).
\end{claim}

\begin{proof}[Proof of \pref{lem:from-1-percent-to-list-agreement}]
Recall the definition of \(M_i\) is \pref{def:mi}. Given \pref{prop:good-parameters}, for every \(s \in M_1, t \in M_2\) and \(u \in M_3\) we take such a list \(\overline{L}(x) = \set{L_x^1,L_x^2,\dots,L_x^\ell}\), where \(x\) is \(s,t,u\) respectively (for \(x \notin M_i\) we take arbitrary lists). By definition, for every \(x \in M_i\) and every function in \(\overline{L}_x\) satisfies item \((1)\) in \pref{lem:from-1-percent-to-list-agreement} by definition. Let us verify the second and third items. We show the following for pairs \(s \subseteq t\), \(s \subseteq u\) and \(t \subseteq u\), we denote such a pair as \(x \subseteq y\) for short. For such a pair where \(y \in M_i, x \in P_y(i') \cap M_{i'}\) (where \(i,i'\) are the corresponding sets) we define \(\pi_{x,y}:[\ell]\to [\ell]\) by
\begin{equation} \label{eq:def-of-permutation}
    \pi_{x,y}(j)=k \Leftrightarrow L_y^j|_x \overset{1-\gamma}{\approx} L_x^k.
\end{equation} 
If either \(y \notin M_i\) or \(x \notin P_y(i')\cap M_{i'}\) we take \(\pi_{s,t}=Id\) as an arbitrary choice. We show that the assumptions in \pref{claim:injective-function} hold, which implies that this function is well defined. If this permutation is well defined for \(1-\exp(-\poly(\varepsilon)\frac{d_1}{k})\) of the pairs \(x \subseteq y\) this implies item \((2\) in \pref{lem:from-1-percent-to-list-agreement}.

First we note that if \(x\) samples \(\overline{L}(y)\)'s agreement well then \(\overline{L}(y)|_x \subseteq \overline{L}_{\tau-\varepsilon^{100},\zeta}(x)\) and \(\overline{L}(y)|_u\) is \(8\gamma\)-separated. Moreover if \(\overline{L}(x)\) is \(\gamma\)-dense in \(\overline{L}_{\tau-\varepsilon^{100},\zeta}(x)\), and \(\overline{L}(y)|_x \subseteq \overline{L}_{\tau-\varepsilon^{100},\zeta}(x)\) then in particular \(\overline{L}(x)\) is \(\gamma\)-dense in \(\overline{L}(y)|_x\). Finally, by definition \(\overline{L}(y)\) is also \(9\gamma\)-separated. Hence \pref{claim:injective-function} implies that this is a well defined and injective function. As both \(\overline{L}(y)|_x\) and \(\overline{L}(y)\) have the same size \(\ell\), this is indeed a permutation.

Finally, we will prove the third item in \pref{lem:from-1-percent-to-list-agreement},
\begin{equation*}
    \Prob[s \subseteq t \subseteq u, s \in X(d_1),t \in X(d_2), u \in X(d_3)]{\pi_{s,t}\circ \pi_{t,u} = \pi_{s,u}} \geq 1-\exp(-\Omega(\poly(\varepsilon) \frac{d_1}{k}).
\end{equation*}
Note that the following events are all occur with probability \(1-\exp(-\Omega(\poly(\varepsilon) \frac{d_1}{k})\):
\begin{enumerate}
    \item \(u \in M_1, s \in M_2, t \in M_3\).
    \item \(t \in P_u(i_2), s \in P_t(i_1) \cap P_u(i_1)\).
    \item For every \(j\) and \(k\), 
    \begin{equation} \label{eq:dist-preserved-in-list-restriction}
        \Abs{\dist(L_t^j|_s,L_u^k|_s) - \dist(L_t^j,L_u^k|_t)} \leq \gamma.
    \end{equation}
\end{enumerate}
The first two items are by \pref{prop:good-parameters}. The last item is via the following claim, proven later on (we will need this claim also for later):
\begin{claim} \label{claim:most-sets-sample-well}
    Fix \(t\), \(L\) and \(\zeta\). Then the fraction of \(s \subseteq t, s\in X(d_1)\) that \emph{don't} sample \(L\)'s agreement well is at most \(|L|^2 \exp(-\poly(\varepsilon) \frac{d_1}{k})\).
\end{claim}
Indeed, if \(t\) samples \(\overline{L}(u)\) well and \(s\) samples \(\overline{L}(u)|_t \cup \overline{L}(t)\) well, then distances between members of \(\overline{L}(u)|_t \cup \overline{L}(t)\) are preserved and the third item holds (and both these events occur with probability  \(1-\exp(-\poly(\varepsilon) \frac{d_1}{k})\)).

Thus to prove the theorem we assume \(s \subseteq t \subseteq u\) are such that all these events occur and show that \(\pi_{s,t}\circ \pi_{t,u} = \pi_{s,u}\). Let \(k = \pi_{s,t}\circ \pi_{t,u}(j)\) and \(k'=\pi_{s,u}(j)\). Then it holds that \(L_u^j|_s \overset{1-\gamma}{\approx} L_s^{k'}\), and that \(L_u^j|_t \overset{1-\gamma}{\approx} L_t^{\pi_{t,u}(j)}\) and \(L_t^{\pi_{s,t}(\pi_{t,u}(j))}|_s \overset{1-\gamma}{\approx} L_s^{k}\). By \eqref{eq:dist-preserved-in-list-restriction}, distances are approximately preserved, so \(L_u^j|_s \overset{1-2\gamma}{\approx} L_t^{\pi_{t,u}(j)}|_s\). Using the triangle inequality (where \(L^{\pi_{t,u}(j)}_t|_s\) is the point in the middle) it holds that \(L_u^j|_s \overset{1-3\gamma}{\approx} L_s^{k}\), which implies that \(L_s^k \overset{1-4\gamma}{\approx} L_s^{k'}\). By \(10\gamma\)-separation of \(\overline{L}(u)\), \(k=k'\).
\end{proof}

The proof of \pref{claim:injective-function} is direct.
\begin{proof}[Proof of \pref{claim:injective-function}]
     By density of \(L_2\) in \(L_1\), it holds that for every \(g_j\) there exists some \(h_k\) such that \(g_j \overset{1-\gamma}{\approx} h_k\). Moreover, there is only one such \(k\): assume that for some \(j\) there are \(k,k'\) such that \(g_j \overset{1-\gamma}{\approx} h_k\) and \(g_j \overset{1-\gamma}{\approx} h_{k'}\). Then by the triangle inequality \(h_k \overset{1-2\gamma}{\approx} h_{k'}\), which by \(2\gamma\)-separation implies that \(k=k'\). Thus it is a well defined function.

The same argument shows that this is an injection. Indeed, Let \(k=\pi(j)=\pi(j')\). That is, \(h_j \overset{1-\gamma}{\approx} h_k \overset{1-\gamma}{\approx} h_{j'}\). Then \(h_j \overset{1-2\gamma}{\approx} h_{j'}\). by \(2\gamma\)-separation of \(\overline{L}(s)|_u\), this implies that \(j=j'\), i.e. the function is an injection.
\end{proof}

So is the proof of \pref{claim:most-sets-sample-well}.
\begin{proof}[Proof of \pref{claim:most-sets-sample-well}]
    Let us start with sampling well. There are \(\leq \binom{|L|}{2}\) pairs \(g,g' \in L\) so by a union bound it is enough to argue that for at most \(\exp(-\poly(\varepsilon)d_1)\) of the \(s \subseteq t\) it holds that \(\Abs{\dist(g,g')-\dist(g|_s,g'|_s)} > \varepsilon^{100}\). Let \(A \subseteq t\) be the set of \(v \in t\) such that \(g(v)\ne g'(v)\). By \pref{thm:sampling-in-HDXs} the fraction of \(s \subseteq t\) such that \(\Abs{\prob{A} - \Prob[v \subseteq s]{A}} > \varepsilon^{100}\) is \(\exp(-\poly(\varepsilon)d_1)\). On the other hand, \(\prob{A} = \dist(g,g')\) and \(\Prob[v \subseteq s]{A} = \dist(g|_s,g'|_s)\) giving us the that at most \(\binom{|L|}{2}\exp(-\poly(\varepsilon)d_1)\) don't sample \(L\) well.

    Now let us make sure that most \(s \subseteq t\) also sample \(L\)'s agreement well. Indeed, let \(g \in L\). By \pref{claim:edge-sampling} there is at most \(\exp(-\poly(\varepsilon)\frac{d_1}{k})\) of the \(s \subseteq t\) such that 
    \[\Abs{A_{\zeta}(g) - A_{\zeta}(g|_s)} > \varepsilon^{100}.\]
    Another union bound gives us us the claim.
\end{proof}

The proof of \pref{prop:good-parameters} will follow from these assertions. We prove them after the proof of the proposition.
\begin{claim} \torestate{\label{claim:good-tau-and-delta-for-every-t}
    Let \(d_1 \leq d\) and fix \(t \in X(d)\) such that \(\agr_{\mathcal{D}}(\mathcal{F}_t) \geq \varepsilon - \varepsilon^{100}\). There exists \(m\leq \frac{1}{\varepsilon^{10}}\) such that the following holds for \(\tau_m = \varepsilon^2(1-\varepsilon^{50})^m\) and \(\zeta_m = 3\eta 20^m\). There exists \(L \subseteq \overline{L}_{\tau_m,\zeta_m}(t)\) such that for \(1-\exp(-\poly(\varepsilon)\frac{d_1}{k})\) of the \(u \subseteq t, u \in X(d_1)\), \(L|_u\) is \(20\zeta_m\)-dense in \(\overline{L}_{\tau_{m+1},\zeta_m}(u)\).}
\end{claim}

\begin{claim} \label{claim:good-gamma-for-every-t-and-delta}
    Let \(\zeta,\tau > 0\) and let \(t \in X(d)\). Then there exists some \(\gamma = \frac{24}{23} 24^i \zeta\) for \(i=1,2,\dots,\varepsilon^{-10}\) and some \(L \subseteq \overline{L}_{\tau,\zeta}(t)\) such that \(L\) is \(\gamma\)-dense, \(\abs{L} \leq \frac{1}{\tau^3}\) and \(23\gamma\)-separated.
\end{claim}

\begin{proof}[Proof of \pref{prop:good-parameters}]
    By \pref{claim:edge-sampling} there are \(1-\exp(-\poly(\varepsilon)\frac{d_1}{k})\) faces \(t \in X(d)\) such that \(\agr_{\mathcal{D}}(\mathcal{F}_t) \geq \varepsilon - \varepsilon^{100}\) (we use \pref{claim:edge-sampling} on the set of \(\set{r_i}\) such that \(\set{f_{r_i}}\) pass the agreement test).
    By \pref{claim:good-tau-and-delta-for-every-t}, for every such \(t \in X(d)\) there exists \[\tau = \varepsilon^2(1-\varepsilon^{50})^i, \zeta = 3\eta 20^i\] (for \(i \in \set{1,2,\dots,\frac{1}{\varepsilon^{10}}}\)), and a list \(L'(t) \subseteq \overline{L}_{\tau,\zeta}(t)\) as in \pref{claim:good-tau-and-delta-for-every-t}. 
    
    Moreover, by \pref{claim:good-gamma-for-every-t-and-delta} there exists some 
    \[\gamma' \in [\frac{24}{23} 24\zeta,\frac{24}{23} 24^{\frac{1}{\varepsilon^{10}}}\zeta]\]
    and \(\overline{L}(t) \subseteq \overline{L}_{\tau,\zeta}(t)\) such that \(\overline{L}(t)\) has length \(\leq \poly(1/\varepsilon)\), \(\overline{L}(t)\) is \(\gamma'\)-dense in \(\overline{L}_{\tau,\zeta}(t)\) and \(23\gamma'\)-separated.

    Thus by the pigeonhole principle, there exists \(\tau,\zeta,\gamma',\ell\) and a set \(N_d = N_d(\tau,\zeta,\gamma',\ell) \subseteq X(d)\) of probability \(\prob{N_d} \geq \poly(\varepsilon)\), such that for every \(t \in X(d)\) there exists \(L'(t) \subseteq \overline{L}_{\tau,\zeta}(t)\) as in \pref{claim:good-gamma-for-every-t-and-delta} and \(\overline{L}(t) \subseteq \overline{L}_{\tau,\zeta}(t)\) of size \(\Abs{\overline{L}(t)}=\ell\), \(\overline{L}(t)\) is \(\gamma\)-dense in \(\overline{L}_{\tau,\zeta}(t)\) and \(23\gamma\)-separated.

    Let us define \(M_i \subseteq X(d_i)\) be the set of \(s \in X(d)\) such that there exists some \(t \in N_d, t\supset s\) such that:
    \begin{enumerate}
        \item \(s \subseteq t\) samples \(L'(t) \cup \overline{L}(t)\) well.
        \item \(L'(t)\) is \(20\zeta\)-dense in \(\overline{L}_{\tau-\varepsilon^{100},\zeta}(s)\).
        \item If \(i=2,3\) and \(i' < i\), then \(\Prob[u \subseteq s]{M_{i'}} \leq \exp(-\poly(\varepsilon)\frac{d_1}{k}))\) (This definition is recursive but \(M_i\) is always well defined).
    \end{enumerate}
    We show that for every \(s \in M_i\) the proposition holds for \(\tau,\zeta,\ell\) and \(\gamma=2\gamma'\). We take \(\overline{L}(s)=\overline{L}(t)|_s\) for some arbitrary \(t \in N_d\) that contains \(s\) and such that the items above holds for this \(t\). By the fact that \(\overline{L}(t)\) is sampled well we have that \(\overline{L}(s) \subseteq \overline{L}_{\tau-\varepsilon^{100},\zeta}(t)\). Moreover if \(\overline{L}(t)\) is \(23\gamma'=11.5\gamma\)-separated and in particular it holds that \(11\gamma'\)-separated (and in particular, all restrictions indeed result in distinct functions, i.e. \(\Abs{\overline{L}(s)}=\Abs{\overline{L}(t)|_s}=\ell\)). 
    
    Next we show density in \(\overline{L}_{\tau-\varepsilon^{100},\zeta}(s)\). Note that \(L'(t)|_s\) is \(20\zeta\)-dense in \(\overline{L}_{\tau-\varepsilon,\zeta}(s)\), and \(\overline{L}(t)\) is \(\gamma'\)-dense in \(L'(t)\) (not to say that this is a subset of \(L'(t)\), just that for every \(g' \in L'(t)\) there exists some \(g \in \overline{L}(t)\) such that \(g' \overset{1-\gamma'}{\approx} g\)). Thus by the fact that \(s\) samples \(L'(t) \cup \overline{L}(t)\) well, this implies that \(\overline{L}(t)\) is \(\gamma'+\varepsilon^{100}\)-dense in \(L'(t)\). Hence, it holds that \(\overline{L}(t)\) is \(\gamma'+\varepsilon^{100}+20\zeta\leq \gamma\)-dense in \(\overline{L}_{\tau-\varepsilon^{100},\zeta}(s)\).

    Finally, we show that when \(i=2,3\) and \(i' \leq i\), it holds that
    \(P_s(i')\) from the definition of the proposition has size \(1-\exp(-\Omega(\poly(\varepsilon)\frac{d_1}{k}))\). Let us go one by one:
    \begin{enumerate}
        \item For every \(t\), the fraction of \(u \subseteq s\) that don't sample \(\overline{L}(s)\) well is \(\exp(-\Omega(\poly(\varepsilon)\frac{d_1}{k}))\) by \pref{claim:most-sets-sample-well}. %
        \item The probability that \(u \subseteq s\) is in \(M_{i'}\) is  \(\exp(-\Omega(\poly(\varepsilon)\frac{d_1}{k}))\). Let us show that in this case \(\overline{L}(t)|_s\) is \(2\gamma\)-dense. \(u \in M_{i'}\), therefore there is some \(\overline{L}(u)\) of size \(\ell\) that is \(\gamma\)-dense in \(\overline{L}_{\tau-\varepsilon^{100},\zeta}(u)\), and \(11\gamma\)-separated (by what we already proved above). Thus if we show that for every \(g \in \overline{L}(u)\) there is some \(g' \in \overline{L}(t)|_s\) such that \(g \overset{1-\gamma}{\approx} g'\) it will follow that \(\overline{L}(t)|_s\) is \(2\gamma\)-dense. By \pref{claim:injective-function}, the fact that \(\overline{L}(u)\) is dense in \(\overline{L}(t)|_s\) and the fact that both are \(2\gamma\)-separated, there is an injective function \(\phi:\overline{L}(t)|_s \to \overline{L}(u)\) where \(\phi(g') = g\) if and only if \(g \overset{1-\gamma}{\approx} g'\). As both lists have the same size, this shows that this is a surjection, i.e. that for every \(g \in \overline{L}(u)\) there is some \(g' \in \overline{L}(t)|_s\).
    \end{enumerate}
    
    To conclude, let us show that \(\prob{M_i} \geq 1- \exp(-\Omega(\poly(\varepsilon)\frac{d_1}{k})\). Let us begin with \(M_1\). Let \(B_1\) be the fraction of \(s\)'s such that \(\Prob[t \supseteq s]{N_d} \leq \poly(\varepsilon)\). Its fraction is \(\prob{B_1} = O(\frac{d_1}{d\poly(\varepsilon)}) = \exp(-\poly(\varepsilon)\frac{d_1}{k})\) by \pref{cor:good-sampler-from-expansion}. Let \(B_2 \subseteq X(d_i)\) be the event that more than \(\frac{1}{3}\) of the \(t \supseteq s, t\in N_d\) have:
    \begin{enumerate}
        \item Either \(s\) doesn't sample \(\overline{L}(t) \cup L'(t)\) well.
        \item \(L'(t)\) is \(20\zeta\)-dense in \(\overline{L}_{\tau-\varepsilon^{100},\zeta}(s)\).
    \end{enumerate}
    On the one hand, by \pref{claim:most-sets-sample-well} and \pref{claim:good-tau-and-delta-for-every-t} there are at most \(\exp(-\poly(\varepsilon)\frac{d_1}{k})\) such pairs. On the other hand, every \(s \in B_2 \setminus B_1\) contributes \(\poly(\varepsilon)\) such pairs. By Markov's inequality, \(\prob{B_2 \setminus B_1} = \exp(-\poly(\varepsilon)\frac{d_1}{k})\). Hence, if \(s \notin B_1 \cup B_2 = B_1 \cup (B_2 \setminus B_1)\), then \(s \in M_1\) so \(\prob{M_1} =  1 - \exp(-\poly(\varepsilon)\frac{d_1}{k})\).

    Continuing with \(M_2\). The same arguments for \(M_1\) show that \(1-\exp(-\poly(\varepsilon)\frac{d_1}{k})\) of the \(s \in X(d_2)\) sample \(\overline{L}(t)\cup L'(t)\) well and that \(L'(t)\) is \(20\zeta\)-dense in \(\overline{L}_{\tau-\varepsilon^{100},\zeta}(s)\). So we concentrate on the last property. The probability that \(s \in X(d_2)\) contains more than \(\\prob{M_1}^{1/2}\)-fraction of \(u \notin M_1\) is at most \(\left ( \prob{M_1} \right )^{1/2}\) (by Markov's inequality), and this is also \(\exp(-\poly(\varepsilon)\frac{d_1}{k})\), albeit with a slightly worse \(\poly(\varepsilon)\) than before, since since \(\prob{M_1}=\exp(-\poly(\varepsilon)\frac{d_1}{k})\). Hence \(\prob{M_2} =1-\exp(-\poly(\varepsilon)\frac{d_1}{k})\). The proof for \(M_3\) is similar.
\end{proof}

\begin{proof}[Proof of \pref{claim:good-gamma-for-every-t-and-delta}]
    Consider the following method, resembling \cite{DinurHKLT2018}. Let \(L_1 \subseteq \overline{L}_{\tau,\zeta}(t)\) be a maximal \(24\zeta\)-separated list. Consider the following chain \(L_1 \supseteq L_2 \supseteq L_3 \dots\) where every \(L_i\) is a maximal \(24^i \zeta\)-separated set inside \(L_{i-1}\). Obviously, the sizes of these sets are monotonically decreasing (but they are always non-empty), therefore there must be some \(j\) such that \(L_j = L_{j+1}\). Take as \(L=L_i\) for \(i\) being the first such \(j\).

    Recall that \(\gamma_i = \frac{24}{23} 24^i\). By definition, \(L_{i+1}\) is \(23\gamma_i = 24^{i+1}\zeta\)-separated. Let us see that is it also \(26\cdot 24^i \zeta\)-dense in \(\overline{L}_{\tau,\zeta}(t)\). First note that every \(L_j\) is \(24^i \zeta\)-dense in \(L_{j-1}\) (if it weren't, we could have added to \(L_j\) another \(g \in L_{j-1}\setminus L_j\) that is far from the current \(L_j\), contradicting maximality). Thus for every \(g \in \overline{L}_{\tau,\zeta}(t)\) there is a sequence \(g=g_0,g_1,g_2,\dots, g_i\) where \(g_j \in L_j\) for \(j>0\), such that \(\dist(g_{j-1},g_j)\leq 24^j \zeta\). Thus concluding that
    \[\dist(g,g_i) \leq \sum_{j=1}^i 24^j  \zeta = \frac{24}{23}(24^i-1)\zeta \leq \frac{24}{23} \cdot 24^i = \gamma_i.\]

    Finally, we must show that \(i \leq \frac{1}{\tau^3}\). If we show that \(\Abs{L_1} \leq \frac{1}{\tau^3}\) this will follow, since \(\Abs{L_i}\geq 1\) and for \(j \leq i\), \(\Abs{L_j}\leq \Abs{L_{j-1}}-1\). We distinguish this part of the proof to a separate claim since we will need it in also later on.
    \begin{claim} \label{claim:lists-are-short}
        let \(\tau \gg \exp(-\Omega(\zeta^2 k))\) and let \(\zeta\) and \(\gamma \geq 3\zeta\). Let \(L \subseteq \overline{L}_{\tau,\zeta}(t)\) be a \(\gamma\)-separated set. Then \(\Abs{L}\leq\frac{1}{\tau^3}\).
    \end{claim}
\end{proof}

\begin{proof}[Proof of \pref{claim:lists-are-short}]
    Assume towards contradiction that \(|L_1|>\frac{1}{\tau^3}\) and take some subset \(L'\subseteq L_1\) of size \(\frac{1}{\tau^3}+1\). As it is at least \(3\zeta\)-separated, for every distinct \(g,g' \in L'\), the fraction of \(r \in X(k)\) such that \(\dist(g|_r,g'|_r) \leq 2\zeta\) (i.e. \(g|_r \overset{1-2\zeta}{\approx} g'|_r\)) is at most \(\exp(-\Omega(\zeta^2 k))\leq \tau^{10}\). This shows that there are almost no intersection between the supports of \(g\) and \(g'\), since \(\supp_\zeta (g) \cap \supp_{\zeta}(g') \subseteq \sett{r \in X(k)}{\dist(g|_r,g'|_r) \leq 2\zeta}\). That is, \(\prob{\supp_\zeta (g) \cap \supp_{\zeta}(g')} \leq \exp(-\Omega(\zeta^2 k))\leq \tau^{10}\).
    
    Denote by 
    \[D(g) = A_{\zeta}(g) \setminus \bigcup_{g' \in L', g' \ne g}A_{\zeta}(g').\]
    Then in particular 
    \begin{align}
        \prob{D(g)} &\geq \prob{A_{\zeta}(g)} - \sum_{g' \in L', g' \ne g}\prob{A_\zeta (g) \cap A_\zeta (g')} \\
        &\geq \prob{A_{\eta}(g)} - \sum_{g' \in L', g' \ne g} 2 \prob{\supp_{\zeta}(g) \cap \supp_{\zeta}(g')} \\
        &\geq \tau - 2\tau^6.
    \end{align}
    On the other hand, the sets \(\set{D(g)}_{g \in L'}\) are mutually disjoint, namely
    \(1 \geq \sum_{g \in L'} \prob{D(g)} \geq \frac{\tau-2\tau^6}{\tau^3+1}>1\) which is a contradiction.
\end{proof}

\subsection{Proof of \pref{claim:good-tau-and-delta-for-every-t}}
We start with a weaker claim than \pref{claim:good-tau-and-delta-for-every-t}, saying that \(\overline{L}_{\tau^3,\zeta}(t)\) restricted to \(s \subseteq t\) is dense in \(\overline{L}_{\tau,\zeta}(s)\) for most \(s\). This claim is where we use the agreement soundness of the original distribution.

\begin{claim} \label{claim:density-of-smaller-tau}
    Fix \(t\),\(\tau\), \(\zeta\) and \(\gamma \geq 3\zeta\). Let \(L \subseteq \overline{L}_{\tau^3,\zeta}(t)\) be a \(\gamma\)-dense and \(\gamma\)-separated list. Then
    \begin{equation*}
        \Prob[s \subseteq t,s \in X(d_1)]{L|_s \text{ is \((\gamma+4\zeta)\)-dense in } \overline{L}_{\tau,\zeta}(s)} \geq 1-\exp(-\Omega(\frac{d_1}{k})).
    \end{equation*}
\end{claim}

\begin{proof}[Proof of \pref{claim:density-of-smaller-tau}]
    The claim will follow from these assertions.
    \begin{enumerate}
        \item If \(L\) is \(\gamma\)-dense and has size \(\poly(1/\varepsilon)\) in \(\overline{L}_{\tau^3,\zeta}(t)\) then 
        \[\Prob[r_1,r_2 \subseteq t, r_1,r_2 \sim D]{f_{r_1}=f_{r_2} \ve r_1,r_2 \notin \bigcup_{g \in L}\supp_{\gamma+2\zeta}(g)} \leq 2\tau^{1.5}.\]
        \item This quantity is sampled well. That is, let \(P = \sett{(r_1,r_2)}{f_{r_1}=f_{r_2} \ve r_1,r_2 \notin \bigcup_{g \in L}\supp_{\gamma+2\zeta}(g)}\). Then for \(1-\exp(-\Omega(\poly(\varepsilon)\frac{d_1}{k}))\) of the \(s \subseteq t\) it holds that
        \[\Abs{\Prob[r_1,r_2 \subseteq s]{P} - \prob{P}} \leq \varepsilon^{100}.\]
    \end{enumerate}
    Let us explain why. If there exists some \(h \in \overline{L}_{\tau,\zeta}(s)\) such that \(\dist(h,g) \geq \gamma+4\zeta\) for all \(g \in L|_s\), then \(A_{\zeta}(h) \geq \tau\) and moreover it holds that \(\prob{\supp_\zeta(h) \cap \supp_{\gamma+2\zeta}(g)} = \exp(-\poly(\varepsilon)k) \ll \poly(\varepsilon)\) (since whenever \(f_r\overset{1-\zeta}{\approx}h|_r\) and \(f_r\overset{1-\gamma-2\zeta}{\approx}g|_r\) then it holds that \(\dist(g|_r,h|_r)\leq \gamma+3\zeta\) which happens with \(\exp(-\poly(\varepsilon)k)\) probability at most). By \pref{claim:lists-are-short}, \(|L|=\poly(1/\varepsilon)\) so all these intersections are negligible. In particular, more than (say) half of the edges \((r_1,r_2) \in A_{\zeta}(h)\) have that \(r_1,r_2 \notin \bigcup_{g \in L}\supp_{\gamma+3\zeta+\varepsilon^{100}}(g)\). Thus \(\Prob[r_1,r_2 \subseteq s, r_1,r_2 \sim D]{P} \geq \frac{1}{2} A_\zeta(h) \geq \frac{1}{2}\tau\). On the other hand, by the first item \(\prob{P} \leq 2\tau^{1.5}\). By the second item \(\Prob[r_1,r_2 \subseteq s]{P} \leq \prob{P} + \varepsilon^{100}\leq 3\tau^{1.5}\) with probability as high as \(\exp(-\poly(\varepsilon)\frac{d_1}{k})\) and hence this cannot occur for more than \(\exp(-\poly(\varepsilon)\frac{d_1}{k})\).

\medskip

    The second item follows directly from \pref{claim:edge-sampling}. The main effort is to show the first item. Assume that it doesn't hold. Define \(\mathcal{F}'_t\) by erasing \(f'_r\) for every \(r \in \bigcup_{g \in L}\supp_{\gamma+3\zeta+\varepsilon^{100}}(g)\) (and taking \(f'_r=f_r\) for the rest of the faces)\footnote{Formally, when we ``erase'' a function, we actually set the \(f'_r\) to be such that it agrees with no other functions. This can be done easily by extending the alphabet \(\Sigma\) to a large enough set.}. Then it still holds that 
     \[\agr_{\mathcal{D}}(\mathcal{F}'_t) \geq 2\tau^{1.5}.\]
     By the agreement soundness (assuming that \(\tau^{1.5}\geq \varepsilon_0\)), there is some \(h \in \overline{L}_{\tau^3,\eta}(t)\) such that \(A_{\eta}^{\mathcal{F}'}(h) \geq \tau^3\). On the other hand, there must by some \(g \in L\) such that \(\dist(h,g) \leq \gamma\). This implies that \(\prob{\supp_{\eta}^{\mathcal{F}'}(h) \setminus \supp_{\gamma+2\zeta}(g)} \leq \exp(-\poly(\zeta^2 k)) \ll \tau^3\). This is a contradiction since \(A_{\eta}^{\mathcal{F}'}(h) \subseteq \supp_{\eta}^{\mathcal{F}'}(h) \setminus \supp_{3\zeta+\gamma}(g)\) (as we erasing \(g\)'s support).
\end{proof}

\restateclaim{claim:good-tau-and-delta-for-every-t}

\begin{proof}[Proof of \pref{claim:good-tau-and-delta-for-every-t}]
    The proof logic is the following. We construct a sequence of lists \(L_1,L_2,\dots\) where \(L_1 \subseteq \bar{L}_{\tau_1,\zeta_1}(t)\), \(L_2 \subseteq \bar{L}_{\tau_2,\zeta_2}(t)\) and so on. If \(L_m \subseteq \bar{L}_{\tau_m,\zeta_m}(t)\) has the property required in the claim, i.e. that for \(1-\exp(-\poly(\varepsilon) \frac{d_1}{k})\) of the \(u \subseteq t\), \(L_m |_u\) is \(20\zeta_m\) dense inside \(\bar{L}_{\tau_{m+1},\zeta_m}(u)\) we conclude; else we continue to the next list \(L_{m+1}\). What we need to make sure is that there is an \(m \leq \frac{1}{\varepsilon^{10}}\) such that this process concludes with \(L_m\). To do so, we define a potential function \(p:\set{1,2,\dots, \frac{1}{\varepsilon^{10}}} \to [0,1]\) that is monotone non-decreasing (and \(p \leq 1\)), and show that if we don't conclude at step \(m\), then \(p(m+1) \geq p(m) + \varepsilon^{10}\), i.e.\ one of the first \(\frac{1}{\varepsilon^{10}}\) first lists must satisfy the property else \(p(\frac{1}{\varepsilon^{10}}) > 1\). The potential function is 
    \begin{equation} \label{eq:potential-function}
        p(m) = \Prob[r_1,r_2 \subseteq t, r_1,r_2 \sim D]{\bigcup_{g_j \in L_m} A_{9\zeta_{m-1}}(g_j)}.
    \end{equation}
    That is, \(p(m)\) measures how much agreement is explainable by \(L_m\), and we show that while the list \(L_m\) is not dense for most \(u \subseteq t\), then in fact we can find a list \(L_{m+1}\) whose functions `explains' more agreement. Details follow.
    
    Recall that \(\tau_m = \varepsilon^2 (1-\varepsilon^{10})^m\) and that \(\zeta_m = 3\eta \cdot 20^m\). Fix some \(t\). For every \(\tau_m,\zeta_m\) we choose some (arbitrary) \(L^{aux}_{m}(t) \subseteq \overline{L}_{\tau_{m}^3,\zeta_m}(t)\) that is \(3\zeta\)-dense and \(3\zeta\)-separated. We will sequentially construct lists \(L_m=\set{g_1,g_2,\dots,g_m} \subseteq \overline{L}_{\tau_{m},\zeta_{m}}(t)\), for \(m=1,2,\dots\) until \(L_m\) satisfies the claim. The potential function \(p\) shall be defined as in \pref{eq:potential-function} with respect to these lists, once we construct them. %

    For \(m=0\) we take \(L_0 =\emptyset \subseteq \overline{L}_{\tau_0,\zeta_0}(t)\). Let us describe how to construct \(L_m\) using \(L_{m-1}\), provided that \(L_{m-1}\) is not \(\zeta_m=20\zeta_{m-1}\)-dense in \(\overline{L}_{\tau_{m},\zeta_{m-1}}(u)\) for \(1-\exp(-\poly(\varepsilon)\frac{d_1}{k})\) of the \(u \in X(d_1)\):
        \begin{enumerate}
            \item Find some \(u \subseteq t\) such that:
            \begin{itemize}
                \item \(L_{m-1}\) is not \(20\zeta_{m-1}\)-dense in \(\overline{L}_{\tau_{m},\zeta_{m-1}}(u)\).
                \item \(u\) samples \(L^{aux}_{m} \cup L_{m-1}\) well.
                \item \(L^{aux}_{m}|_u\) is \(7\zeta_{m}\)-dense in \(\overline{L}_{\tau_{m},\zeta_{m-1}}(u)\). 
                \item For every \(g \in L^{aux}_{m}(t)\) it holds that \(\Abs{A_{9\zeta_{m-1}}^{\mathcal{F}_t}(g) - A_{9\zeta_{m-1}}^{\mathcal{F}_u}(g|_u) }\leq \varepsilon^{100}\).\footnote{The fraction of \(u \subseteq t\) such that these items do not hold is \(\exp(-\Omega(\poly(\varepsilon)\frac{d_1}{k}))\) by \pref{claim:edge-sampling} and \pref{claim:density-of-smaller-tau}. Hence if \(L_m\) is not \(20\zeta_{m-1}\)-dense for for \(1-\exp(-\Omega(\poly(\varepsilon)\frac{d_1}{k}))\) we can find a \(u \subseteq t\) such that this holds.}
            \end{itemize} 
            \item Find some \(h \in \overline{L}_{\tau_{m},\zeta_{m-1}}(u)\) that is \(20\zeta_{m-1}\)-far from every element in \(L_{m-1}\). There exists such an \(h\) since \(L_{m-1}\) is not \(20\zeta_{m-1}\)-dense in \(\overline{L}_{\tau_{m},\zeta_{m-1}}(u)\).
            \item Take some \(g=g_{m} \in L^{aux}_{m}\) such that \(\dist(g|_u,h_u) \leq 7\zeta_{m-1}\). Set \(L_{m} = L_{m-1} \cup \set{g_{m}}\). This step is possible since \(L^{aux}_{m}(t)|_u\) is \(7\zeta_{m}\)-dense inside \(\overline{L}_{\tau_{m},\zeta_{m}}(u)\).
        \end{enumerate}
    Observe that because \(\dist(g_{m}|_u,h) \leq 7\zeta_{m-1}\) then the fraction of \(r \in X(k)\) such that \(\dist(g_{m}|_r,h|_r) > 8\zeta_{m-1}\) has probability \(\exp(-\Omega(\zeta^2 k)) \ll \varepsilon^{100}\). In addition, when \(f_r \overset{1-\zeta_m}{\approx} h_r\) and \(h_r \overset{1-8\zeta_{m-1}}{\approx} g_{m}|_r\) then \(h_r \overset{1-9\zeta_{m-1}}{\approx} g_{m}\). Hence it holds that \(\prob{A_{9\zeta_{m-1}}(g_{m}|_u)} \geq \prob{A_{\zeta_{m-1}}(h)} - \varepsilon^{100}\). By assumption on \(u\) it holds that \(\prob{A_{9\zeta_{m-1}}(g_{m})} \geq \prob{A_{9\zeta_{m-1}}(g_{m}|_u)} -\varepsilon^{100}\) so
    \[\prob{A_{9\zeta_{m-1}}(g_{m})} \geq \prob{A_{9\zeta_{m-1}}(h)} - 2\varepsilon^{100} \geq \tau_{m-1}-2\varepsilon^{100}\geq \tau_{m}.\]
    That is, \(g_{m} \in \overline{L}_{\tau_{m},\zeta_{m}}(t)\) (and thus every \(L_m \subseteq \overline{L}_{\tau_{m},\zeta_{m}}(t)\)).

    Moreover, note that for every \(j \leq m\) it holds that \(\dist(g_{m}|_u,g_j|_u)\geq \dist(h,g_j|_u)-\dist(g_{m}|_u,h) \geq 13\zeta_{m}\). The face \(u\) samples \(L^{aux}_{m} \cup L_{m-1}\) well, hence \(\dist(g_{m},g_j) \geq \dist(g_{m}|_u,g_j|_u) - \zeta_{m-1} \geq 12\zeta_{m-1}\). This is much larger than \(9\zeta_{m-1} + \zeta_{j-1}\). Hence (as before) it holds that \(\prob{\supp_{9\zeta_{j-1}}(g_j) \cap \supp_{9\zeta_{m-1}}(g_{m})} \ll \varepsilon^{100}\). This implies that \(\prob{A_{9\zeta_{j-1}}(g_j) \setminus \bigcup_{j' \ne j} A_{9 \zeta_{j'-1}}(g_{j'})} \geq \frac{1}{2} \prob{A_{9\zeta_{j-1}}(g_j)}\).
    
    Thus it holds that while \(m \leq \varepsilon^{-30}\), then \(p(m) \geq \frac{1}{2}\sum_{j=1}^m \prob{A_{9\zeta_{j-1}}(g_j)} \geq \frac{m}{2}\varepsilon^{10}\). For \(m > 2\varepsilon^{10}\) this is greater than \(1\) hence the process must stop beforehand.
\end{proof}

%% file: sections/cover-construction.tex
\section{Constructing a cover} \label{sec:construction-of-cover}
\begin{figure}
    \centering
    \includegraphics[scale=0.4]{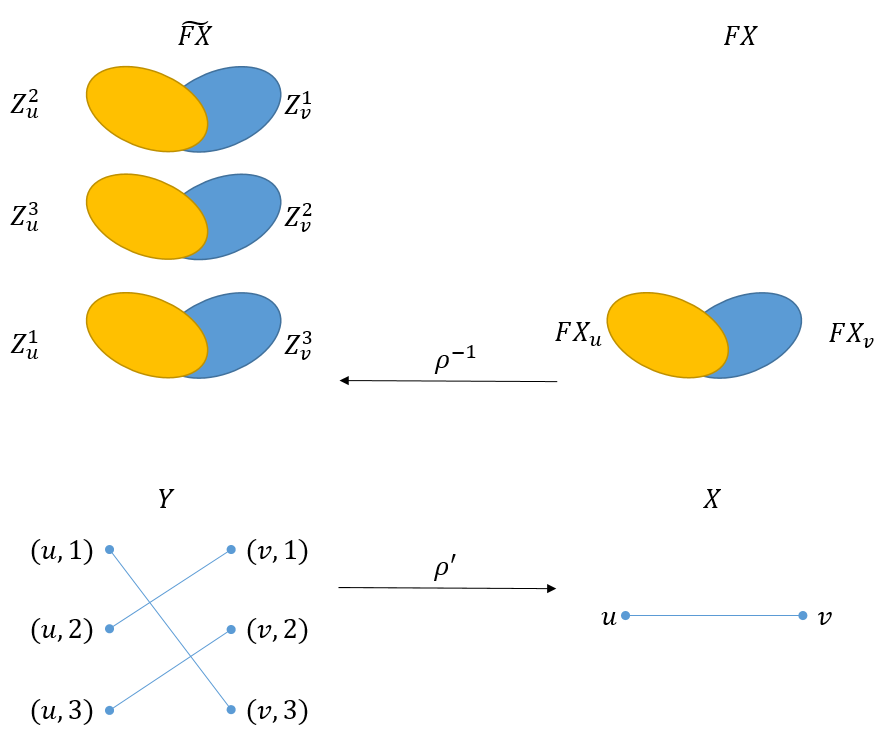}
    \caption{Constructing a cover of \(X\) from a cover of \(\FX{d_1}X\)}
    \label{fig:cover}
\end{figure}

In this subsection we prove \pref{lem:Y-is-an-honest-to-god-cover-of-X}. We do so in two parts: first we construct a simplicial complex \(Y\) and prove that it covers \(X\). In the second part we explicitly construct the isomorphism \(\iota:\FX{d_1}Y \to \widetilde{\FX{d_1}X}\). In addition to proving \pref{lem:Y-is-an-honest-to-god-cover-of-X}, we will use the explicit definition of \(\iota\) for proving \pref{claim:agreement-with-majority} later on.

Let \(X\) be a \(d\)-dimensional clique complex. Let \(d_1\) be such that \(3d_1 \leq d-2\). Let \(\FX{d_1}X=F(X,d_1)\) be as in \pref{def:face-complex} and assume that \(X\) is well connected as in \pref{def:well-connected-complex}. Fix \(\nu:\widetilde{\FX{d_1}X}\to \FX{d_1}X\) to be an \(\ell\)-cover. In this section, we will show how to construct an $\ell$-cover of \(\rho:Y \to X\) such that \(\widetilde{\FX{d_1}X}\cong \FX{d_1}Y\). 

The main idea will be to encode the vertices \(v \in X(0)\) by the links \(\FX{d_1}X_v \subseteq \FX{d_1}X\) (these are the \(\FX{d_1}X_v = F(X_v,d_1)\) as in \pref{def:face-complex} which are contained as sub complexes inside \(\FX{d_1}X\)). We observe that this encoding has the property that if \(v \sim u\) then \(\FX{d_1}X_v \cap \FX{d_1}X_u = \FX{d_1}X_{uv}\). We will use these non-empty intersections to define the edges in the cover $Y$. For every \(r \in X\), let \(Z_r = \nu^{-1}(\FX{d_1}X_r)\subset \widetilde {\FX{d_1}X}\). Using the well-connectedness of \(X\), we will show that
\begin{claim} \label{claim:Z_v-decomposes}~
    \begin{enumerate}
        \item For every non empty \(r \in X^{\leq d_1}\), \(Z_r\) decomposes to exactly \(\ell\) connected-components\footnote{When we say ``connected components'' we mean that the induced complex \(Z_r\) has these \(\ell\) connected components. Of course it may hold that there are paths from \(Z_i\) to \(Z_j\) in \(\widetilde{\FX{d_1}X}\), but they must go through some vertex \(v \in \widetilde{\FX{d_1}X} \setminus Z_r\).} \(Z_r^1,Z_r^2,...,Z_r^\ell \subseteq \widetilde{\FX{d_1}X}\), such that \(\nu\) is an isomorphism between each component and \(\FX{d_1}X_r\), \(\nu|_{Z_r^i}:Z_r^i \overset{\sim}{\to} \FX{d_1}X_r\). 
        \item For every non empty \(r,s \in X^{\leq d_1}\) such that \(r \subseteq s\) there is a permutation \(\pi_{r,s}:[\ell] \to [\ell]\) such that \(Z_{s}^{i} \subseteq Z_{r}^{j}\) if and only if \(\pi_{r,s}(i)=j\).
    \end{enumerate} 
\end{claim}
We prove this claim at the end of this subsection. We define the cover \(Y\) of \(X\) as follows. 
Let \(\psi \in C^1(X,Sym(\ell))\) be
\begin{equation} \label{eq:cocycle-def}
    \psi(vu)=\pi_{uv,u}^{-1}\circ \pi_{uv,v}
\end{equation}
where \(\pi_{uv,u},\pi_{uv,v}\) are the permutations promised by the second item of \pref{claim:Z_v-decomposes}. Let \(Y=X^\psi\), i.e. \(Y\) is the clique complex such that \(Y(0)=X \times [\ell]\) and \(\set{(v,i),(u,j)} \in Y(1)\) if \(vu \in X(1)\) and \(j=\psi(vu).i\). Then
\begin{lemma} \label{lem:Y-is-a-cover-of-X}
    \(Y\) is a cover of \(X\). Moreover, it holds that \((v,i) \sim (u,j)\) if and only if \(vu \in X(1)\) and there exists some \(k\) such that \(Z_v^i \cap Z_u^j \supseteq Z_{vu}^k\).
\end{lemma}

For the proof of \pref{lem:Y-is-a-cover-of-X} we use the fact that the \(\pi_{r,s}\) in \pref{claim:Z_v-decomposes} give us a cocycle on the flag complex of \(\FX{d_1}X\) (see below). For this we recall the following corollary: 
\restatecorollary{cor:cocycles-in-flag-complex}
In this notation \(\phi(u,uv)=\pi_{uv,u}\) and thus \(\psi(uv)=\pi_{uv,u}^{-1}\circ \pi_{uv,v}\).

\begin{proof}[Proof of \pref{lem:Y-is-a-cover-of-X}]
    Let us show that \(Y\) is a cover. This is equivalent to showing that \(\psi\) is a cocycle. Let \(X^{\leq 2}\) be the two skeleton of \(X\) and let \(G\) be its flag-complex as in \pref{def:flag-complex}. Let \(\phi:GX(1) \to Sym(\ell)\) be given by \(\phi(r,s)= \pi_{s,r}\) for every \(r \subseteq s\) (where \(\pi_{s,r}\) are the permutations in \pref{claim:Z_v-decomposes}). If we show that \(\phi \in Z^1(GX,Sym(\ell))\) then by \pref{cor:cocycles-in-flag-complex} \(\psi\) is also a cocycle. This amounts to show that for every \(uvw \in X\) it holds that \(\pi_{uvw,uv}\circ \pi_{uv,u} = \pi_{uvw,u}\).

    Indeed, fix \(i\) and let \(k=\pi_{uvw,u}(i)\), and \(k'= \pi_{uvw,uv}\circ \pi_{uv,u}(i)\). By definition this implies that \(Z_{uvw}^k \subseteq Z_u^i\) and that \(Z_{uvw}^{k'} \subseteq Z_{uv}^{\pi_{u,uv}(i)} \subseteq Z_u^i\). The second item of \pref{claim:Z_v-decomposes} says that \(Z_{uvw}^{k'} \subseteq Z_u^i\) implies that  \(k' = \pi_{u,uvw}(i) = k\). 
    
    Next, for the ``moreover'' statement, note that \((v,i) \sim (u,j)\) if and only if there is some \(k\) such that \(\pi_{v,uv}(i)=k\) and such that \(\pi_{u,uv}(j)=k\). This occurs if and only if \(Z_{uv}^k \subset Z_u^i \cap Z_v^j\).
\end{proof}

\begin{proof}[Proof of \pref{claim:Z_v-decomposes}]
Let us start with the first item. We begin by explaining why it is enough to prove it assuming \(r\) is a vertex. Suppose that for every vertex \(v \in X(0)\) it holds that \(Z_v\) decomposes to \(\ell\) connected-components. Let \(r \in X^{\leq d_1}\) and let \(v \in r\) be an arbitrarily chosen vertex inside \(r\). Take \(Z_r^i\) to be the subcomplex induced by \(Z_r^i(0) := Z_v^i(0) \cap \nu^{-1}(\FX{d_1}X_r(0))\). As every \(Z_v^i\) is isomorphic to \(\FX{d_1}X_v\) then in particular, the induced sub-complex \(Z_r^i\) is isomorphic to \(\FX{d_1}X_{r}\). By well connectedness of \(X\), \(\FX{d_1}X_r\) is connected and thus we get that these are \(\ell\) connected-components (they cannot be connected to one another because each lies in a different \(Z_v^i\)).

Now we show the first item for every vertex \(v \in X(0)\). Towards this end, we note that the restriction \(\nu|_{Z_v}:Z_v \to \FX{d_1}X_v\) is also a covering map (for a proof of this statement, see \pref{claim:restriction-of-cover}). From the well connectivity of \(X\), \(\FX{d_1}X_v\) is simply connected, so \(Z_v\) must decompose to \(\ell\) disjoint connected-components. 

We move to the second item. Fix \(r \subseteq s\). We need to show that for every \(i \in [\ell]\) there is a unique \(j\) such that \(Z_s^j \subseteq Z_r^i\), so let us fix some \(i \in [\ell]\). By the first item, the restriction \(\nu|_{Z_r^i}\) is an isomorphism between \(Z_r^i\) and \(\FX{d_1}X_r\). In particular, it follows from the definition of an isomorphism that \(\nu^{-1}(\FX{d_1}X_s) \cap Z_r^i\) is isomorphic (by \(\nu|_{Z_r^i}\)) to \(\FX{d_1}X_s\subset \FX{d_1}X_r\), and thus \(Z_s^j\) is one of the connected components \(Z_s^j \subseteq \nu^{-1}(\FX{d_1}X_s)\) promised by the first item we already proved. This is the index \(j\) we needed. Note that it is unique since \(Z_s^j = \nu^{-1}(\FX{d_1}X_s) \cap Z_r^i\) (so there are no more faces in \(\nu^{-1}(\FX{d_1}X_s)\) and  \(Z_r^i\) that are not already in \(Z_s^j\)).
\end{proof}

\subsection[The isomorphism between FY to the cover of FX]{The isomorphism between \(\FX{d_1}Y\) to \(\widetilde{\FX{d_1}X}\)}
\begin{figure}
    \centering
    \includegraphics[scale=0.6]{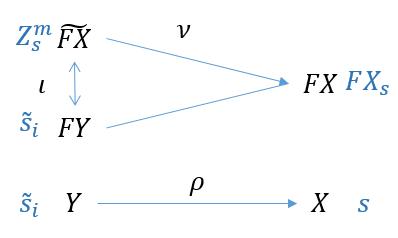}
    \caption{Commutative diagram}
    \label{fig:commutative}
\end{figure}
\begin{figure}
    \centering
    \includegraphics[scale=0.4]{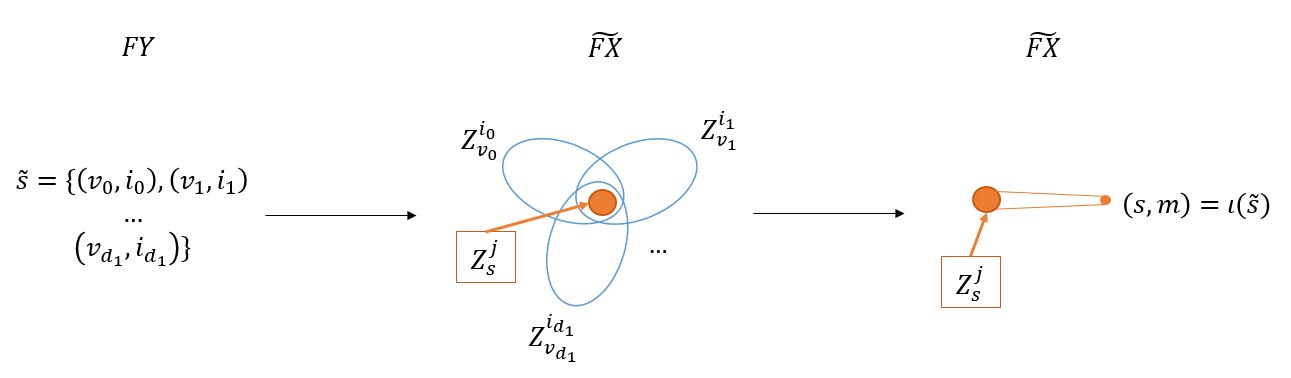}
    \caption{From \(\tilde{s}\) to \(\iota(\tilde{s})=(s,m)\)}
    \label{fig:iota-description}
\end{figure}
We continue with the notation in the previous subsection. In this subsection we describe the isomorphism \(\iota:\FX{d_1}Y \to \widetilde{\FX{d_1}X}\) explicitly. See \pref{fig:commutative} for the relations between \(X,Y,\FX{d_1}X,\FX{d_1}Y\) and \(\widetilde{\FX{d_1}X}\).

The idea is simple. Given some \(\tilde{s} = \set{(v_0,i_0),(v_1,i_1),...,(v_{d_1},i_{d_1})} \in \FX{d_1}Y(0)\) , such that \(\rho(\tilde{s})=s = \set{v_0,v_1,\ldots,v_{d_1}}\), 
we need to specify which \(j \in [\ell]\) is such that \(\iota(\tilde{s})=(s,j)\). Using arguments similar to \pref{lem:Y-is-a-cover-of-X}, we will show that the face \(\tilde{s}\) is such that \(\bigcap_{m=0}^{d_1} Z_{v_m}^{i_m}\) contains a ``copy'' \(Z_s^j\) of \(\FX{d_1}X_s\).
We will define \(\iota(\tilde{s})=(s,m)\) which is connected to this copy. See \pref{fig:iota-description} for a pictorial description of \(\iota\). We first claim that \(\iota\) is well defined.
\begin{claim} \label{claim:unique-link-of-s}
    Let \(\tilde{s} \in \FX{d_1}Y(0)\),  \(\tilde{s} = \set{(v_0,i_0),(v_1,i_1),...,(v_{d_1},i_{d_1})}\) such that \(\rho(\tilde{s})=s\). Then there is a unique \(j \in [\ell]\) such that \(Z_s^j \subseteq \bigcap_{m=0}^{d_1} Z_{v_m}^{i_m}\).
\end{claim}

We prove this claim below. Assuming it is true, let us show that this is actually an isomorphism.
\begin{lemma}\label{lem:iota-is-an-isomorphism}
    \(\iota:\FX{d_1}Y \to \widetilde{\FX{d_1}X}\) is an isomorphism.
\end{lemma}

\begin{proof}[Proof of \pref{lem:iota-is-an-isomorphism}]
    Let us start by explaining why this is a bijection. First note that for every \(s \in \FX{d_1}X(0)\), \(\iota\) maps the \(\ell\)-elements of \(\rho^{-1}(s) \subseteq \FX{d_1}Y(0)\) to the \(\ell\)-elements of \(\nu^{-1}(s) \subseteq \widetilde{\FX{d_1}X}\). Thus if we show that \(\iota\) is injective, it will also follow that it is a bijection. Let \(\tilde{s}_1, \tilde{s}_2 \in \rho^{-1}(s)\) be distinct elements and we denote \(\iota(\tilde{s}_i)=(s,m_i)\). We'll show that \(m_1 \ne m_2\). Take some vertex \(v \in X(0)\) in \(s \in \FX{d_1}X(0)\), that is \(v \in s=\rho(\tilde{s}_i)\). In particular there are two distinct vertices \((v,i_1),(v,i_2) \in Y(0)\) such that \((v,i_1) \in \tilde{s}_1\) and \((v,i_2) \in \tilde{s}_2\). Note that \(\iota(\tilde{s}_i)\) is the vertex \((s,m_i)\) that is connected to the \(Z_s^{j_i}\) that sits in the intersection \(\bigcap_{(u,k) \in \tilde{s}_i} Z_u^k\). In particular \(Z_s^{j_1} \subseteq Z_v^{i_1}\) and \(Z_s^{j_2} \subseteq Z_v^{i_2}\), thus \(j_1 \ne j_2\), i.e. \((s,m_1)\) and \((s,m_2)\) are connected to different copies of the link of \(s\). In particular this means that \(m_1 \ne m_2\) and the map is injective.
    
    It remains to prove that this is an isomorphism of simplicial complexes. First, we note that both \(\FX{d_1}Y\) and \(\widetilde{\FX{d_1}X}\) are clique complexes:
    \begin{enumerate}
        \item \(X\) is a clique complex, hence by \pref{claim:cover-of-clique-complex} its cover \(Y\) is also a clique complex. The face complex \(\FX{d_1}Y\) of a clique complex is also a clique complex.
        \item \(\FX{d_1}X\) is a clique complex since it is a face complex of a clique complex. By \pref{claim:cover-of-clique-complex} the cover \(\widetilde{\FX{d_1}X}\) is also a clique complex.
    \end{enumerate}
    Thus it is enough to show that the bijection \(\iota\) is a graph isomorphism between \(\FX{d_1}Y^{\leq 1}\) to \((\widetilde{\FX{d_1}X})^{\leq 1}\).
    
    Let us begin by showing that \(\iota\) is a graph homomorphism. Let \(\set{\tilde{s}_1,\tilde{s}_2} \in \FX{d_1}Y(1)\) be an edge. Let \(\iota(\tilde{s}_1) = (s_1,m_1), \iota(\tilde{s}_2)=(s_2,m_2)\). We need to show that \(\set{(s_1,j_1), (s_2,j_2)} \in \widetilde{\FX{d_1}X}(1)\). Since \(\set{\tilde{s}_1,\tilde{s}_2} \in \FX{d_1}Y(1)\), then in particular they are contained in some triangle \(\set{\tilde{s}_1,\tilde{s}_2,\tilde{s}_3} \in \FX{d_1}Y(2)\)\footnote{\(Y\) is pure so \(\FX{d_1}Y\) is also pure.} and denote by \(\nu(\tilde{s}_3)=s_3\). Let \(Z_{s_1}^{j_1}, Z_{s_2}^{j_2}\) be the copies of \(\FX{d_1}X_{s_1},\FX{d_1}X_{s_2}\) such that \((s_1,j_1), (s_2,j_2)\) are connected to \((s_1,m_1),(s_2,m_2)\) respectively (these copies are uniquely determined by the definition of \(\iota\)). If we show that the same copy of \(s_3\) is in the intersection of \(Z_{s_1}^{j_1},Z_{s_1}^{j_2} \subseteq \widetilde{\FX{d_1}X}\), namely that there is some \((s_3,k) \in Z_{s_1}^{j_1} \cap Z_{s_2}^{j_2}\) this will imply that \(\set{(s_1,j_1), (s_2,j_2)} \in \widetilde{\FX{d_1}X}(1)\) (because this implies that both \((s_1,j_1),(s_2,j_2)\) are in the link of \((s_3,k) \in \widetilde{\FX{d_1}X}\). This link is isomorphic by \(\nu\) to the link of \(s_3 \in \FX{d_1}X\). There we know that \(\set{s_1,s_2}\) is an edge, so in particular \(\set{(s_1,j_1), (s_2,j_2)} \in \widetilde{\FX{d_1}X}(1)\)).
    
    Take two vertices \((v_1,k_1),(v_2,k_2) \in Y\) such that \((v_i,k_i) \in \tilde{s}_i\). Because \(\set{\tilde{s}_1,\tilde{s}_2, \tilde{s}_3}  \in \FX{d_1}Y(2)\) implies that \(\tilde{s}_1 \dunion \tilde{s}_2 \dunion \tilde{s}_3 \in Y\). This means that 
    \begin{enumerate}
        \item \(\set{(v_1,k_1),(v_2,k_2)} \in Y(1)\), and that
        \item \(s_3 \in \FX{d_1}X_{v_1 v_2}\).
    \end{enumerate} 
    By \pref{lem:Y-is-a-cover-of-X} the fact that \(\set{(v_1,k_1),(v_2,k_2)} \in Y(1)\) implies that there is a copy \(Z_{v_1,v_2}^{p} \subseteq Z_{v_1}^{k_1} \cap Z_{v_1}^{k_2}\). The fact that \(s_3 \in \FX{d_1}X_{v_1 v_2}\) implies that there is some \((s_3,k) \in Z_{v_1,v_2}^{p}\) and in particular \((s_3,k) \in Z_{v_1}^{k_1},Z_{v_1}^{k_2}\). We now show that this \((s_3,k)\) is the vertex we are looking for. Let \((s_3,k') \in Z_{s_1}^{j_1}\) be the copy of \(s_3\) in \( Z_{s_1}^{j_1}\). by definition of \(\iota\), it holds that \(Z_{s_1}^{j_1} \subseteq Z_{v_1}^{k_1}\). As \(Z_{v_1}^{k_1}\) is isomorphic by \(\nu\) to \(\FX{d_1}X_{v_1}\) and both \((s_3,k),(s_3,k') \in Z_{v_1}^{k_1}\) are sent to \(s_3\) by \(\nu\), this implies that \((s_3,k)=(s_3,k')\). The same argument applies also to \(Z_{s_2}^{j_2}\). We've proven there is an \((s_3,k) \in Z_{s_1}^{k_1} \cap Z_{s_3}^{k_2}\) and this shows that this is a grpah homomorphism.
    
    To be convinced that this is also an isomorphism, note that both \(\widetilde{\FX{d_1}X}\) and \(\FX{d_1}Y\) are $\ell$-covers of \(\FX{d_1}X\). In particular, the degree of every \(\tilde{s} \in \FX{d_1}Y\) is equal to the degree of \(\iota(\tilde{s}) \in \widetilde{\FX{d_1}X}\) (both are equal to the degree of \(\rho(s)\) in \(\FX{d_1}X\)). Any bijective graph isomorphism that preserves degrees must be an isomorphism. 
\end{proof}

\begin{proof}[Proof of \pref{claim:unique-link-of-s}]
    Let \(j\) be such that \(Z_s^j \subseteq Z_{v_0}^{i_0}\) (there is a unique such \(j\) by the second item of \pref{claim:Z_v-decomposes}). We need to show that \(Z_s^j \subseteq Z_{v_m}^{i_m}\) for all \(m=1,2,...,d_1\). Indeed, fix \(m\) and let \(k\) be such that \(Z_{v_0 v_m}^{k} \subseteq Z_{v_0}^{i_0} \cap Z_{v_m}^{i_m}\) (there is such a \(k\) by \pref{lem:Y-is-a-cover-of-X}). Then by the fact that \(\FX{d_1}X_{s} \subseteq \FX{d_1}X_{v_0 v_m} \subseteq \FX{d_1}X_{v_0}\) and \(\FX{d_1}X_{s} \subseteq \FX{d_1}X_{v_0 v_m} \subseteq \FX{d_1}X_{v_m}\), it holds by the isomorphism that \(\nu\) induces that \(Z_s^j \subseteq Z_{v_0 v_m}^k \subseteq Z_{v_0}^{i_0}\), and that \(Z_s^j \subseteq Z_{v_0 v_m}^k \subseteq Z_{v_m}^{i_m}\). In particular \(Z_s^j \subseteq Z_{v_m}^{i_m}\). This shows existence of \(Z_s^j\). For uniqueness we note that there is a unique \(j\) such that \(Z_s^j \subseteq Z_{v_0}^{i_0}\), and \(\bigcap_{m=0}^{d_1} Z_{v_m}^{i_m} \subseteq Z_{v_0}^{i_0}\) so there must be only a single such \(j\).
\end{proof}

%% file: sections/agreement-on-glob-func.tex
\section[A global function on Y]{A global function on \(Y\)} \label{sec:global-func}
We remind the reader of our convention that \(\pi_{x,y}=\pi_{y,x}^{-1}\).
Let us rephrase \pref{cor:good-cocycle}.
\begin{lemma} \label{lem:good-exact-cover}
    There exists a cocycle \(\psi \in Z^1(X,Sym(\ell))\) such that \(\Prob[s_1,s_2, t=s_1 \dunion s_2]{\psi(s_1, s_2)=\pi_{t,s_2}^{-1} \circ \pi_{t_1,s_1} } = 1-\exp(-\Omega(d_1/k)),\) where \(\pi_{t,s_i}\) are the permutations over the list of functions of \(s_i\) and \(t\) respectively as in \pref{lem:from-1-percent-to-list-agreement}.
\end{lemma}
Let us denote the \(\ell\)-cover that this cocycle induces by \(\rho:Y\to X\) such that we have an isomorphism $\iota:\FX{d_1}Y\to\widetilde{\FX{d_1}X}$. 
Recall that we defined, in Section \ref{sec:globalfunc},
functions \(h_{\tilde{s}}:\tilde{s} \to \Sigma\) by
\begin{equation}
    h_{\tilde{s}}((v,i))=L_s^j(\rho(v,i))=L_s^j(v)
\end{equation}
where \(\iota(\tilde{s})=(s,j)\). Using these functions we defined \(G:Y(0) \to \Sigma\) via \(G(v,i)=plurality \sett{ h_{\tilde{s}}((v,i))}{(v,i) \in \tilde{s}}\), i.e. the most popular assignment of \((v,i)\) from all the \(h_{\tilde{s}}\) where \(\tilde{s} \ni (v,i)\).

The final component we need for proving our main theorem is that most \(h_{\tilde{s}}\) agree with \(G\) on most vertices.
\restateclaim{claim:agreement-with-majority}

For the proof of the lemma, we need the following agreement theorem taken from \cite{DiksteinH2023}.
\begin{theorem}[{\cite[Claim 8.5]{DiksteinH2023}}]\label{thm:agreement-high-acceptance}
    There exists a universal constant \(c>0\) such that the following holds. Let \(Y\) be a \(\frac{1}{d_1^3}\)-one sided high dimensional expander. Let \(DU_n\) be the \((d_1,\frac{3}{2}d_1)\)-non-lazy-up-down walk in \(Y\) as in \pref{def:non-lazy-up-down}. Let \(\mathcal{H}=\set{h_{\tilde{s}}}_{\tilde{s} \in Y(d_1)}\) be an ensemble of functions such that
    \[\Prob[\tilde{s}_1,\tilde{s}_2 \sim DU_n]{h_{\tilde{s}_1}|_{\tilde{s}_1 \cap \tilde{s}_2} \overset{1-\eta}{\approx} h_{\tilde{s}_2}|_{\tilde{s}_1 \cap \tilde{s}_2} } \geq 1-\gamma,\]
    then 
    \[\Prob[\tilde{s} \in Y(d_1)]{G|_{\tilde{s}} \overset{1-8(\eta + \gamma)}{\approx} h_{\tilde{s}}} \geq 1-\frac{1}{d_1^c}-\gamma^c.\]
\end{theorem}

\begin{proof}[Proof of \pref{claim:agreement-with-majority}]
\begin{figure}
    \centering
    \includegraphics[scale=0.6]{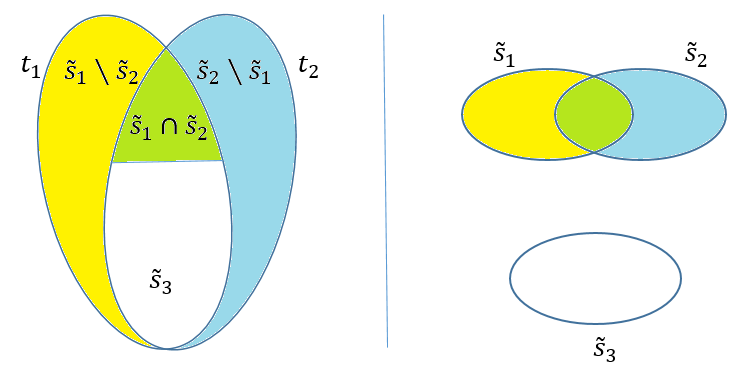}
    \caption{The sets sampled by \(P\). Note that \(\tilde{s_1}\) is the yellow and green parts, \(\tilde{s_1}\) is the blue and green parts, and the \(\tilde{t_i}\)'s are the \(s_i\)'s together with \(\tilde{s_3}\) (the white part which is disjoint from \(\tilde{s_1}\) and \(\tilde{s_2}\)).}
    \label{fig:dist-P}
\end{figure}
    We want to use \pref{thm:agreement-high-acceptance}, and thus we need to show that
    \[\Prob[\tilde{s}_1,\tilde{s}_2 \sim DU_n]{h_{\tilde{s}_1}|_{\tilde{s}_1 \cap \tilde{s}_2} \overset{1-O(\eta)}{\approx} h_{\tilde{s}_2}|_{\tilde{s}_1 \cap \tilde{s}_2} } \geq 1-\exp(-\Omega(d_1/k)),\]

    Let us define the distribution \((\tilde{s}_1,\tilde{s}_2,\tilde{s}_3) \sim P\) where:
    \begin{enumerate}
        \item \(\tilde{s_1}\) and \(\tilde{s_2}\) are sampled via the non-lazy-up-down-walk. That is, \(\tilde{s_1} \cup \tilde{s}_2 \in X(\frac{3}{2}d_1)\).
        \item \(\tilde{s}_3 \dunion (\tilde{s_1} \cup \tilde{s}_2) \in X\), i.e. \(\tilde{s}_3\) is sampled via a swap walk step from \((\tilde{s_1} \cup \tilde{s}_2)\).
    \end{enumerate}
    Let us denote by \(\tilde{t_i} = \tilde{s_i} \dunion \tilde{s_3}\) for \(i=1,2\), and let us denote \(\tilde{r} = \tilde{s}_1 \cap \tilde{s}_2\). See \pref{fig:dist-P} for an illustration of the sets sampled. We also use the convention that \(\rho(\tilde{s}_i) = s_i, \rho(\tilde{t_i})=t_i, \rho(\tilde{r})=r\) (hopefully, it will also be clear from context which faces belong to \(Y\) and which belong to \(X\)).
    
    Let \(j_i\) be the index that has that \(L_{s_i}^{j_i}\circ \rho = h_{\tilde{s_i}}\) for \(i=1,2,3\).  and let \(j_i' = \pi_{t_i,s_i}(j_i)\). We define the following two ``bad'' events
    \begin{enumerate}
        \item \(E = \set{L_{s_i}^{j_i} \overset{1-\eta}{\napprox} L_{t_i}^{j_i'}|_{s_i} \text{ or } L_{s_3}^{j_3} \overset{1-\eta}{\napprox} L_{t_i}^{j_i'}|_{s_3} \text{ for } i=1 \text{ or } i=2}\).
        \item \(F = \set{L_{t_1}^{j_1'}|_{s_3} \overset{1-2\eta}{\approx} L_{t_2}^{j_2'}|_{s_3} \ve L_{t_1}^{j_1'}|_{r} \overset{1-9\eta}{\napprox} L_{t_2}^{j_2'}|_{r}}\).
    \end{enumerate}
    Let us spell out these events. The event \(E\) considers the lists of \(s_1,s_2,s_3\) and \(t_1,t_2\). The first statement in the event, i.e. that \(L_{s_i}^{j_i} \overset{1-\eta}{\napprox} L_{t_i}^{j_i'}|_{s_i}\), roughly states that at least one of the \(i=1,2\), the lists of \(s_i\) and \(t_i\) are not compatible (at a specific index \(j_i\)). The second statement, namely that \(L_{s_3}^{j_3} \overset{1-\eta}{\napprox} L_{t_i}^{j_i'}|_{s_3}\) is the event that for at least one of the \(i=1,2\), the entry of \(j_3\) in \(s_3\)'s list, is not compatible with \(j_i'=\pi_{t,s_i}(j_i)\). We will see that with very high probability \(j_3 = \pi_{s_3,t}(j_i')\) so this part also essentially says that the one of the lists of \(t_i\) and \(s_3\) are not compatible with one another.
    
    The event \(F\) says that ``even though'' \(L_{t_1}^{j_1'}\) and \( L_{t_2}^{j_2'}\) agree on most vertices of \(s_3\), they do not agree on most vertices of \(r\), the other part of their intersection.

    We claim that both events happen with only small probability (and prove this later on in this section):
    \begin{claim} \label{claim:E-is-small}~ 
        \(\prob{E} \leq \exp(-\Omega(d_1/k))\).
    \end{claim}
    \begin{claim} \label{claim:F-is-small}~ 
        \(\prob{F} \leq \exp(-\Omega(d_1))\).
    \end{claim}
    Let us explain why
    \begin{equation}
        \prob{h_{\tilde{s_1}} \overset{1-13\eta}{\napprox} h_{\tilde{s_2}}} \leq \prob{E} + \prob{F} = \exp(-\Omega(d_1/k)),
    \end{equation}
    or in other words, why when \(\neg E \ve \neg F\) occur, then \(h_{\tilde{s_1}} \overset{1-13\eta}{\approx} h_{\tilde{s_2}}\).
    
    Indeed, if \(E\) doesn't occur then for \(i=1,2\), \(L_{s_i}^{j_i} \overset{1-\eta}{\approx} L_{t_i}^{j_i'}|_{s_i}\),  and thus it holds that \(L_{s_i}^{j_i}|_r \overset{1-2\eta}{\approx} L_{t_i}^{j_i'}|_r\) (since \(r\) is at \(\frac{1}{2}\) the size of \(\tilde{s_i}\) for both \(i=1,2\)). For similar reasons, \(L_{t_1}^{j_1'}|_{s_3} \overset{1-\eta}{\approx} L_{s_3}^{j_3} \overset{1-\eta}{\approx} L_{t_2}^{j_2'}|_{s_3}\) and hence \(L_{t_1}^{j_1'}|_{s_3} \overset{1-2\eta}{\approx} L_{t_2}^{j_2'}|_{s_3}\). By \(\neg F\) this implies that \(L_{t_1}^{j_1'} |_r \overset{1-9\eta}{\approx} L_{t_2}^{j_2'}|_r\). In total, it holds that
    \[h_{\tilde{s_1}}|_{\tilde{r}} = (L_{s_1}^{j_1} \circ \rho)|_{\tilde{r}} \overset{1-2\eta}{\approx} (L_{t_1}^{j_1'} \circ \rho)|_{\tilde{r}} \overset{1-9\eta}{\approx} (L_{t_2}^{j_2'} \circ \rho)|_{\tilde{r}} \overset{1-2\eta}{\approx} L_{s_1}^{j_2'} \circ \rho|_{\tilde{r}} = h_{\tilde{s_2}}|_{\tilde{r}},\]
    so \(h_{\tilde{s_1}}|_{\tilde{r}} \overset{1-13\eta}{\approx} h_{\tilde{s_2}}|_{\tilde{r}}\). The claim follows from \pref{thm:agreement-high-acceptance}.
\end{proof}

\begin{proof}[Proof of \pref{claim:E-is-small}]
Let us begin with bounding the probability that \(L_{s_i}^{j_i} \overset{1-\eta}{\napprox} L_{t_i}^{j_i'}|_{s_i}\) for some \(i=1,2\). We note that \(s_i,t_i\) are chosen according to the usual joint distribution of \((s\subset t)\) in \(X\). Thus the probability that \(\prob{L_{s_i}^{j_i} \overset{1-\eta}{\napprox} L_{t_i}^{j_i'}}\) for both \(i=1,2\) is at most
\[1-\Prob[t\in X(\frac{3}{2}d_1),s\subset t, s\in X(d_1)]{\forall j=1,2,...,\ell, \; L_s^j \overset{1-\eta}{\approx} L_t^{\pi_{t,s}(j)}|_s} = \exp(-\Omega(d_1/k))\]
by \pref{lem:from-1-percent-to-list-agreement}. When this occurs then \(L_{s_i}^{x} \overset{1-\eta}{\approx} L_{t_i}^{\pi_{t_i,s_i}(x)}|_{s_i}\) for all indexes \(x\), and in particular for \(j_i\).

We move towards bounding the probability that \(L_{s_3}^{j_3} \overset{1-\eta}{\napprox} L_{t_i}^{j_i'}|_{s_3}\) for some \(i=1,2\). The pair \(s_3,t_i\) is also distributed according to the joint distribution of \((s\subset t)\) in \(X\). Hence it holds that \(\prob{L_{s_3}^{j_3} \overset{1-\eta}{\napprox} L_{t_i}^{\pi_{t_i,s_3}(j_3)}} = \exp(-\Omega(d_1/k))\).

Moreover, by \pref{lem:good-exact-cover}, with probability \(1- \exp(-\Omega(d_1/k))\) it holds that \(\psi(s_i, s_3)=\pi_{t_i,\tilde{s_3}}^{-1} \circ \pi_{t_i,\tilde{s_i}}\) (where \(\psi\) is the cocycle that defined the cover \(\rho\)). Recall that \(\iota(\tilde{s_i})=(s_i,j_i)\) where \(\iota:\FX{d_1}Y \to \widetilde{\FX{d_1}X}\) and \(\tilde{s_i},\tilde{s_3} \in \FX{d_1}Y(0)\) are neighbors in \(\FX{d_1}Y\) which implies that for \(i=1,2\), \((s_i,j_i) \sim (s_3,\psi(s_i,s_3) . j_i)\) are neighbors in \(\widetilde{\FX{d_1}X}\). Thus with probability \(1- \exp(-\Omega(d_1/k))\), 
\[j_3 = \pi_{t_i,s_3}^{-1} \circ \pi_{t_i,s_i}(j_i)\]
for \(i=1,2\). When this occurs, then
\[\pi_{t_i,s_3}(j_3) = \pi_{t_i,s_i}(j_i)=j_i',\]
and as above \(\prob{L_{s_3}^{j_3} \overset{1-\eta}{\napprox} L_{t_i}^{j_i'}|_{s_3}} \leq \exp(-\Omega(d_1/k))\).

\end{proof}

The proof of \pref{claim:F-is-small} is just by basic probabalistic arguments, and doesn't rely on any of the machinery we developed.
\begin{proof}[Proof of \pref{claim:F-is-small}]
    Fix any \(t_1,t_2\), and denote \(t' = t_1 \cap t_2\). Fix any \(f=L_{t_1}^{j_1'}|_{t'}, g=L_{t_2}^{j_2'}|_{t'}:t' \to \Sigma\). The choice of \(r, s_3 \subseteq t'\) is uniform among all pairs such that \(r \dunion s_3 = t'\). We will show that
    \[\Prob[\tilde{s_3},r]{f|_{s_3} \overset{1-2\eta}{\approx} g|_{s_3} \ve f|_{r} \overset{1-9\eta}{\napprox} g|_{r}} \leq \exp(-\Omega(\eta^2 |s|))\]
    which proves the claim since \(|s| =\Omega(d_1)\) and \(\eta\) is a constant.
    
    Indeed, if \(f|_{r} \overset{1-9\eta}{\napprox} g|_{r}\), then \(f \overset{1-3\eta}{\napprox} g\) (since \(r\) is a third of the size of \(t'\)). In this case, the probability of choosing \(s_3 \subseteq t'\) of size greater or equal to half of \(t'\) such that \(L_{t_1}^{j_1'}|_{s_3} \overset{1-2\eta}{\approx} L_{t_2}^{j_2'}|_{s_3}\) is at most \(\exp(-\Omega(\eta^2 d_1))\) by a standard Chernoff bound.
\end{proof}

%% file: sections/counterexample.tex
\section{Covers and Counterexamples}
\subsection{Proof of \pref{lem:cex}} \label{sec:proof-of-counterexample}
In this section we elaborate on the counterexample we presented in the introduction. The setup is the following. \(X\) is a simplicial complex that admits a connected \(2\)-cover \(\rho:Y \to X\). We also assume that \(X\) is a \(\lambda=2^{-7k-1}\) two-sided high dimensional expander. Such complexes were constructed in \cite{LubotzkySV2005b} (in the referenced paper they only give a bound on one sided expansion, but see also \cite{DinurK2017} for obtaining two-sided expansion). Let \((r,r') \sim \mathcal{D}\) be any distribution over \(k\)-faces of \(X\) such that the marginal of \(\mathcal{D}\) is the distribution over \(k\)-faces in \(X\). 

\restatelemma{lem:cex}

\begin{proof}[Proof of \pref{lem:cex}]
Without loss of generality we denote \(Y(0)=\sett{(v,0),(v,1)}{v \in X(0)}\), the covering map is \(\rho(v,i)=v\).

By \pref{claim:inverse-image-of-l-cover} every face \(r \in X(k)\) has exactly two preimages under $\rho$, which we denote \(\tilde{r}_1,\tilde{r}_2 \in Y(k)\).

We will sample a random ensemble \(\mathcal{F}\), and show that at least one random ensemble has agreement at least \(\frac{1}{2}\) by showing that \(\Ex[\mathcal{F}]{\agr_{\mathcal{D}}(\mathcal{F})} \geq \frac{1}{2}\).

The randomized ensemble is constructed as follows. For every \(r \in X(k)\) we select one of the two preimage \(\tilde{r} \in Y(k)\) of \(r\), independently and uniformly at random. Then for every \(v \in r\), 
\[f_r(v) = \begin{cases} 0 & (v,0) \in \tilde{r} \\ 1 & (v,1) \in \tilde{r} \end{cases}.\]
In the introduction we stated that \(f_r\) comes from a function \(H|_{\tilde{r}}\). This function is \(H:Y(0) \to \set{0,1}\), \(H(v,i)=i\). 
Let us verify that \(\Ex[\mathcal{F}]{\agr_{\mathcal{D}}(\mathcal{F})} \geq \frac{1}{2}\). The expectation is over the choice of preimages \(\tilde{r}\) for every \(r \in X(k)\).
\[\Ex[\mathcal{F}]{\agr_{\mathcal{D}}(\mathcal{F})} = \Ex[\mathcal{F}]{\Ex[r,r' \sim D]{\one_{f_r = f_{r'}}}} = \Ex[r,r' \sim D]{\Ex[\mathcal{F}]{\one_{f_r = f_{r'}}}}.\]
Hence it is enough to show that for every \(r,r' \sim D\), \(\Ex[\mathcal{F}]{\one_{f_r = f_{r'}}} \geq \frac{1}{2}\), or equivalently that \(\Prob[\mathcal{F}]{f_r = f_{r'}} \geq \frac{1}{2}\).

Fix \(r,r' \sim D\) an edge in the agreement distribution. If \(r,r'\) have an empty intersection or \(r=r'\) this is trivial, so let us assume this is not the case.
Let \(\tilde{r}, \tilde{r}\) be the preimages of \(r,r'\) randomly chosen, respectively. Let \(v \in r \cap r'\) and suppose without loss of generality that \((v,0) \in \tilde{r}\) be the preimage of \(v\) inside the preimage of \(r\). As \(\tilde{r}\) is chosen independently, with probability \(\frac{1}{2}\) it holds that \((v,0) \in \tilde{r}\). In this case, because \(\rho|_{X_{(v,0) }}:X_{(v,0) }(0) \to X_{v}(0)\) is an isomorphism it holds that for every \(u \in r \cap r'\), there is a preimage \((u,i)\) such that \((u,i) \in \tilde{r} \cap \tilde{r}'\). In particular \(f_r(u)=f_{r'}(u)\) by definition. Thus indeed \(\Prob[\mathcal{F}]{f_r = f_{r'}} \geq \frac{1}{2}\).

Thus we take \(\mathcal{F}\) be any such assignment where \(\agr_{\mathcal{D}}(\mathcal{F}) \geq \frac{1}{2}\) (which exists since the expectation over the agreement is at least \(\frac{1}{2}\)). The first item in the lemma holds for this ensemble.

\medskip

We conclude by showing that for \emph{any} ensemble constructed as above, and any \(G:X(0) \to \set{0,1}\) it holds that 
\[\Prob[r \in X(k)]{f_r \overset{1-\zeta}{\approx} G|_r} =\exp(-\Omega_\zeta(k)).\]

Fix a global function \(G:X(0) \to \set{0,1}\). Recall that \(H:Y \to \set{0,1}\) is \(H(v,i)=i\). Let \(\tilde{G}:Y(0) \to \set{0,1}\) be \(\tilde{G}((v,i))=G(v)\). Note that if \(G|_r \overset{1-\zeta}{\approx} f_r\) and \(f_r\) answers according to the preimage \(\tilde{r}\), then \(H|_{\tilde{r}} \overset{1-\zeta}{\approx} \tilde{G}|_{\tilde{r}}\). This is because if \(G(v)=f_r(v)=j\) and \(f_r\) answers according to \(\tilde{r}\), this means that \((v,j) \in \tilde{r}\), and hence \(\tilde{G}(v,j)=G(v)=f_r(v)=H(v,j)\). We conclude that
\begin{equation} \label{eq:counterexample-equivalence-inequality}
    \Prob[r \in X(k)]{G|_r \overset{1-\zeta}{\approx} f_r} \leq 2\Prob[\tilde{r} \in Y(k)]{\tilde{G}|_{\tilde{r}} \overset{1-\zeta}{\approx} H|_{\tilde{r}}},
\end{equation}
since the choice of \(\tilde{r} \in Y(k)\) is done by first chosing \(r \in X(k)\) and then chosing one preimage uniformly at random. Thus we will argue that \(\Prob[\tilde{r} \in Y(k)]{\tilde{G}|_{\tilde{r}} \overset{1-\zeta}{\approx} H|_{\tilde{r}}} = \exp(-\Omega_\zeta(k))\).

We observe that \(\dist(H,\tilde{G})=\frac{1}{2}\) because for every \(v \in X(0)\), \(H,\tilde{G}\) agree on exactly one of \((v,0),(v,1)\). By \pref{claim:expansion-of-cover} it holds that \(Y\) is also a \(2^{-7k}\)-two sided spectral expander. 

Let \(A = \sett{(v,i) \in Y(0)}{H(v,i) \ne \tilde{G}(v,i)}\). Then \(\prob{A} = \dist(H,\tilde{G})=\frac{1}{2}\), and a face \(\tilde{r}\) is such that \(\tilde{G}|_{\tilde{r}} \overset{1-\zeta}{\approx} H|_{\tilde{r}}\) if and only if \(\cProb{(v,i)\in Y(0)}{A}{(v,i) \in \tilde{r}} < \zeta\).

Thus by \pref{thm:sampling-in-HDXs}, it holds that 
\[\Prob[\tilde{r} \in Y(k)]{\tilde{G}|_{\tilde{r}} \overset{1-\zeta}{\approx} H|_{\tilde{s}} } = \exp(-\Omega(\zeta^2 k))\]
and by \eqref{eq:counterexample-equivalence-inequality}, 
\[\Prob[r \in X(k)]{G|_r \overset{1-\zeta}{\approx} g_r} = \exp(-\Omega_\zeta( k )).\]    
\end{proof}

\begin{remark}
A similar argument can generalized for complexes \(X\) with connected \(\ell\)-covers and an alphabet of \(\Sigma = [\ell]\). We can also consider test distributions that samples \(q\) \(k\)-sets instead of \(2\), provided that every \(\set{r_i} \subseteq \supp \mathcal{D}\) are contained in a simply connected sub complex of \(X\).

In this case we can show that there exists an ensemble with agreement at least \(\frac{1}{\ell^{q-1}}\), but 
\[\Prob[r \in X(k)]{G|_r \overset{1-\zeta}{\approx} g_r} = \exp(-\Omega_\zeta( k ))\]  
will hold for any \(G:X(0) \to \Sigma\) and any \(\zeta < \frac{\ell-1}{\ell}\).
\end{remark}

\subsection[Gap amplification fails below 1/2]{Gap amplification fails below $1/2$} \label{sec:bog}
In \cite{Bogdanov-comment05}, Bogdanov addresses the question of whether the soundness of a PCP can be reduced to any $\epsilon>0$ by a step of gap amplification a la \cite{Dinur07}. The starting point is a constraint graph $\calG = (V,E,\set{\Sigma_v},\set{\pi_{uv}})$ where $V$ is a set of vertices, $E$ a set of edges, $\Sigma_v$ a finite set of labels to be assigned to $v$ and $\pi_{uv} \subset \Sigma_u\times\Sigma_v$ a constraint specifying which pairs of labels are considered satisfying.
The value of $\calG$ is the maximal fraction of satisfiable constraints, namely 
\[ val(\calG)=\max_{H:V\to\bits} \Pr_{uv\in E}[(H(u),H(v))\in \pi_{uv}].\] 
Our initial assumption is that it is NP-hard to decide if $val(\calG)=1$ or $val(\calG)<1-\eps_0$.

The gap amplification, or ``powering'' step (with parameter $\ell\in\mathbb{N}$), constructs a family $X$ of subsets of $V$ corresponding to balls in $\calG$ as follows. For each $v\in V$, we define $s_v = \sett{u\in V}{\dist(u,v) \leq \ell}$. We let $X = \sett{s_v}{v\in V}$. Furthermore, for a set $s_v$ we will only allow local functions $f_{s_v}:s_v\to\bits$ that satisfy all $\calG$ constraints inside $s_v$. Formally, let
\begin{equation}\label{eq:valid-assignments}
    \Sigma_v = \sett{f_v:s_v\to\bits}{f_v\hbox{ satisfies all $\calG$ constraints in }s_v}.
\end{equation}
A valid ensemble of local functions on $X$ is $\sett{f_{s_v}:s_v\to\bits}{f_{s_v}\in \Sigma_v}$. An agreement test for $X$ specifies a distribution over pairs of sets in $X$, but we do not specify the precise distribution because the counterexample works for any distribution. 
\begin{itemize}
    \item The completeness of this transformation is that whenever there is an assignment satisfying all of the constraints of $\calG$ then there is a valid ensemble of local functions $\set{f_s}$ with $\agr(\set{f_s})=1$. 
    \item The soundness of this transformation is that whenever there is a valid ensemble $\set{f_s}$ with $\agr(\set{f_s})>\epsilon$ there must be some global $H:V\to\bits$ that satisfies more than $1-\eps_0$ of the constraints of $\calG$, namely, $val(\calG)>1-\eps_0$ which means, by our assumption, that $val(\calG)=1$. 
\end{itemize}
The ``soundness'' of the gap amplification transformation is the value of $\eps$ for which we can guarantee that 
\[ \agr(\set{f_s})>\eps \quad\Longrightarrow \quad val(\calG)>1-\eps_0.\]

Contrapositively, a counterexample to gap amplification is an ensemble that has high agreement yet $val(\calG)<1-\eps_0$. 

In Bogdanov's example, the graph $\calG$ is chosen to be a high-girth expander graph, the alphabet is $\Sigma_v=\bits$ for all $v\in V$, and all edges carry inequality constraints, namely $\pi_{uv}=\set{(0,1),(1,0)}$ for all $uv\in E$. The expansion asserts that $val(\calG) \approx 1/2$ since every cut in the graph violates about half of the edges or more. 

Looking at valid ensembles, we observe that the set $\Sigma_v$ of possible local functions for a given $s_v$ consists of only two assignments, per \eqref{eq:valid-assignments}. One where $v\in s_v$ is given $0$ and one where $v\in s_v$ is given $1$. The unique constraints along the edges propagate the values uniquely from $v$ to the rest of $s_v$. Bogdanov chooses an ensemble at random by choosing, for each $s_v$ one of the two available local functions independently at random. Clearly, for each edge separately,  there is probability $1/2$ of choosing a matching pair, so the overall local agreement is $\approx 1/2$, whereas, with high probability there is no global assignment agreeing with even an $\epsilon$ faction of labels, assuming $\calG$ is a good expander.\\

Now let us recast this in terms of covers (and then generalize).
We first construct the double cover $\calG'$ of $\calG$. Let $\calG'$ be a bipartite graph with vertices $V'=V\times\bits$ and edges $E' = \sett{\set{(u,0),(v,1)} }{uv\in E}$ (and still putting inequality constraints on all edges). 

Observe that this is a $2$-cover a la \pref{def:cover}. Moreover, since $\calG'$ is bipartite, we actually have $val(\calG)=1$. This is because all constraints are satisfied by the global assignment $H:V'\to\bits$ defined by $H(v,b)=b$.  

Now we follow the recipe of the proof of \pref{lem:cex} (see Figure). 
\begin{center}
\begin{tikzcd}[arrows=rightarrow]
Y  \ar[d,"\rho"] \quad & \quad \calG'\ar[l,"powering"]\ar[d,"\rho"] \\
		 X \quad & \quad\calG\ar[l,"powering"]  
\end{tikzcd}
\end{center}
Assuming that the girth of $\calG$ is larger than $\ell$, the family $Y = \sett{\tilde s_{(v,b)}}{(v,b)\in V'}$ is a $2$-cover of the family $X$ in that there is a $2$-to-$1$ map $\rho:V'\to V$ such that the image of every $\tilde s\in Y$ is some $s\in X$, and the preimage of each $s\in X$ is two disjoint sets $\tilde s_1,\tilde s_2\in Y$\footnote{This is also formally a cover, only that \(X,Y\) are cell complexes, not simplicial complexes.}. Specifically $\rho(\tilde s_{(v,b)}) = s_v$, and $\rho^{-1}(s_v) = \set{\tilde s_{(v,0)},\tilde s_{(v,1)}}$.
From the global assignment $H:V'\to\bits$ we can construct a perfectly agreeing ensemble (on $Y$) by setting $ f_{\tilde s_{(v,b)}} = H|_{\tilde s_{(v,b)}}$.

We choose, as in the proof of \pref{lem:cex}, for each $s_v\in X$ an assignment
let \[f_{s_v}  =
\begin{cases}
			f_{\tilde s_{(v,0)}}, & \text{with probability 1/2}\\
                f_{\tilde s_{(v,1)}}, & \text{with probability 1/2}
		 \end{cases}
\;.   \]
Finally, observe that this is exactly the same assignment as in the direct description of Bogdanov's counterexample.
\subsubsection*{Discussion}
One possible conclusion of this section is that the failure of gap amplification to go below $1/2$ is due to the graph-based nature of the family $X$ which allows existence of connected covers with perfect assignments. For a better gap amplification we need to start with PCP instances $\calG$ that are free of such covers. 

An additional observation is that Bogdanov's example generalizes to any unique games instance $\calG$, where the size of the alphabet $t=\card \Sigma$ dictates the size of the $t$-to-$1$ cover. Even for $\calG$ whose constraints are not unique, one can easily sparsify them into unique constraints and continue the construction as before. So without revising the powering construction itself, no choice of $\calG$ can help. 